%% file: main.tex
\documentclass[twoside,leqno,twocolumn]{article}
\usepackage[letterpaper]{geometry}

\usepackage{ltexpprt}

\input{sty_package.tex}

\begin{document}
\input{sty.tex}
\input{macro.tex}
\newcommand\relatedversion{}
\renewcommand\relatedversion{\thanks{The full version of the paper can be accessed at \protect\url{http://arxiv.org/abs/2410.17226}}} % Replace URL with link to full paper or comment out this line

\fancyhead{}

% \title{Parallel Two Hop Distance Oracle for Exact Shortest-Path Distance Queries on Large Networks}

\iffullversion{\title{Parallel Cluster-BFS and Applications to Shortest Paths}}
\ifconference{\title{Parallel Cluster-BFS and Applications to Shortest Paths\relatedversion}}

% \hide{
% \author{
%   Letong Wang\\
%   UC Riverside\\
%   lwang323@ucr.edu
%   \and      
%   Guy Blelloch\\
%   Carnegie Mellon University\\
%   guyb@cs.cmu.edu
%   \and 
%   Yan Gu\\
%   UC Riverside\\
%   ygu@cs.ucr.edu
%   \and
%   Yihan Sun\\
%   UC Riverside\\
%   yihans@cs.ucr.edu
% }
\author{Letong Wang\thanks{University of California, Riverside.}
\and Guy Blelloch\thanks{Carnegie Mellon University.}
\and Yan Gu \footnotemark[1]
\and Yihan Sun \footnotemark[1]
}

% }s
\ifconference{
\fancyfoot[R]{\scriptsize{Copyright \textcopyright\ 2025 by SIAM\\
Unauthorized reproduction of this article is prohibited}}
}
\date{}
\maketitle
\input{abstract.tex}
\hide{
\begin{CCSXML}
<ccs2012>
 <concept>
  <concept_id>00000000.0000000.0000000</concept_id>
  <concept_desc>Do Not Use This Code, Generate the Correct Terms for Your Paper</concept_desc>
  <concept_significance>500</concept_significance>
 </concept>
 <concept>
  <concept_id>00000000.00000000.00000000</concept_id>
  <concept_desc>Do Not Use This Code, Generate the Correct Terms for Your Paper</concept_desc>
  <concept_significance>300</concept_significance>
 </concept>
 <concept>
  <concept_id>00000000.00000000.00000000</concept_id>
  <concept_desc>Do Not Use This Code, Generate the Correct Terms for Your Paper</concept_desc>
  <concept_significance>100</concept_significance>
 </concept>
 <concept>
  <concept_id>00000000.00000000.00000000</concept_id>
  <concept_desc>Do Not Use This Code, Generate the Correct Terms for Your Paper</concept_desc>
  <concept_significance>100</concept_significance>
 </concept>
</ccs2012>
\end{CCSXML}

\ccsdesc[500]{Do Not Use This Code~Generate the Correct Terms for Your Paper}
\ccsdesc[300]{Do Not Use This Code~Generate the Correct Terms for Your Paper}
\ccsdesc{Do Not Use This Code~Generate the Correct Terms for Your Paper}
\ccsdesc[100]{Do Not Use This Code~Generate the Correct Terms for Your Paper}
}

\hide{
\keywords{Do, Not, Us, This, Code, Put, the, Correct, Terms, for,
  Your, Paper}
}

\hide{
\received{20 February 2007}
\received[revised]{12 March 2009}
\received[accepted]{5 June 2009}
}

%%%%%%%%%%%%%%%%%%%%%%%%%%%%%%%%%%%%%%%%%%%%
%% Display spacing
%%%%%%%%%%%%%%%%%%%%%%%%%%%%%%%%%%%%%%%%%%%%
% Put this after \begin{document}
\setlength\abovedisplayskip{0pt}
\setlength\belowdisplayskip{0pt}
\setlength\abovedisplayshortskip{0pt}
\setlength\belowdisplayshortskip{0pt}

% \newpage
% ~
% \newpage
\setcounter{page}{1}
\pagenumbering{arabic}

% Our code and data are available at \url{https://doi.org/10.5281/zenodo.13905461} and \url{https://doi.org/10.5281/zenodo.13909778}.
We released our code~\cite{wang_2024_13905461} and graph instances used in the experiments~\cite{wang_2024_13909778}.
\input{intro.tex}
\input{prelim.tex}

\input{algorithm.tex}

\input{application.tex}

\input{experiment.tex}
\input{related.tex}
\input{conclusion.tex}

\input{acknowledge.tex}

\clearpage
{
\small
\bibliographystyle{ACM-Reference-Format}
\bibliography{bib/strings,bib/main}
}

\appendix
\iffullversion{
  \input{appendix.tex}
}

\end{document}
\endinput
%%
%% End of file `sample-sigconf.tex'.

%% file: sty_package.tex
\usepackage{etoolbox}
\usepackage{graphicx}  % pictures and figures
\usepackage{lipsum}  % random paragraphs
\usepackage{xspace}
\usepackage{textcomp}
\usepackage{comment} % \begin{comment} ... \end{comment}
\usepackage{verbatim}
\usepackage[usenames,dvipsnames,svgnames,table,x11names]{xcolor}
\usepackage[numbers,sort,compress]{natbib}
\usepackage[colorlinks,
citecolor=Sepia,
linkcolor=Blue,
pagebackref=true
]{hyperref}%         % for online version
\usepackage{latexsym,amsmath,amsfonts,amssymb,stmaryrd,mathtools}
\usepackage{pifont}
\usepackage[shortlabels]{enumitem}
\usepackage{float}
\usepackage[labelfont=bf,font={small},aboveskip=0em, belowskip=0em]{caption}
\usepackage[labelfont=bf,list=true,skip=0em]{subcaption}
\usepackage[rightcaption]{sidecap}
\usepackage{wrapfig}
\usepackage{array}% for extended column definitions
\usepackage{rotating}

\usepackage{booktabs} % provides toprule, bottomrule, midrule, cmidrule, etc.
\usepackage{multirow}
\usepackage{longtable} % provides long table
\usepackage{supertabular} % similar to long table, allowing tables to take more than one page
\usepackage{colortbl}
\usepackage{bigstrut}
\usepackage{titlesec}

\usepackage[ruled,lined,linesnumbered,noend]{algorithm2e}
\usepackage{nameref,cleveref}
\usepackage{scalerel} % stretch a formula
\usepackage{mdframed}
\usepackage{framed}
\usepackage{listings}
\usepackage{tikz} % draw geometric objects
\usepackage{titlecaps} 
\usepackage{xparse} % For advanced command definitions
\usepackage{fancyhdr}
\usepackage{etoolbox}
\makeatletter
\patchcmd{\@maketitle}
  {\if@twocolumn\ifnum \c@page>1 \@thanks \fi\else \ifnum \c@page>0 \@thanks \fi\fi}
  {\if@twocolumn\ifnum \c@page>1 \@thanks \fi\else \ifnum \c@page>1 \@thanks \fi\fi}
  {}{}
\makeatother

%% file: sty.tex
% \settopmatter{printfolios=true,printccs=false,printacmref=false}

%%%%%%%%%%%%%%%%%%%%%%%%%%%%%%%%%%%%%%%%%%%%
%% macro for conf and full versions
%%%%%%%%%%%%%%%%%%%%%%%%%%%%%%%%%%%%%%%%%%%%
% Define toggles
\newtoggle{conference}
\newtoggle{fullversion}
% Uncomment the appropriate line to set the toggle
% \toggletrue{conference}
\toggletrue{fullversion}

% New command definitions using toggles
\newcommand{\ifconference}[1]{\iftoggle{conference}{#1}{}}
\newcommand{\iffullversion}[1]{\iftoggle{fullversion}{#1}{}}

%%%%%%%%%%%%%%%%%%%%%%%%%%%%%%%%%%%%%%%%%%%%
%% Various Useful Packages
%%%%%%%%%%%%%%%%%%%%%%%%%%%%%%%%%%%%%%%%%%%%
% \usepackage{graphicx}  % pictures and figures
% \usepackage{lipsum}  % random paragraphs
\newcommand{\hide}[1]{} % hide
\newcommand{\mtext}[1]{{\mbox{{#1}}}} % text in math mode
\newcommand{\mtextit}[1]{{\mbox{\emph{#1}}}} % text in math mode in it
\newcommand{\mtextsc}[1]{{\mbox{\sc{#1}}}} % text in math mode in sc
\newcommand{\smtext}[1]{{\mbox{\scriptsize{#1}}}} % small text in math mode
\newcommand{\smtextit}[1]{{\mbox{\scriptsize\emph{#1}}}} % small text in math mode in it
\newcommand{\algname}[1]{{\textsc{#1}}} % algorithm name format
\newcommand{\defn}[1]{\emph{\textbf{#1}}} % definition style
\newcommand{\emp}[1]{\textbf{#1}} % highlight 
\newcommand{\fname}[1]{\textsc{#1}} % function name format
\newcommand{\vname}[1]{\mathit{#1}} % variable name format
\newcommand{\textcode}[1]{{\texttt{#1}}} % code format in text
\newcommand{\mathfunc}[1]{\mathit{#1}}
\newcommand{\nodecircle}[1]{{\textcircled{\footnotesize{#1}}}} % text in a circle

\newcommand{\R}{\mathbb{R}}
\newcommand{\N}{\mathbb{N}}
\newcommand{\Z}{\mathbb{Z}}
%    Some letters:
\newcommand{\mathd}{\mathcal{D}\xspace}
\newcommand{\mathq}{\mathcal{Q}\xspace}
\newcommand{\mathc}{\mathcal{C}\xspace}

\newcommand{\whp}[1]{\emph{whp}}

\newcommand{\true}{\emph{true}}
\newcommand{\false}{\emph{false}}
\newcommand{\True}{\textsc{True}\xspace}
\newcommand{\False}{\textsc{False}\xspace}
\newcommand{\zero}{\textbf{zero}}
\newcommand{\one}{\textbf{one}}

\newcommand{\cmark}{\ding{51}} % checkmark
\newcommand{\xmark}{\ding{55}} % crossmark

% Binary Forking / other cost models:
\newcommand{\modelop}[1]{\texttt{#1}}
\newcommand{\forkins}{\modelop{fork}}
\newcommand{\insend}{\modelop{end}}
\newcommand{\allocateins}{\modelop{allocate}}
\newcommand{\freeins}{\modelop{free}}
\newcommand{\thread}{thread}
\newcommand{\bfmodel}{binary-forking model}
\newcommand{\MPC}[0]{\ensuremath{\mathsf{MPC}}}
\newcommand{\spanterm}{span}
\newcommand{\paradepth}{span}
\newcommand{\TAS}[0]{$\mf{TestAndSet}$}
\newcommand{\TS}[0]{$\texttt{TS}$}
\newcommand{\insformat}[1]{\texttt{#1}}
\newcommand{\tas}{{\insformat{test\_and\_set}}}
\newcommand{\faa}{{\insformat{fetch\_and\_add}}}
\newcommand{\FAA}{{\insformat{FAA}}}
\newcommand{\cas}{{\insformat{compare\_and\_swap}}}
\newcommand{\CAS}{{\insformat{CAS}}}
\newcommand{\atinc}{{\insformat{atomic\_inc}}}
\newcommand{\WriteMin}{\insformat{write\_min}\xspace}
\newcommand{\WriteMax}{\insformat{write_max}\xspace}
\newcommand{\writemax}{\WriteMax}
\newcommand{\writemin}{\WriteMin}

%%%%%%%%%%%%%%%%%%%%%%%%%%%%%%%%%%%%%%%%%%%%
%% List: Enumerate / Itemize
%%%%%%%%%%%%%%%%%%%%%%%%%%%%%%%%%%%%%%%%%%%%
% \usepackage[shortlabels]{enumitem}

% Spacing for lists. This can also be set for each list separately by using [...]
\setlist{topsep=0.3em,itemsep=0.2em,parsep=0.1em,leftmargin=*}
% add "wide" to remove in-item indents
%\setlist{topsep=0.3em,itemsep=0.2em,parsep=0.1em,leftmargin=*,wide}
% "nosep" removes all vertical spacing
%\setlist{nosep,wide}

%%%%%%%%%%%%%%%%%%%%%%%%%%%%%%%%%%%%%%%%%%%%
%% Floating: Figure / Table / Algorithm
%%%%%%%%%%%%%%%%%%%%%%%%%%%%%%%%%%%%%%%%%%%%
% \usepackage{float}
% \usepackage[labelfont=bf,font={small},aboveskip=0em, belowskip=0em]{caption}

%%%%%%%%%% Floating Spacing %%%%%%%%%%%%%
% Can also set space around the caption separately
%\setlength\abovecaptionskip{0em}
%\setlength\belowcaptionskip{0em}
% Space between multiple floatings
\setlength{\floatsep}{0em}
% space below floating (distance to the rest of text)
\setlength{\textfloatsep}{0.5em}
% space above tables/figures (distance from the text above)
% for tables in the middle of the page (i.e., not top or bottom), this number is both the top spacing and bottom spacing
\setlength{\intextsep}{0.5em}
%%%%% These two are for double-column floatings, e.g., figure* and table*
\setlength{\dbltextfloatsep}{1em} % floating to text
\setlength{\dblfloatsep}{0.5em} % between floatings

%%%% \tabcolsep changes the horizontal spacing between columns
%\setlength{\tabcolsep}{10pt} % default = 6
%%%% \arraystretch changes the vertical spacing between rows
%\renewcommand{\arraystretch}{1.1} % default = 1

%%%%%%%%%% Subfigures %%%%%%%%%%%%%
% Use "\begin{subfigure}[b]{0.3\textwidth}" for a subfigure
% \usepackage[labelfont=bf,list=true,skip=0em]{subcaption}
\captionsetup[table]{textfont=normalfont,position=bottom}
\captionsetup[figure]{textfont=normalfont,position=bottom}
% Useful environments: \subcaptionbox{caption}{content}
% Useful environments: \subcaptiongroup{ all captions will be labeled as subcaptions }

%%%%%%%%%% Side Captions %%%%%%%%%%%%%
% \begin{SCtable} [⟨relwidth⟩][⟨float⟩] ... \end{SCtable}
% \begin{SCfigure} [⟨relwidth⟩][⟨float⟩] ... \end{SCfigure}
% \begin{SCtable*} [⟨relwidth⟩][⟨float⟩] ... \end{SCtable*}
% \begin{SCfigure*}[⟨relwidth⟩][⟨float⟩] ... \end{SCfigure*}
% \usepackage[rightcaption]{sidecap}

%%%%%%%%%% Wrapfigure %%%%%%%%%%%%%
% \begin{wrapfigure}[lineheight]{position}{width}  ... \end{wrapfigure}
% \usepackage{wrapfig}

%%%%%%%%% Table settings %%%%%%%%%%%%%%
% \usepackage{array}% for extended column definitions
% Multi-line column with fixed length. Use \par for a new line.
\newcolumntype{L}[1]{>{\raggedright\let\newline\\\arraybackslash\hspace{0pt}}m{#1}}
\newcolumntype{C}[1]{>{\centering\let\newline\\\arraybackslash\hspace{0pt}}m{#1}}
\newcolumntype{R}[1]{>{\raggedleft\let\newline\\\arraybackslash\hspace{0pt}}m{#1}}
% bold and center
\newcolumntype{B}{>{\bf}c}
% Rotate text in cells
% \usepackage{rotating}

% \usepackage{booktabs} % provides toprule, bottomrule, midrule, cmidrule, etc.
% \usepackage{multicol,multirow}
% \usepackage{longtable} % provides long table
% \usepackage{supertabular} % similar to long table, allowing tables to take more than one page
% \usepackage{colortbl}
% \usepackage{bigstrut}
% minitab alignment to change inside one cell
\newcommand{\minitab}[2]{\multicolumn{1}{#1}{#2}}

%%%%%%%%%%%%%%%%%%%%%%%%%%%%%%%%%%%%%%%%%%%%
%% Section titles
%%%%%%%%%%%%%%%%%%%%%%%%%%%%%%%%%%%%%%%%%%%%
% \usepackage{titlesec}
% Change section and subsection title to normal font size
\titleformat{\section}{\normalfont\large\bfseries}{\thesection}{1em}{}
\titleformat{\subsection}{\normalfont\normalsize\bfseries}{\thesubsection}{1em}{}

% % Change title spacing
% \titlespacing{\section}{0pt}{0.3em}{0.2em} % left margin, space before, space after
% \titlespacing{\subsection}{0pt}{0.3em}{0.2em} % left margin, space before, space after
% \titlespacing{\subsubsection}{0pt}{0.1em}{1em} % left margin, space before, space after (horizontal)
% \newcommand{\mysubsubsection}[1]{\underline{#1}.}
% \titleformat{\subsubsection}[runin]
% {\normalfont\normalsize\bfseries}{\thesubsubsection}{1em}{\mysubsubsection}

\newcommand{\para}[1]{{\bf \emph{#1}}\,}
\newcommand{\myparagraph}[1]{\noindent\emp{#1}~~}

%%%%%%%%%%%%%%%%%%%%%%%%%%%%%%%%%%%%%%%%%%%%
%% Algorithms
%%%%%%%%%%%%%%%%%%%%%%%%%%%%%%%%%%%%%%%%%%%%
% \usepackage[ruled,lined,linesnumbered,noend]{algorithm2e}
% \usepackage[noend]{algpseudocode}

\makeatletter
\setlength{\algomargin}{.5em}
% Remove right hand margin in algorithm
\patchcmd{\@algocf@start}% <cmd>
  {-1.5em}% <search>
  {0pt}% <replace>
  {}{}% <success><failure>
\setlength{\algomargin}{.5em}   % left margin

\newcommand{\nosemic}{\renewcommand{\@endalgocfline}{\relax}}% Drop semi-colon ;
\newcommand{\dosemic}{\renewcommand{\@endalgocfline}{\algocf@endline}}% Reinstate semi-colon ;
\newcommand{\popline}{\Indm\dosemic}% Undent
\newcommand{\pushline}{\Indp}% Indent

% \SetSideCommentLeft
\newcommand{\SideCommentLeft}[1]{\hfill (*@ \% #1 @*)}
% use \notations{...} or \notes{...} etc.
\SetKwInput{notations}{Notations}
\SetKwInput{notes}{Notes}
\SetKwInput{maintains}{Maintains}

% use \myfunc(right-aligned comment){function name}{...function content...}
\SetKwProg{myfunc}{Function}{}{}
\SetKwProg{lambdafunc}{lambda function}{}{}
% use \parForEach as a regular \For
\SetKwFor{parForEach}{ParallelForEach}{do}{endfor}
\SetKwFor{Justrepeat}{Repeat}{}{}

\SetKw{MIN}{min}
\SetKw{MAX}{max}
\SetKw{OR}{or}
\SetKw{AND}{and}

% CommentStyle
\newcommand\mycommfont[1]{\textit{\textcolor{blue}{#1}}}
%\definecolor{commentgreen}{RGB}{0,128,0}
%\newcommand\mycommfont[1]{\textit{\textcolor{commentgreen}{#1}}}
\SetCommentSty{mycommfont}
% Define the comment font style

%%%%%%%%%%%%%%%%%%%%%%%%%%%%%%%%%%%%%%%%%%%%
%% cref (cleveref)
%%%%%%%%%%%%%%%%%%%%%%%%%%%%%%%%%%%%%%%%%%%%
% \usepackage{nameref,cleveref}
\crefname{section}{Sec.}{Sec.}
\crefname{theorem}{Thm.}{Thm.}
%\crefname{thm}{Thm.}{Thm.}
\crefname{lemma}{Lem.}{Lem.}
\crefname{corollary}{Cor.}{Cor.}
\crefname{table}{Tab.}{Tab.}
\crefname{algorithm}{Alg.}{Alg.}
\crefname{figure}{Fig.}{Fig.}
\crefname{fact}{Fact}{Fact}
% \crefname{table}{Tab.}{Tab.}
\crefname{problem}{Problem}{Problem}

%%%%%%%%%%%%%%%%%%%%%%%%%%%%%%%%%%%%%%%%%%%%
%% Math and Theorem
%%%%%%%%%%%%%%%%%%%%%%%%%%%%%%%%%%%%%%%%%%%%
%% No need to use the packages below with acmart
%\usepackage{latexsym,amsthm,amsmath,amsfonts,amssymb,stmaryrd,mathtools}

% \newtheorem{theorem}{Theorem}%[section]
% \newtheorem{lemma}[theorem]{Lemma}%[section]
% \newtheorem{corollary}[theorem]{Corollary}
% \newtheorem{claim}[theorem]{Claim}
% \newtheorem{fact}[theorem]{Fact}
% \newtheorem{invariant}[theorem]{Invariant}
% \newtheorem{definition}[theorem]{Definition}
\newtheorem{remark}{Remark}

% \left and \right
\let \originalleft \left
\let\originalright\right
\renewcommand{\left}{\mathopen{}\mathclose\bgroup\originalleft}
\renewcommand{\right}{\aftergroup\egroup\originalright}

% \usepackage{scalerel} % stretch a formula

% compact theorem
% \newtheoremstyle{exampstyle}
% {.5em} % Space above
% {1em} % Space below
% {\it} % Body font
% {.5em} % Indent amount
% {\it \bfseries} % Theorem head font
% {.} % Punctuation after theorem head
% {.5em} % Space after theorem head
% {} % Theorem head spec (can be left empty, meaning `normal')

%\theoremstyle{exampstyle} \newtheorem{example}{Example}
%\theoremstyle{exampstyle} \newtheorem{remark}{Remark}
% \theoremstyle{exampstyle} \newtheorem{compactdef}[theorem]{Definition}
% \theoremstyle{exampstyle} \newtheorem{compactlem}[theorem]{Lemma}
% \theoremstyle{exampstyle} \newtheorem{compactprob}[theorem]{Problem}
% \theoremstyle{exampstyle} \newtheorem{compactthm}[theorem]{Theorem}

% \makeatletter
% \newenvironment{proof}[1][Proof]{\par
%   \vspace{-2\topsep}% remove the space after the theorem
%   \normalfont \topsep0pt \partopsep0pt % no space before
%   \trivlist
%   \item[\hskip\labelsep
%         \itshape
%     #1\@addpunct{.}]\ignorespaces
% }{%
%   \hfill\ensuremath{\square} % QED symbol
%   \endtrivlist\@endpefalse
%   %\addvspace{3pt plus 3pt} % some space after
% }
% \makeatother

%%%%%%%%%%%%%%%%%%%%%%%%%%%%%%%%%%%%%%%%%%%%
%% Framedbox
%%%%%%%%%%%%%%%%%%%%%%%%%%%%%%%%%%%%%%%%%%%%
% Use "\begin{mdframed}[style=mystyle] ... \end{mdframed}"
% \usepackage{mdframed}
\definecolor{framelinecolor}{RGB}{68,114,196}
\mdfdefinestyle{mystyle}{linecolor=framelinecolor,innertopmargin=1pt,innerbottommargin=2pt,backgroundcolor=gray!20,skipabove=2pt,skipbelow=0pt}%leftmargin=0,topmargin=
\mdfdefinestyle{densestyle}{linecolor=framelinecolor,innertopmargin=0,innerbottommargin=0,leftmargin=0,rightmargin=0,backgroundcolor=gray!20}
\mdfdefinestyle{compactcode}{linecolor=framelinecolor,innertopmargin=1pt,innerbottommargin=1pt,backgroundcolor=gray!20,skipabove=0pt,skipbelow=0pt,leftmargin=0,rightmargin=0}

% package framed also provides a simple framed box
% \usepackage{framed}

%%%%%%%%%%%%%%%%%%%%%%%%%%%%%%%%%%%%%%%%%%%%
%% Listing codes
%%%%%%%%%%%%%%%%%%%%%%%%%%%%%%%%%%%%%%%%%%%%
% \usepackage{listings}
\newcommand{\codeskip}{{\vspace{.05in}}}

%% LISTING ENVIRONMENT (lstlisting)
\newdimen\zzsize
\zzsize=8pt
\newdimen\kwsize
\kwsize=8pt

\newcommand{\basicstyle}{\fontsize{\zzsize}{1\zzsize}\ttfamily}
\newcommand{\keywordstyle}{\fontsize{\kwsize}{1\kwsize}\ttfamily\bf}

\newdimen\zzlstwidth
%\newlength{\zzlstwidth}
\settowidth{\zzlstwidth}{{\basicstyle~}}
\newcommand{\lcm}{}

\lstset{
%  aboveskip=-0.5 \baselineskip,
%  belowskip=-0.8 \baselineskip,
  xleftmargin=0.5em,
  basewidth=\zzlstwidth,
  basicstyle=\basicstyle,
  columns=fullflexible,
  captionpos=b,
  numbers=left, numberstyle=\small, numbersep=4pt,
  language=C++,
  keywordstyle=\keywordstyle,
  keywords={return,signature,sig,structure,struct,fun,fn,case,type,datatype,let,fn,in,end,functor,alloc,if,then,else,while,with,AND,start,do,parallel,for,parallel_for},
  commentstyle=\rmfamily\slshape,
  morecomment=[l]{\%},
  lineskip={1.5pt},
  columns=fullflexible,
  keepspaces=true,
  mathescape=true,
  escapeinside={@}{@}
}

%%%%%%%%%%%%%%%%%%%%%%%%%%%%%%%%%%%%%%%%%%%%
%% Edits
%%%%%%%%%%%%%%%%%%%%%%%%%%%%%%%%%%%%%%%%%%%%
\newcommand{\newchange}[1]{{\color{red}#1}}
\newcommand{\revision}[1]{{\color{purple}#1}}
\newcommand{\todo}[1]{{\color{red}#1}}

%%%%%%%%%%%%%%%%%%%%%%%%%%%%%%%%%%%%%%%%%%%%
%% Other tools
%%%%%%%%%%%%%%%%%%%%%%%%%%%%%%%%%%%%%%%%%%%%
% \usepackage{tikz} % draw geometric objects

%%%%%%%%%%%%%%%%%%%%%%%%%%%%%%%%%%%%%%%%%%%%
%% PENALTY
%%%%%%%%%%%%%%%%%%%%%%%%%%%%%%%%%%%%%%%%%%%%

% See their definitions and default values in: https://en.wikibooks.org/wiki/TeX/penalty
\binoppenalty=700
\brokenpenalty=0 %100
\clubpenalty=0   %150
\displaywidowpenalty=0   %50
\exhyphenpenalty=50
\floatingpenalty=0
\hyphenpenalty=50
\interlinepenalty=0
\linepenalty=10
\postdisplaypenalty=0
\predisplaypenalty=0 %10000
\relpenalty=500
\widowpenalty=0  %150

%%%%%%%%%%%%%%%%%%%%%%%%%%%%%%%%%%%%%%%%%%%%
%% For Submissions, acmart template
%%%%%%%%%%%%%%%%%%%%%%%%%%%%%%%%%%%%%%%%%%%%
% \setcopyright{none}
% \renewcommand\footnotetextcopyrightpermission[1]{} % This line removes the footnote about the conference and year.
%\def\@titlefont{\huge\sffamily\bfseries} % THIS LINE CHANGES THE FONT OF THE TITLE

% \usepackage{titlecaps} 

%% file: macro.tex
%%%%%%%%%% COMMENTS  %%%%%%%%%%%%%%%

\newcommand{\yan}[1]{{\color{violet}{\bf Yan:} #1}}
\newcommand{\yihan}[1]{{\color{purple}{\bf Yihan:} #1}}
\newcommand{\guy}[1]{{\color{cyan}{\bf Guy:} #1}}
\newcommand{\letong}[1]{{\color{blue}{\bf Letong:} #1}}

\newcommand{\modify}[1]{{\color{blue}#1}}

%%%%%%%%%%%%%%%%%%%%
% Our terms
%%%%%%%%%%%%%%%%%%%%
\newcommand{\batchBFS}{cluster-BFS\xspace}
\newcommand{\BatchBFS}{Cluster-BFS\xspace}
\newcommand{\batch}
{compact-cluster\xspace}
\newcommand{\ccbfs}{C-BFS\xspace}
\newcommand{\ff}{\mathcal{F}}
\newcommand{\edgemap}{\textsc{EdgeMap}\xspace}
\newcommand{\mapsparse}{\textsc{EdgeMap-Sparse}\xspace}
\newcommand{\mapdense}{\textsc{EdgeMap-Dense}\xspace}

\newcommand{\bitsubset}{bit-subset\xspace}
\newcommand{\compactdis}{cluster distance vector\xspace}

\newcommand{\bitvector}{\bitsubset}
\newcommand{\Ssubset}[3]{{#1}_{#2}[#3]}
\newcommand{\Sdist}[1]{\Delta_{#1}}
\newcommand{\Sseen}{S_{seen}}
\newcommand{\Snext}{S_{next}}
\newcommand{\Snew}{S_{new}}

\newcommand{\parlay}{ParlayLib}

\newcommand{\edgef}{\textsc{Edge\_F}\xspace}
\newcommand{\condf}{\textsc{Cond\_F}\xspace}

\newcommand{\cnt}{\mbox{\it cnt\:\!}\xspace}
\newcommand{\dis}{\mbox{\it dis\:\!}\xspace}
\newcommand{\ans}{\mbox{\it ans\:\!}\xspace}
\newcommand{\OLD}{\mbox{\it OLD\:\!}\xspace}
\newcommand{\query}{\mbox{\it query\:\!}\xspace}

%% file: abstract.tex
\begin{abstract}
Breadth-first Search (BFS) is one of the most important graph
processing subroutines, especially for computing the \emph{unweighted distance}.
Many applications may require running BFS from multiple
sources.  
Sequentially, when running BFS on a cluster of
nearby vertices, a known optimization is using 
\emph{bit-parallelism}.
Given a subset of vertices with size $k$ and the distance between any pair of them is no more than $d$, BFS can be applied to all of them in total work $O(dm(k/w+1))$, where $w$ is the length of a word in bits and $m$ is the number of edges.
We will refer to this approach as \emp{cluster-BFS (C-BFS)}. 
Such an approach has been studied and shown effective both in theory and in practice in the sequential setting. 
However, it remains unknown how this can be combined with thread-level
parallelism.

In this paper, we focus on designing efficient \emph{parallel} C-BFS
based on BFS to answer unweighted distance queries. 
Our solution combines the strengths of bit-level parallelism and thread-level parallelism, and achieves significant speedup over the plain sequential solution. 
We also apply our algorithm to real-world applications. In particular, we identified another application (landmark-labeling for the approximate distance oracle) that can take advantage of parallel C-BFS. 
Under the same memory budget, our new solution improves accuracy and/or time on all the 18 tested graphs. 
\end{abstract}

\hide{
Particularly, for a vertex $s$ and a subset of its $k$ neighbors within $d$ hops, BFS can be applied to all of them in total $O(dk/w+1)$ work, where $w$ is the length of a word in bits. 
}

%% file: intro.tex
% \vspace{-1em}
\section{Introduction} 
Breadth-First Search (BFS) is one of the most important graph
processing subroutines.  Given a graph $G=(V,E)$ and a vertex
$s\in V$, BFS visits all vertices in $V$ in increasing order of
(hop) distance to $s$.  
BFS can be used for many purposes. 
One of the most common use scenarios for BFS is to compute the \emph{unweighted distance} 
from the source.  
%In this paper, we focus on using BFS to compute the unweighted
%distances on an undirected graph.
In this paper, we focus on designing efficient \emph{parallel} approaches
based on BFS to answer unweighted distance queries. 
%In this paper, we focus on the one of the most common use scenarios for BFS, which is to compute the \emph{unweighted distance} 
%from the source.,
Throughout the paper, we use $n=|V|$ and $m=|E|$, and use ``distance'' to refer to the hop distance on an unweighted graph. 

Many applications may require running BFS from multiple
sources.  Examples of this are using BFS for low-diameter
decomposition~\cite{miller2013parallel}, all-pairs shortest paths (APSP), or oracles for
exact or approximate APSP. 
A key observation is that
\emph{bit-parallelism}~\cite{chan2012all} can be used effectively when running BFS on a cluster of
nearby vertices~\cite{chan2012all,akiba2013fast}.
% Particularly, for a vertex $s$ and a subset of its $k$ neighbors within $d$ hops, BFS can be applied to all of them in total $O(dk/w+1)$ work (number of operations), where $w$ is the length of a word in bits~\cite{chan2012all}.
Given a subset of vertices with size $k$ and the distance between any pair of them is no more than $d$, BFS can be applied to all of them in total $O(dm(k/w+1))$ work (number of operations), where $w$ is the length of a word in bits and $m$ is the number of edges~\cite{chan2012all}.
Since a machine word must hold at least $\Omega(\log n)$ bits to store a pointer, this means $w=\Omega(\log n)$.
We will refer to this approach as \emph{\batchBFS (\ccbfs)},
and present more details in \cref{sec:algorithm}. 
Chan~\cite{chan2012all} used this idea to
develop an all-pair shortest-path algorithm that runs in
$O(mn / w)$ work (when $m =\Omega(n \log n \log \log \log n))$.
In addition to saving time, \ccbfs{} also
saves space: 
it only uses $O(d)$ words per $w$ vertices 
instead of a word (or at least enough bits to store a distance) per vertex as in standard BFS. 
Akiba et al.~\cite{akiba2013fast} used this idea in the exact two-hop distance oracle but only considered the special case for $d=2$ (a star-shaped cluster: a vertex and its neighbors). 
We refer to this algorithm as the AIY algorithm. 
Both of the previous papers focus on the sequential setting. 

While the \ccbfs{} with bit-level parallelism has shown to be effective in
sequential settings, surprisingly, we know of no previous work
combining it with thread-level parallelism. 
BFS is one of the most well-studied parallel graph processing problems, 
and state-of-the-art solutions have been highly optimized 
using techniques such as direction optimizations~\cite{Beamer12,shun2013ligra}. 
To be practical, any \ccbfs{} would have to compete with these. 
Our goal is to develop an efficient \ccbfs{} with high parallelism such that it 
(1) achieves the same level of parallelism as the standard parallel BFS, 
with additional benefits by using clusters,
(2) supports a clean interface that is flexible for different parameter settings (i.e., varying $d$ and $k$), 
and (3) facilitates various real-world applications. 
In this paper, we provide a systematic study of parallel \ccbfs{} and achieve all three goals above. 

To achieve \emph{high performance}, we design an efficient parallel algorithm. 
Our algorithm is work-efficient (i.e., it has the same asymptotical work as the sequential counterpart). 
It has the same span as regular parallel BFS algorithms (e.g.,~\cite{shun2013ligra}), 
which $\tilde{O}(D)$ for graph diameter $D$, and thus is best suited for small-diameter graphs, such as
social networks, computer networks, or web graphs.  
Our algorithm uses the \emph{directional optimization} that has been shown to be useful for parallel BFS. 
By doing this, our algorithm achieves the strengths of both bit-level and thread-level parallelism: it has high parallelism as in the state-of-the-art parallel BFS algorithms and obtains additional performance gain by using bit-level parallelism. 

To achieve a \emph{flexible interface}, we designed our algorithm for
general $k$ and $d$, easily integrating into various applications with
user-defined parameters. 

To understand how \ccbfs{} can \emph{facilitate real-world applications}, we study two Distance Oracle (DO) techniques that can benefit from \ccbfs{}: the exact DO as in~\cite{akiba2013fast} and the Landmark Labeling (LL) for an approximate DO. As far as we know, our work is the first to use \ccbfs{} to accelerate LL. 

We implemented our \ccbfs{} algorithm and the two applications.
We compare our algorithm with multiple baselines to study the
performance gain in depth, and test it on 18 graphs with various types and sizes on a 96-core machine. 
%Our approach has three potential sources of improvement over 
%standard sequential BFS: 1) direction optimization, 2) bit-parallelism,
%and 3) thread-level parallelism.  
Compared to standard sequential BFS, 
our algorithm employs both bit-parallelism (on clusters) and thread-level parallelism 
(along with optimizations used in parallel BFS) to improve performance. 
In the simplest setting where $k=64$ and $d=2$, the combination of them 
enables up to {1119}$\times$ speedup ({500}$\times$ on average) 
compared to the plain sequential BFS,
%Each of the three techniques contributes to this high improvement, 
%Both techniques
where bit-level and thread-level parallelism each contributes about 20$\times$
speedup. 
%We summarize the performance gain of each technique in \cref{fig:performance}. 
\hide{Interestingly, we observed that both bit-level parallelism
and thread-level parallelism more positively impact the performance when the other technique presents, 
and work very well in synergy. 
Similar results are achieved for different values of $d$, where the total improvement is xx times. \letong{?}
Therefore, we believe our work on an efficient implementation combining both of them fills the gap in the existing 
study of both \batchBFS{} and parallel BFS. }
Interestingly, by comparing with the performance of existing work 
(Ligra~\cite{shun2013ligra}, where only thread-parallelism is used, 
and AIY~\cite{akiba2012shortest}, where only bit-level parallelism is used),
we observed that both bit-level parallelism
and thread-level parallelism and work very well in synergy. 
Each of them still fully contributes to the performance when the other
is present, achieving the same level of improvement
as when used independently (see \cref{fig:performance}). 
%Similar results are achieved for different values of $d$, where the total improvement is xx times. \yihan{check if this is tue}
Therefore, we believe our work on an efficient implementation combining thread-level and bit-level parallelism fills the gap in the existing 
study of both \ccbfs{} and parallel BFS. 

We also studied the performance of our \ccbfs{} with different parameters. Typically, $k$ is set to be $\Theta(w)$, and the work (and space) is proportional to $d$. Our result shows that the running time increases almost linearly with value of $d$, especially when $d$ is small. 
This explains why existing work (e.g.,~\cite{akiba2012shortest}) tends to choose the smallest $d=2$ case in real-world applications.

Applying our algorithm also gives significant improvement to the aforementioned applications. 
For the 2-hop distance oracle, our parallel implementation outperforms the sequential AIY algorithm~\cite{akiba2013fast} by 9--36$\times$, 
and can process much larger graphs than the AIY algorithm. 
For landmark labeling (LL), with a fixed memory budget, 
\ccbfs{} improved regular LL in either accuracy or preprocessing time on all 18 tested graphs,
and improved \emph{both} on {14/18} graphs. 
This is due to the saving in space allowing more
landmarks to be used for \ccbfs{}. 
We observed that using $d=2$ achieved better overall performance in accuracy, time and space. 
% We plan to release our code.
% \modify{We release our code and data on Zenodo~\cite{wang_2024_13905461,wang_2024_13909778}.}
\ifconference{Due to the page limit, we present more results in our full version. }
\iffullversion{Due to the page limit, we present more results in the Appendix. }

\input{figs_algs/fig_performance.tex}
\hide{Since the importance of BFS, how to improve the efficiency of BFS is a hot topic.  Researchers approach this goal through different technique paths. One approach is exploring bitwise parallelism.  The representations of this approach is Cluster-BFS, which conduct BFSs from a cluster of vertices that are close to each other (within restricted diameter).  In theory, the bitwise operation usually can save a $\log n$ factor of work and space \cite{chan2012all}. In practice, it was used to accelerate and save space for the exact 2-hop distance oracles \cite{akiba2013fast}. Another approach is taking the advantage of thread-level parallelism. One representation is Ligra\cite{shun2013ligra}. Ligra BFS explore the thread level parallelism and directional optimizations, that accelerate running BFS from a single source. Both approaches are shown to be efficient in practice. However, there is no existing work that combine the two techniques together. In this paper, we propose parallel \batchBFS, that both explore the benefits of bitwise and thread-level parallelism. We observed an interesting fact that they can enhance the effects of each other.  On a 96-core machine, the speedup of related work Ligra BFS\cite{shun2013ligra} and Akiba BFS\cite{akiba2013fast} to the standard sequential BFS is xx and xx, and the speedup of our parallel \batchBFS is xx, which is greater than the product of previous two speedups. The relation is shown in \cref{fig:performance}. }

%% file: figs_algs/fig_performance.tex
\begin{figure}
  \centering
  \includegraphics[width=\columnwidth]{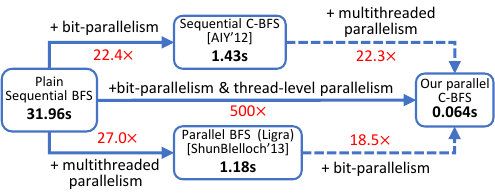}
  %\caption{\small\textbf{Summary of performance improvements gained by three main techniques.} 
%  We test the  running time of BFSs from a cluster of 64 vertices on different baselines. 
%  Some baselines without clusters are achieved by using Ligra~\cite{shun2013ligra}. 
%  The numbers are geometric means across 21 tested graphs.  Full results are shown in \cref{table:microbenchmark} and \cref{fig:par_compare}. 
%  \yihan{update the figure and caption}
%  }\label{fig:performance}
\caption{\small\textbf{Performance comparison with existing work.} 
  We test the running time of BFSs from a cluster of 64 vertices. 
  The baselines are Ligra~\cite{shun2013ligra} that only uses thread-level parallelism and AIY'12~\cite{akiba2012shortest} that only
  uses bit-level parallelism. 
  The numbers are geometric means across 18 graphs.
  Full results are shown in \cref{table:microbenchmark} and \cref{fig:par_compare}.
  }\label{fig:performance}

\end{figure}

%% file: prelim.tex
\section{Preliminaries}
\label{sec:prelim}
% \subsection
% {Notations and Computational Model}
\myparagraph{Notations.} Let $G=(V,E)$ be an unweighted graph.  
%We use $n$ and $m$ to denote the number of vertices $|V|$ and the number of edges $|E|$, respectively. 
We use $n=|V|$, $m=|E|$, %to denote the number of vertices $|V|$ and the number of edges $|E|$, respectively. 
and use $D$ to denote the diameter of the graph.
%We focus on undirected graphs because both of our applications (landmark labeling and 2-hop distance oracle)
%are used on undirected graphs, but the \batchBFS{} idea is applicable to directed graphs as in regular parallel BFS. For simplicity, 
%Both of our applications (landmark labeling and 2-hop distance oracle) are designed for undirected graphs, 
%but the idea of \batchBFS{} is general on both directed and undirected graphs. 
%Therefore, in \cref{sec:parallel_BFS} and \cref{sec:algorithm}, we present our algorithms on directed graphs to keep it general.
%In \cref{sec:applications}, we assume an undirected graph to better describe the applications. %\letong{decide}
Let $N(v)=\{u\in V~|~(v,u) \in E\}$ be the set of neighbors of vertex $v\in V$. In directed graphs, $N^+(v)$ and $N^-(v)$ represent outgoing and incoming neighbors, respectively. 
We use $\delta(u,v)$ to denote the shortest distance between $u$ and $v$. 
We assume machine word size $w=\Omega(\log n)$,  
%In realistic settings, we are interested in large input size where 
such that the vertex and edge IDs are within constant words. 
Let $S=\{s_1,s_2,...,s_k\}$ represent a cluster, where $k$ is the cluster size. Let $d$ be the diameter (maximum distance between any pair) of the cluster. 
We summarize the notations in \cref{table:notations}.
%Note that $w=\Omega({\log n})$ to hold the vertex and edge ids within constant words.  
%We make this assumption throughout the paper. 

\myparagraph{Computational Model.} We use the binary fork-join parallel model~\cite{CLRS,blelloch2020optimal}, with work-span analysis~\cite{blumofe1999scheduling,gu2021parallel}.
We assume a set of \thread{}s that access a shared memory.
A thread can \forkins{} two child \thread{s} to work in parallel, and then waits.
When both children complete, the parent thread continues.
A parallel for-loop can be simulated by recursive \forkins{s} in logarithmic levels.
The \defn{work} of an algorithm is the total number of instructions, and
the \defn{span} is the length of the longest sequence of dependent instructions.
We can execute the computation using a randomized work-stealing scheduler~\cite{blumofe1999scheduling,gu2022analysis}.

We assume two unit-cost \emph{atomic} operations. 
% \textsc{Compare\_And\_Swap} $(p, v_{\mathit{old}}, v_{\mathit{new}})$ and \textsc{fetch\_and\_or}$(p,v_{\mathit{new}})$.
\textsc{compare\_and\_swap}$(p, v_{\mathit{old}}, v_{\mathit{new}})$ atomically reads the memory location pointed to by $p$, 
and writes value $v_{\mathit{new}}$ to it if the current value is $v_{\mathit{old}}$. 
It returns \emph{true} if it succeeds and \emph{false} otherwise.
%\textsc{fetch\_and\_or} can be implemented by \textsc{compare\_and\_swap}.
\textsc{Fetch\_And\_Or}$(p,v_{\mathit{new}})$ atomically reads the memory location pointed to by $p$, takes the bitwise \textsc{Or} operation with value $v_{\mathit{new}}$, and stores the results back. It returns \emph{true} if $v_\mathit{new}$ successfully sets any bit stored in $p$ to be 1, and \emph{false} otherwise.  
Most machines directly support these instructions. %in hardware.

\myparagraph{Parallel BFS.} We briefly review parallel BFS, because it is one of our baselines, and some of the concepts are also used in our \batchBFS. 
Parallel BFS starts from a single source $s\in V$ \iffullversion{(high-level idea in \cref{alg:simpleBFS})}. 
The algorithm maintains a \emp{frontier} of vertices to explore in each round, starting from the source, and finishes in at most $D$ rounds. In round $i$, the algorithm processes (visits their neighbors) of the current frontier $\ff_{i}$, and puts all their (unvisited) neighbors in the next frontier $\ff_{i+1}$. If multiple vertices in $\ff_{i}$ attempt to add the same vertex to $\ff_{i+1}$, a \textsc{compare\_and\_swap} is used to guarantee that only one will succeed. 

One widely-used optimization for parallel BFS is directional optimization~\cite{Beamer12,shun2013ligra}. At a high level, when the frontier size $|\ff_i|$ is large, the algorithm will not process $\ff_i$, but instead visit each unprocessed vertex $v$, and determine if $v$ has an incoming neighbor in $\ff_i$. If so, $v$ will be put in $\ff_{i+1}$. Such an optimization is observed to be effective, especially on small-diameter graphs. \iffullversion{We present more details in \cref{sec:parallel_BFS}.}\ifconference{We present more details in our full version paper.}

\input{figs_algs/table_notation.tex}

%% file: figs_algs/table_notation.tex
\begin{table}[t]
\small

\rule{\columnwidth}{.05em} % \vspace{-.05in}

% \vspace{-.05in}

\begin{description}[labelwidth=.1in,leftmargin=.2in]%[labelwidth=.3in, leftmargin=.3in, itemindent=.2in]
    \item [$G=(V,E)$]: the input graph. $n=|V|$ and $m=|E|$. 
    \item[$S=\{s_1,...,s_k\}$]: the source cluster for the BFS.
    \item[$k$]: the cluster size, i.e., $k=|S|$.
    \item[$d$]: the diameter of the cluster.
    \item[$w$]: the length of a word in bits. $w=\Omega(n)$. 
    \item[$D$]: the diameter of the graph. 
    \item[$\delta(u,v)$]: the shortest distance between $u$ and $v$.    
\end{description}
\vspace{-0.08in}
\rule{\columnwidth}{.05em} % \vspace{-.05in}
% \vspace{-.08in}
\caption{\textbf{Notations in the paper.}\label{table:notations}}
\end{table}

%% file: algorithm.tex
\section{Parallel \BatchBFS{}}\label{sec:algorithm}

\BatchBFS (\ccbfs) runs BFS from a cluster of sources $S \subseteq V$.  
If the sources have diameter $d$ (maximum distance between any pair), then all distances from $S$ to any vertex $v \in V$ will differ by at most $d$.  
This means that if we run a set of BFSs from $S$, synchronously, every $v \in V$ will appear in at most $d + 1$ consecutive frontiers.  
\BatchBFS{} takes advantage of this by representing all the sources in $S$ that can visit a given vertex, on a given round, as a vector of booleans (bits).
In this way, each vertex will be visited at most $d+1$ times instead of $|S|$ times if all searches are performed separately.  
More details are described in \cref{sec:bitwise}.  
Importantly, if the bit-vector fits in $O(1)$ words, 
the distances of all $|S|$ sources can be handled (propagated from a vertex in the current frontier to a neighbor)
with $O(1)$ bitwise logical operations.  
%all operations can be done with $O(1)$ bitwise logical operations.  
Since a machine word must hold at least $\Omega(\log n)$ bits (so it can represent a pointer), this means the algorithm can save a factor of $\Omega(\log n/d)$ work.

It appears that Chan first described this idea~\cite{chan2012all}, but he did not go into any details of the implementation, but just saying that this is possible.
Akiba et al.~\cite{akiba2013fast} later showed a concrete implementation based on this idea, with the limitations that it is sequential and only works on the cluster of a star with $d=2$ (a center vertex and its neighbors).
To the best of our knowledge, there has been no previous work on developing a parallel implementation of the \batchBFS algorithm. 
Indeed it is challenging to achieve high performance given how BFS has been widely studied with numerous optimizations both sequentially and in parallel.
In this paper, we propose our algorithm, given as \cref{alg:c-bfs}, which is an efficient parallel \batchBFS{} implementation with a general interface and low coding effort.
In the following sections, we will first introduce the bitwise representation to maintain the distances in \batchBFS{} and then give our parallel \batchBFS algorithm.

\input{normal-cluster_BFS.tex}

%We will discuss the technical contributions 
\hide{
\letong{
Akiba et. al.~\cite{akiba2013fast} used this idea in exact distance oracles to save index space and construction time. However, their \ccbfs{} needs to classify all the edges to either sibling edges or child edges before propagate the $S_{-1}, S_{0}$ subset labels, which first cost and $O(E)$ extra space to store the edge lists, and also make it harder to extend to more general clusters. Chan propose an theoretical \ccbfs that applied to more general clusters, but he did not go into any details of the implementation. 
Our \ccbfs implicit finding the subset labels by the arrive order of soruces, which avoid $O(E)$ auxilarate memory usage, can be easily extend to $d$-hop stars ($d>1$) and even more general clusters that not necessary to be a star. 
We further take the advantage of \mapsparse (forward) and \mapdense (backward) that are commonly used in most state-of-art parallel BFS algorithms, and parallel it to take the advantage of thread-level parallelism.
}
}
\hide{
%We combine the bitwise compression idea with \edgemap framework and put extra effort to extend batches to more general restrictions, easy to implement by \edgemap framework and make it efficient. 
% Conceptually, \batchBFS is not hard to parallelize, but designing a parallel algorithm that is efficient, simple to implemented by existing graph processing tools, and can be extended to more general batches, is challenging. 
% Both of them choose a center vertex and $O(\log n)$ neighbors as a source batch. 
Directly apply existing work does not give high performance in parallel.  
The algorithm in \cite{akiba2013fast} always chooses a center vertex and $O(\log n)$ neighbors as batch. Overall, it performs a single BFS frontier by frontier from the center vertex. This BFS has two phases in each round: classifying the edges into two categories by whether the two endpoints are in the same frontier and processing the two class of edges separately. 
In order to classify the edges into two categories and process them in parallel, we either need $O(m)$ extra memory or re-visiting edges, which is the main cost of BFS.
% Because it has two phases to process edges, naively parallellizing the algorithm will cost the span to be twice as a single parallel BFS. 
The algorithm in \cite{chan2012all} also choose a center vertex, and it extends the source batch to be $c$-hop neighbors. But it did not provide an efficient implementation of the algorithm. In order to parallel \batchBFS efficiently and extend to more general batches, we give a new abstraction of this algorithm and combine with \edgemap to reduce the implementation effort with good performance.  In the following sections, we will first introduce our bitwise-representation based on our abstraction, and then give our parallel \batchBFS algorithm.
}

% Bitwise-Batch BFS compress a batch of vertices into a bit vector and do BFS from them. Compared to naively running BFS from each sources individually, it costs sublinear time and space. 

\hide{It takes the advantage of "word tricks". A subset of $\{1,...,k\}$ can be represented as a bit vector and stored in a sequence of $O(\lceil k/b \rceil)$ words, where $b$ is the word length and $b = \lfloor \alpha \log n \rfloor$ for some suitable constant $\alpha$. Given the bit-vector representations of sets $S_1,S_2,...$, the cost of computing the bit-vector representations of $S_1 \cup S_2$, $S_1 \cap S_2$ and $S_1 - S_2$ is $O(k/log n + 1)$ time.}

%% file: normal-cluster_BFS.tex
\subsection{Cluster Distance Representations}\label{sec:bitwise}
Given a set $S$ of source vertices with diameter $d$, \batchBFS computes a compact representation of the shortest distance from every source in $S$ to every vertex in $V$. 
% Before we introduce our bitwise representation, let us see some lemmas first.
The idea of \batchBFS{} is based on the following fact.

\begin{fact}
  \label{fact:triangle_inequality}
  On an unweighted graph, if the distance betweeen vertex $s_1$ and $s_2$ is $d$, then for any vertex $v\in V$, $|\delta(s_1,v)-\delta(s_2,v)|\leq d$.
  %$\delta(s_1,s_2)=c$, 
\end{fact}

%\cref{fact:triangle_inequality} can be proved by the triangle inequality. 
%Because the graph is unweighted, the distance in the graph satisfies the triangle inequality.
%\cref{fact:triangle_inequality} says that if two vertices $s_1$ and $s_2$ are $d$ distance apart, for any vertex $v\in V$, the difference of distance from $v$ to $s_1$ and $s_1$ is no larger than $d$.  It can be provided by the triangle inequality.
For example, if $s_1$ and $s_2$ are neighbors, the distances from a vertex $v$ to them can differ by at most 1. 
% Without loss of generality, assume $\delta(s_1,v)> \delta(s_2,v)$. 
% Then $\delta(s_1,v)$ must be no larger than $\delta(s_1, s_2)+\delta(s_2,v)$. 
% \cref{fact:triangle_inequality} can then be proved by observing that $\delta(s_1,v)-\delta(s_2,v) \leq \delta(s_1, s_2)\le d$. 
We can further extend Fact~\ref{fact:triangle_inequality} to \cref{corollary:distance_range}, which says if a cluster of vertices $S$ has diameter $d$, 
the distances from $v$ to vertices in $S$ differ by at most $d$. 
% which is used to define the compact distance representation and design \batchBFS.

\begin{corollary}
  \label{corollary:distance_range}
  On an unweighted graph, given a set $S$ of vertices with diameter no more than $d$, for any vertex $v\in V$, we have 
  \[\max_{s\in S}\delta(s,v) - \min_{s\in S}\delta(s,v) \leq d\]
\end{corollary}
\hide{
\cref{corollary:distance_range} can be easily prooved by the definition of diameter and \cref{fact:triangle_inequality}. 
The diameter of $S$ is no more than $c$ means for any two vertices $s_1, s_2 \in S$, their distance $\delta(s_1, s_2)\leq c$. According to \cref{fact:triangle_inequality}, for any vertex $v\in V$, $\forall s_1, s_2\in S, |\delta(s_1,v)-\delta(s_2,v)|\leq \delta(s_1, s_2)$, so that $\max_{s\in S}\delta(s,v) - \min_{s\in S}\delta(s,v)=\max_{s_1, s_2\in S} |\delta(s_1,v)-\delta(s_2,v)| \leq \max_{s_1,s_2 \in S} \delta(s_1, s_2) \leq c$.}

% Based on \cref{lemma:distance_range}, we can say for any vertex $v\in V$, the vertices in $S$ will visit $v$ in $c+1$ continuous rounds. 

Therefore, for each vertex $v$, we can classify the sources in $S$ by their distances to $v$.  
Let $\Sdist{v}=\min_{s\in S}\delta(s,v)$ be the smallest distance from any source in $S$ to $v$. 
According to \cref{corollary:distance_range}, the distance between $v$ and any $s\in S$ must be in range $[\Sdist{v}, \Sdist{v}+d]$. 
This divides all vertices in $S$ in $d+1$ different subsets based on their distances to $v$. 
% Let $S_i[v]$ be the subset of sources in $S$ that has distance to $v$ as $\delta_{\min}(v)+i$. 
Let $\Ssubset{S}{v}{i}$ be the subset of sources in $S$ that has a distance to $v$ as $\Sdist{v}+i$. %(reach $v$ in round $\Sdist{v}+i$). 
More formally, 
\[\Ssubset{S}{v}{i} = \{s\in S ~|~ \delta(s,v) = \Sdist{v} + i\}\label{eq:subset} \]
% ~ \text{, for $i \gets 0$ to $c$} \]
% Let $\delta^-(v)$ and $\delta^+(v)$ represent $\min_{s\in S}\delta(s,v)$ and $\max_{s\in S}\delta(s,v)$ separately.
% Therefore, for every $v$ we define $c$ subsets of $S$ by: 
% We define the smallest distance from $S$ to $v$ as $\delta_{\min}(v)$ and the largest one as $\delta_{\max}(v)$. For every $v$ we can define $c+1$ disjoint subsets of $S$ by their distance to $v$:
% \[S_i[v] = \{s\in S | \delta(s,v) = \delta_{\min}[v]  + i\}~ \text{, for $i \gets 0$ to $c$} \label{eq:subset}\]
Then for a vertex $v$, the distances between $v$ and all sources in $S$
can be represented by the $(d+2)$-tuple $\langle \Ssubset{S}{v}{0..d}, \Sdist{v} \rangle$,
which we call the \emph{\compactdis{}} of $v$ to $S$.

Note that if $|S|=w$, we can use a one word bit-vector to represent any subset of $S'\subseteq S$: 
bit $i$ is 1 iff. the $i$-th element in $S$ is also in $S'$. 
We call such a representation of a subset of $S$ a \emph{\bitsubset}. 
% Since each $S[v][i]$ is a subset of $S$ with size $w$,
% each of them can be represented as a \bitsubset{} with $O(1)$ words. 
%we can simply use $O(1)$ word bit-flag to represent them. 
In this way, a \compactdis{} 
only takes $d+1$ words for \bitsubset{s} and one byte to store the shortest distances from $v$ to $|S|$ sources (assuming $D < 256$).

\hide{
According to \cref{corollary:distance_range}, $i$ is in range $[0,d]$, thus, $d+1$ subsets of $S$ are enough to cover $S$, i.e. $S_0[v]\cup ... \cup S_{d}[v] = S$ if $v$ is reachable from $S$. 
% Note that if $v$ is reachable from $S$, $S_0[v]\cup ... \cup S_{c}[v] = S$. 
The shortest distances from every vertex in $S$ to each vertex $v$ can be represented by $\langle S_0[v],\dots, S_{d}[v],\delta_{\min}(v) \rangle$. 
}

\input{figs_algs/fig_bitwise.tex}
An illustration for the \bitsubset{} and \compactdis{} is shown in \cref{fig:bitwise}. 
In this example, $S$ is the set $\{A,B,C,D\}$, and the diameter of the subgraph is $d=2$. We need subsets $\Ssubset{S}{v}{0..2}$ for each vertex $v$, which are represented by the \bitsubset{s}, each with four bits. 
%$A$ is represented by the leftmost bit, and $D$ is represented by the right most bit. 
$A$ to $D$ are represented by the four bits from left to right. 
From the \compactdis{}, we can recover the shortest distance from all the sources $s\in S$ to each vertex $v\in V$ by the fact that each source in $\Ssubset{S}{v}{i}$ %contains the vertices in $S$ that 
has distance $\Sdist{v}+i$ to $v$. For example, for vertex $F$, $\Ssubset{S}{F}{1}$ is 1010, which represents the subset $\{A,C\}$, we can infer $\delta(B,A)=\delta(B,C)=\Sdist{F}+1=2$.  

The main idea of \batchBFS{} is to use bitwise operations on \bitsubset{s} to quickly compute the union/intersection of the sets,
allowing us to use the \compactdis{} of $v$ to compute the \compactdis{s} of its neighbor $u$ in constant time. 
In the following, we elaborate on our parallel \batchBFS{} algorithm. 
%Similarly, $S_1$ is 1010, meaning that $A$ and $C$ have distance 2 to $F$. %, and $S_3$ is 0001 that represents $\{D\}$ who has distance 3 to $v$. 
%In practice, we choose an integer length size of sources for $S$, so that $S_i[v]$ can be represented by a integer. 
%For simplicity of description later, we define a type called \textit{\bitvector}, which is actually an integer and  represents a subset of $S$, and we use set operations to describe it.  

% Bit-parallel optimizations explore the fact that computers can perform bitwise operations on a word at once. The word length is commonly 8, 32 or 64 in computers of the day. In the following of this paper, we use $b$ to represent the word length.

% In sequential setting, people develop \batchBFS in All Pairs Shortest Paths algorithms (APSP) \cite{} and exact distance oracles \cite{} to accelerate the process of dense vertices (vertices with large degrees).  Although the definition details are different for different applications, their high level idea is to get BFS distances from a vertex and its neighbors to all vertices in $V$ by one bitwise-batch BFS, which  compacts $b$ neighbors in a word and update it during the edge visit. In this paper, we will only consider the 1-hop neighbors in our algorithm, but it can be easily extend to further hop neighbors. In this section, we first introduce the definition and review the sequential \batchBFS in \cref{sec:seq_bit_bfs}, then we introduce our thread-level parallel \batchBFS in \cref{sec:pal_bit_bfs}.

 \input{figs_algs/alg_batchBFS}

% \subsection{Parallel \titlecap{\batchBFS}}
\subsection{Our Parallel Algorithm}
\label{sec:parallel_batchBFS}

%Parallel \batchBFS is to compute the bitwise representations for each vertex with thread-level parallelism. 
In this section, we introduce our parallel algorithm to compute the \compactdis{} for all vertices given a source cluster $S$, which is $\langle \Ssubset{S}{v}{0..d}, \Sdist{v} \rangle$ for all $v\in V$. 
The pseudocode of our \batchBFS is shown in \cref{alg:c-bfs}.
% The first idea is that, for an edge $(u,v)$, if a vertex $s \in S$ visits $u$ in round $i$, then $s$ should visit $v$ in round $i+1$. 
% Recall that the key to using \edgemap{} is to specify the two functions \edgef{$(u,v)$} to process an edge $(u,v)$,
% and \condf{$(v)$} to decide if $v$ should be put in the next frontier. 

Our \batchBFS algorithm is based on the following fact: 
if $u$ and $v$ are neighbors and there is a path from a source $s\in S$ to $u$ with length $i-1$, then there must exist a path from $s$ to $v$ with length $i$. 
In round $i$, $u$ records all the sources in $S$ that reach it in round $i$ by a \bitvector{}. 
When $u$ visits its neighbor $v$, $u$ propagates this \bitvector{} to $v$ by taking a bitwise $\textsc{Or}$ operation with the \bitsubset{} representing the vertices reaching $v$ in round $i+1$ (\cref{alg:c-bfs}: \cref{line:ccbfs_fao}). 
According to \cref{corollary:distance_range},  all the vertices in $S$ will visit $v$ at least once during round $\Sdist{v}$ to round $\Sdist{v}+d$;
in other words, all the vertices will be put into the frontier for $d+1$ times. 
Therefore, we need to record the number of times $v$ has been put in the frontier. 
When $v$ has been put in the frontier $d+1$ times, since all sources in $S$ must have already visited $v$, we do not need to process $v$ anymore. 
Otherwise, we will process $v$ and put it to the next frontier since other sources in $S$ may visit $v$ in the future. 

%Our \batchBFS algorithm is given in \cref{alg:c-bfs}.
In our algorithm, we use a boolean array $\Sseen[\cdot]$ to store whether a vertex has been visited in all previous rounds by any sources in $S$, and $\Snext[v]$ includes the vertex if it is also in the current frontier (i.e., visited by any vertex in the current round).
We denote $\Sdist{v}$ as the first round that any vertex from $S$ touches vertex $v$.
\cref{alg:c-bfs} has two stages: initialization and traversing.

%we will not process it any more. 
% Instead of only visit each vertex once as in normal BFS, we continue visit a vertex until it has been visited $c$ times. 
% Note that $S_0[v]$ to $S_{c-1}[v]$ are comptued through the \batchBFS, but $S_c[v]$ is computed as post computed by $S[v]-(S_0[v]\cup ... \cup S_c[v])$.

%\input{figs_algs/alg_batchBFS.tex}

% circular queue
\myparagraph{Initialization} 
This step is relatively simple.  
We initialize the arrays of $\Sseen[\cdot]$, $\Snext[v]$, and $\Sdist{v}$.
We use another, array $r[\cdot]$, to avoid duplication of vertices in the frontier. Later in the traversing stage, when multiple vertices want to add $v$ to the next frontier at the same time, only one can successfully set $r[v]$ to the current round number by atomic operation \textsc{compare\_and\_swap} (\cref{line:ccbfs_cas}), and the successful vertex will put $v$ to the next frontier. 

\myparagraph{Traversing} 
At the beginning, the algorithm puts all the sources $s\in S$ into the first frontier $\ff_0$.
Then, we visit all vertices by frontiers. 
In each round, we process frontiers in two stages, where the first stage processes vertices and the second stage processes edges. 
In the first stage (\cref{line:batch_edgemap_begin} to \cref{line:batch_edgemap_first}), we first compute the sources that newly visited $u$ by $\Snew \gets \Snext[u]\setminus\Sseen[u]$ (\cref{line:diff_label}). 
Note that the \bitsubset $\Snew$ contains sources whose distances are $i$, which is also $\Ssubset{S}{u}{i-\Sdist{u}}$ (\cref{line: updateSset}). 
Then we update $\Sseen[u]$ to include newly visited vertices (\cref{line:batch_edgemap_first}), and set $\Sdist{u}$ to the current round number if it has not been set yet (\cref{line: update_Sdist}). 
In the second stage (\cref{line:ccbfs_edgemap_begin} to \cref{line:ccbfs_edgemap_end}), 
we process the neighbor vertices  of the current frontier that have not been visited for $d$ times already (\cref{line:ccbfs_cond}). 
For an edge from $u\in \ff_i$ to its neighbor $v$,
we propagate the sources seen so far by $u$, $\Sseen[u]$, to $v$ (\cref{line:ccbfs_fao}).
In general, if any source $s\in S$ visited $u$ in the previous round, $s$ should also visit $v$ in this round,
and should be included in the $\Snext[v]$ for $v$ in this round. 
If the $\Snext[v]$ is changed, which means there are new sources visiting $v$, $v$ should be added to the next frontier. To avoid duplication in the next frontier, only the one that can successfully set 
$r[v]$ to $i$ (\cref{line:ccbfs_cas}) by \textsc{compare\_and\_swap} will put $v$ to the next frontier.  
Note that we can further benefit from the directional optimization that is commonly used in parallel BFS.  
Additional details about the directional optimization are given in \iffullversion{\cref{sec:parallel_BFS}.}\ifconference{our full version paper.}

%by a compare-and-swap. 

\hide{
At most $c+1$ such \bitvector{s} can cover $S$ if $S$ is reachable to $v$.  
We use an array of $d+1$ \bitvector{s} $Q_v$ for each vertex $v$ to maintain such subsets of $S$,  where $Q_v[i]$ stores the vertices in $S$ that visits $v$ in the $\delta_{\min}+i$ round. 
In other words, for $s\in Q_v[i]$, there is path from $s$ to $v$ with length $\delta_{\min}+i$. 
Note that $Q_v[i]$ is not equal to $S_i[v]$, and the relations between $Q_v[i]$ and $S_i[v]$ is that $S_i[v]\subseteq Q_v[i]$ and $Q_v[i]\subseteq \bigcup\limits_{j=0}^{i} S_j[v]$. We will compute $S_i[v]$ in the post processing using $Q_v$. $Q_v[i]$ are initialized as $\emptyset$ for $v\in V$ and $i\in [0,c+1]$. 
We maintain a counter array $\cnt[\cdot]$ to store the number of rounds a vertex has been visited (put in the frontier). It is also a pointer for $Q_v$ that $Q_v[\cnt[v]]$ points to the last element in $Q_v$ that we write. $\cnt[v]$ are initialized as $0$ for all $v\in V$.
We use $\delta[v]$ to maintain the distance of $v$ when it is last processed, which is the round number when it is processed plus one. $\delta[v]$ is initialized as $\infty$ for all $v\in V$.
}

%It first perform \edgemap as we mentioned above to maintain subsets of $S$ and temporary furthest distance from vertex in $S$ to each $v$. (\cref{line:batch_edgemap_begin} to \cref{line:batch_edgemap_end}), then it processes the subsets and temporary distance to compute the desired representation (\cref{line:batch_post_begin} to \cref{line:batch_post_end}).

\hide{
At the beginning, we put vertices in $S$ to the first frontier, and set their distance as 0 and the first subsets of $S$ that reach them as themselves. In round $i$, we will first increase the $cnt$ for vertices in the frontier by one. Then, for all $v\in V$, $Q_v[cnt[i]]$ is point to a new \bitvector that is $\emptyset$, and $Q_v[cnt[i]-1]$ stores the vertices that can reach $v$ in the last round. In the current round, $edge\_f(u,v)$ want to propagate the vertices reach $u$ in the last round to $v$ in the current round by $\textsc{fetch\_and\_or}(\&Q_v[cnt[v]], Q_u[cnt[u]-1])$. If $u$ successfully set any bit in $Q_v[cnt[v]]$, then $u$ will try to set the distance of $v$ to $i+1$ by $\textsc{compare\_and\_swap}$, if it succeeds, $u$ will take the responsibility to put $v$ to the next frontier (\cref{line:edge_begin} to \cref{line:edge_end}). As we mentioned before, $cond\_f(v)$ will check whether a vertex has been visited for $c$ rounds, if not, we will still process it; otherwise, we will not process it anymore. When the frontier become empty, the \edgemap is done.
}

\hide{
\myparagraph{Post Processing} After the main process of \batchBFS{}, we use postprocessing
to obtain the final \compactdis{s} 
$\langle \Ssubset{S}{v}{0..d}, \Sdist{v}\rangle$ for each vertex $v\in V$. 
Recall that at this point, we have obtained $dis[v]$, which is the latest round where $v$ has been visited,
$\Ssubset{Q}{v}{i}$, which is the set of sources that visits $v$ in round $i$,
and $cnt[v]$, which is the number of rounds $v$ has been visited. 
%Therefore, $\delta[v]$ is essentially the distance between $c$ and the farthest source in $S$,
%and $\delta_{\min}[v]$ can be computed by $\delta[v]-cnt[v]+1$. 
Therefore, the first round that visits $v$ (i.e., the shortest distance from any $s\in S$ to $v$) is 
$\Sdist{v}=\dis[v]-\cnt[v]+1$.
Finally, we obtain all $\Ssubset{S}{v}{0..d}$ from $\Ssubset{Q}{v}{0..d}$.
Based on the discussions above, $\Ssubset{S}{v}{i}$ is all sources in $\Ssubset{Q}{v}{i}$ that have not appeared in $\Ssubset{Q}{v}{0..i-1}$,
i.e., round $i$ is the first time that this source visits $v$.
Therefore, we can compute $\Ssubset{S}{v}{i}$ in the for-loop on \cref{line:batch_post_begin}.
We use $\OLD$ to store the union of $\Ssubset{Q}{v}{0..i-1}$. 
Then $\Ssubset{S}{v}{i}$ is the difference of $\Ssubset{Q}{v}{i}$ and $\OLD$. 
Note that if $S$ can reach $v$, 
the last subset $\Ssubset{S}{v}{d}$ should be $S-\OLD$. 
However, if $OLD=\emptyset$, $\Ssubset{S}{v}{d}$ is set to $\empty$ because $S$ is not reachable to $v$.

\hide{Then we do post processing to compute $\delta_{\min}, S_{0,...,c}$ by $\delta$ and $Q_v[\cdot]$. $\delta_{\min}[v]$ is the smallest shortest distance from vertices in $S$ to $v$, and $\delta[v]$ stores the distance when $v$ is processed for the $c$ times, so $\delta_{\min}[v]$ can be easily computed as $\delta[v]-c$. Recall that $Q_v[i]$ includes $S_i[v]$ completely and some other vertices in $\bigcup\limits_{j = 0}^{i-1} S_j[v]$. To recover $S_i[v]$ from $Q_v[i]$, we need to substract the vertices that are in $\bigcup\limits_{j = 0}^{i-1} S_j[v]$ from $Q_v[i]$. We keep an \bitvector $\OLD$ that maintains all the vertices shown so far, and compute $S_i[v]$ by $Q_v[i]-OLD$. Note that if because $S_{0,...,d}[v]$ can cover $S$, so the last subset $S_d[v]$ is computed by $S-\OLD$, but if $\OLD$ is $\emptyset$, $S_d[v]$ is set to $\empty$ because  $S$ is not reachable to $v$.  }

% Note that $Q[v][i]$ may contains vertices that have visited $v$ earlier than the $i+1$ round, so that $Q[v][i]$ is not equal to $S_i[v]$. $S_i[v]$ is computed by $Q[v][i]-\cup_{j<i}{Q[v][j]}$ when $Q$ are finish computed. Besides, the last subset $S_c[v]$ is computed by $S-\cup_{i<c+1}S_i[v]$ when the first $c$ subsets are computed. Suppose the queue supports three operations, $\textsc{insert}$, $\textsc{top}$ and $\textsc{second\_top}$. $\textsc{top}$ will return the last element in the queue, $\textsc{second\_top}$ will return the second last element in the queue. The pseudocode is shown in \cref{alg:batch_BST}. 
}

The efficiency of the algorithm relies on using bit-operations to compute the union and difference of two \bitvector{s}, stated below.

\begin{lemma}\label{lemma:word_cost}
  Given the \bitvector subsets $S_1$ and $S_2$ of a set $S$ with size $k$, we can compute the \bitvector representation of $S_1\cup S_2$, $S_1 \setminus S_2$ in $O(k/w + 1)$ work and $O(\log (k/w)+1)$ span, where $w$ is the word length. 
\end{lemma}

\begin{proof}
  A \bitvector with $k$ bits needs $\lceil k/w \rceil$ words. These $\lceil k/w \rceil$ words can be processed in parallel.  
  With the constant cost for bitwise or/not operations for a single word, the work and span for $S_1\cup S_2$,  $S_1\setminus S_2$ are $O(k/w+1)$ and $O(\log (k/w)+1)$.
\end{proof}

We now show the cost analysis of the \batchBFS{} algorithm. 

\begin{theorem}\label{theorem:batchBFScost}
Given a set $S$ of $k$ vertices with diameter $d$, we can compute the 
%representation of a shortest distance 
\compactdis{} 
from $S$ to every vertex in $V$ in $O(dm(k/w+1))$ work and $O((D+d)\log n)$ span.
\end{theorem}

\begin{proof} 
  We will analyze the work and span for the traversal stage, which dominates the cost of the initialization stage. The traversal consists of two phases: the first phase processes vertices in the frontier, and the second phase processes edges from the frontier. The cost of processing a single edge and a vertex is asymptotically the same, as it involves applying a constant number of set operations, which results in $O(k/w+1)$ work and $O(\log(k/w)+1)$ span, as described in \cref{lemma:word_cost}.
  Assuming $n < m$, the cost of traversal is primarily determined by edge processing. Therefore, the rest of the proof will focus on analyzing the cost of edge processing. 
    
  For the work, since each vertex is in the frontier for at most $d$ times, each edge is processed at most $d$ times, leading to a total work of $O(dm(k/w+1))$.
  For the span, recall that $D$ is the diameter of the graph, and there are at most $D+d$ rounds in the traversal stage. The span for each round is $O(\log(k/w)+1+\log n)$, which accounts for the cost of generating $O(m)$ parallel tasks (costing $O(\log n)$ in the binary fork-join model) and the cost of processing a single edge (costing $O(\log(k/w)+1)$ as shown in \cref{lemma:word_cost}). Since $k/w$ is smaller than $n$, the span for each round simplifies to $O(\log n)$. Thus, the total span for the $D+d$ rounds of the traversal stage is $O((D+d)\log n)$. Therefore, the total work and span for our parallel \batchBFS{} are $O(dm(k/w +1))$ and $O((D+d)\log n)$, respectively.
\end{proof}

If we take $k=\Theta(w)$, such that each \bitvector{} fits within a constant number of words, the work simplifies to $O(dm)$ and the span remains $O((D+d)\log n)$, which matches the work and span of a single BFS. Since $w=\Omega(\log n)$, this means that we can compute $O(\log n)$ more BFSs with asymptotically the same cost.

\hide{
In round $i$, $u$ records all the sources in $S$ that reach it before round $i$ by a \bitvector{} $\Sseen[u]$. 

When $u$ visits its neighbor $v$, $u$ propagates this \bitvector{} to $v$ by taking bitwise $\textsc{OR}$ operation with the \bitsubset{} $\Snext[v]$ representing the vertices reaching $v$ in round $i$ (\cref{line:ccbfs_fao}). 
If $\Sseen[v]$ stores all the sources that visited $v$ before round $i$, and $\Snext[v]$ stores the sources visiting $v$ including round $i$, then the subset $\Snext[v]-\Sseen[v]$ contains the vertices that visit $v$ for the first time and their distance to $v$ is $i$ (\cref{alg:batch_BST}: \cref{line:diff_label}).  

According to \cref{corollary:distance_range}:  all the vertices in $S$ will visit $v$ at least once during round $\Sdist{v}$ to round $\Sdist{v}+d$,
in other words, all the vertices will be put into the frontier for at most $d+1$ times. 
Therefore, we need to know the number of times $v$ has been put in the frontier. 
When $v$ has been put in the frontier for $d+1$ times, since all sources in $S$ must have already visited $v$, we do not need to process $v$ anymore. 
Otherwise, we will process $v$ and put it to the next frontier since other sources in $S$ may visit $v$ in the future. 
When $v$ gets the newly visited sources represented in \bitsubset{}, we needs to know the index in $\Ssubset{S}{v}{0..d}$ to store them. By the definition, the index is the distance larger than $\Sdist{v}$, which can be computed by $i-\Sdist{v}$. $\Sdist{v}$ is set to the round number when $v$ is put in the frontier for the first time.
The pseudocode of our \batchBFS is shown in \cref{alg:batch_BST} and it has two stages: initialization, and traversing. 
}

%% file: figs_algs/fig_bitwise.tex
\begin{figure}
    \centering
    \includegraphics[width=\columnwidth]{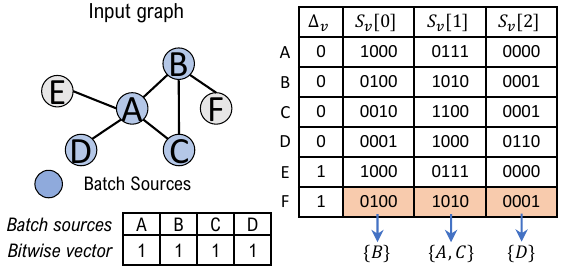}
    \caption{\small \textbf{Illustration of bitwise representation.}  
    The batch set $S$ is $\{A,B,C,D\}$. 4-bit \bitvector{s} are used to represent subsets of $S$. 
    % $A$ is represented by the left most bit, and $D$ is represented by the right most bit. 
    $\Sdist{v}$ is the smallest shortest distance from any vertex in $S$ to $v$. 
    % The subset $S_i(v)$ is defined as $\{s\in S | \delta(s,v) = \delta_{min}(v)+i\}$. 
    The subset $\Ssubset{S}{v}{i}$ is defined as $\{s\in S | \delta(s,v) = \Sdist{v} +i\}$.
    }
    \label{fig:bitwise}
\end{figure} 

%% file: figs_algs/alg_batchBFS.tex
\begin{algorithm}[t]
  \small
  \caption{\BatchBFS search from $S$\label{alg:c-bfs}}
  \SetKwProg{myfunc}{Function}{}{}
  \SetKwFor{parForEach}{ParallelForEach}{do}{endfor}
  \SetKwFor{mystruct}{Struct}{}{}
  \SetKwFor{pardo}{In Parallel:}{}{}
  \SetKwInOut{Maintains}{Maintains}
  \SetKwIF{If}{ElseIf}{Else}{if}{}{else if}{else}{end if}%
  \KwIn{\\
  A graph $G=(V,E)$, a cluster $S\subseteq V$ with diameter $d$} 
  \KwOut{\\
  \compactdis{s} $\langle \Ssubset{S}{v}{0..d},\Sdist{v}\rangle$ for all $v\in V$.
  }
  \Maintains{\\
  \noindent$i$: the current round number, initialized to 0\\
  \noindent$\Sseen[\cdot],\Snext[\cdot]$: array of \bitvector for each $v\in V$\\
  \noindent$r[v]$: the lastest round $v$ is in the frontier\\
  \noindent$\ff_i$: frontier vertices in round $i$
  }
  \DontPrintSemicolon
  \tcp{Initialization}
  \parForEach{$v\in V$}{
    $\Sseen[v]\gets\emptyset$, $\Snext[v]\gets\emptyset$\\
    $\Sdist{v}\gets \infty$\\
    $r[v]\gets \infty$
  }
  \lFor{$s\in S$}{ 
    $\Sseen[s]\gets \{s\}$
  }
  $i\gets 0$\\
  $\ff_0 \gets S$\\
  \tcp{Traversing}
  \While{\upshape$\ff_i \neq \emptyset$}{
    \parForEach{\upshape$u\in \ff_{i}$\label{line:batch_edgemap_begin}}{
      $\Snew \gets \Snext[u] \setminus \Sseen[u]$ \label{line:diff_label}\\
      \lIf{\upshape$\Sdist{u}=\infty$ \textbf{then}}{$\Sdist{u}\gets i$ \label{line: update_Sdist}}
      $\Ssubset{S}{u}{i-\Sdist{u}}\gets \Snew$ \label{line: updateSset}\\
      $\Sseen[u] \gets \Sseen[u] \cup \Snew$
      \label{line:batch_edgemap_first}
    }
    \parForEach{\upshape$u \in \ff_{i}$ \label{line:ccbfs_edgemap_begin}}{
      \parForEach{\upshape$v\in N(u)$ and $i-\Sdist{v}<d$ \label{line:ccbfs_cond}}{
        \If{\upshape\textsc{Fetch\_And\_Or}($\Snext[v], \Sseen[u]$)\label{line:ccbfs_fao}}{
        \If{\upshape\textsc{compare\_and\_swap}$(r[v], r[v], i)$\label{line:ccbfs_cas}}{
          $\ff_{i+1} \gets \ff_{i+1}\cup \{v\}$
          \label{line:ccbfs_edgemap_end}
          }
        }
      }
    }
    $i\gets i+1$\\
  }
  \Return $\langle\Ssubset{S}{v}{1..d}, \Sdist{v}\rangle$ for all $v\in V$
  \end{algorithm}

%% file: application.tex
\section{Applications}\label{sec:applications}
%In this section, we show that \batchBFS (\ccbfs) can efficiently answer distance queries on unweighted graphs, 
%and thus accelerate graph analytic applications relevant to distance queries. 
%especially in the context of \emph{distance oracles}. 
In this section, we show two applications that can benefit from our new parallel \BatchBFS algorithm. %can efficiently answer distance queries on unweighted graphs, 
Both applications are distance oracles (DO) for unweighted graphs. 
A \defn{distance oracle (DO)} is an index designed to
answer the shortest distances between two vertices on a graph. 
%that can help answering the shortest distance between two query vertices quickly. Although we can compute the distance for each query by using a single BFS or Dijkstra's algorithm, they take more than a second for large graphs, which is too slow to use as a building block of these applications that requires response time to be microseconds. In order to answer distance queries much more quickly, we can precompute some distances and store them in index. 
Although such a query can always be answered by computing the distance on the fly (e.g., running a BFS from one of the query vertices),
this can be inefficient for applications requiring low latency or requesting multiple queries. 
A distance oracle aims to store information generated during preprocessing 
to accelerate the distance queries. 

From the perspective of accuracy, distance oracles can be classified into approximate distance oracles (ADO) and exact distance oracles (EDO). 
ADOs may not answer the accurate distance but are cheaper in preprocessing time, query time, and index space, and thus scale to large graphs. 
EDOs always give the exact distance but can be more expensive to compute. 
Therefore, EDOs are usually used in applications that are on small graphs but more sensitive to accuracy. 
%Our \batchBFS can be applied to both ADO and EDO. %to take advantage of parallelism and further accelerate the performance. 
In this section, we will introduce how to apply our \batchBFS to the two existing distance oracles: 
a 2-hop labeling-based EDO and a landmark labeling-based ADO. 
Since both applications work on undirected graphs, we assume the graph to be undirected in this section. Our \batchBFS{} algorithm works for general directed graphs. 

\subsection{An EDO based on 2-Hop Labeling} 
Here, we consider the EDO constructed by Akiba et al.~\cite{akiba2012shortest}, called \emph{2-hop labeling}, and we apply the idea of \ccbfs{} to this algorithm. 
Their original algorithm is sequential, and we refer to it as the \emph{AIY} algorithm. 
We can replace the component for 2-hop labeling in their algorithm with our parallel \ccbfs{} to achieve better performance. 
\iffullversion{For completeness, we describe their algorithm in \cref{sec:2hopLL}. }\ifconference{For completeness, we describe their algorithm in the full version paper. }
In the experiments, we compared our parallel \ccbfs{} with the component of \ccbfs{} in their sequential code.\iffullversion{We also parallelized the entire algorithm for 2-hop labeling using our parallel \ccbfs{}, and present a comparison in \cref{sec:exp:edo}.}\ifconference{We also parallelized the entire algorithm for 2-hop labeling using our parallel \ccbfs{}, and present a comparison in our full version paper.}

\hide{
Recall that a \defn{distance query} gives the shortest distance between two vertices in a graph, which is one of the most fundamental operations on graphs. 
For example, on social networks, distance between two users indicates the closeness, and is widely used in socially-sensitive search to help users to find more related users or contents \cite{}, or to detect influential people and communities \cite{}. On web graphs, distances between web pages usually indicate their relevance, and vertices with smaller distances can be web pages more related to the currently visiting web page \cite{}.
Some other applications include ..\cite{}.

Since we cannot exhaust all applications, in this paper, we focus on applications on distance oracles, and show that CC-BFS can be an effective routine in accelerating them.
Specifically, our \batchBFS can accelerate distance oracle-based applications. \defn{Distance oracles} are prestored distance index that can help answering the shortest distance between two query vertices quickly. Although we can compute the distance for each query by using a single BFS or Dijkstra's algorithm, they take more than a second for large graphs, which is too slow to use as a building block of these applications that requires response time to be microseconds. In order to answer distance queries much more quickly, we can precompute some distances and store them in index. 
% In particular, applications such as socially-sensitive search or context-aware search should have low latency since they involve real-time interactions between users, while they need distances between a number of pairs of vertices to rank items for each search query. Therefore, distance queries should be answered much more quickly, say, microseconds. 
% Generally, in order to answer distance queries with short response time, people use some space to store distance.  
From the perspective of accuracy, distance oracles can be classified into approximate distance oracles (ADO) and exact distance oracles (EDO). Our \batchBFS can be applied to both ADO and EDO to bring in parallelism and further accelerate the performance. 
In this section, we will introduce how to apply our \batchBFS to the two existing distance oracles. More specifically, we introduce a landmark labeling-based ADO in \cref{sec:approxLL} and a 2-hop labeling-based EDO in \cref{sec:2hopLL}.
% , and an all pairs shortest distance algorithm (a specifical type of EDO) in \cref{sec:APSP}.
}

%\subsection{An ADO based on Landmark Labeling}
%\label{sec:approxLL} 

\input{figs_algs/alg_LL_simple.tex}

% \input{figs_algs/alg_LL_cluster.tex}

%Exact distance oracles usually need a large index to recover the exact answer for all vertex pairs, so they are usually not able to run on large graphs because of the large memory usage and construction time.  
%Approximate distance oracles sacrifice accuracy for smaller indexes, so that they can run on large graphs. 
\subsection{An ADO based on Landmark Labeling} \label{sec:approxLL} 
%Our first application is an ADO based on landmark labeling~\cite{tang2003virtual,vieira2007efficient,gubichev2010fast,tretyakov2011fast,qiao2012approximate}. 
In this paper, we mainly focus on this application that can benefit from \ccbfs{}, which we believe is new.
The application is an ADO based on landmark labeling~\cite{tang2003virtual,vieira2007efficient,gubichev2010fast,tretyakov2011fast,qiao2012approximate}. 
As mentioned, ADOs sacrifice accuracy to get lower running time and index space.
Here, as is common, we assume a one-sided error, such that the distance reported by the ADO cannot be smaller than the actual distance.
%To measure the loss in accuracy, we use \emp{distortion} $\psi$ as the expected ratio that the answer from an ADO  (denoted as $query(u,v)$ for $u,v\in V$)
%is off the true distance $\delta(u,v)$, i.e., $\psi = E_{\forall u,v\in V} [({query(u,v)-\delta(u,v)})/{\delta(u,v)}$. 
To measure the loss in accuracy, for a distance query on $u,v\in V$, 
we use \emp{distortion} $\psi(u,v)$ as the ratio between the answer from an ADO  (denoted as $\query(u,v)$)
over true distance $\delta(u,v)$, i.e., $\psi = \query(u,v)/\delta(u,v)$. 
We define the distortion of an ADO as the average distortion over all pairs.  
% \guy{This only makes sense if the error is one sided, and hence I added the one sided assumption above.}
As $\psi$ is always greater than 1, for simplicity, we use $\epsilon$ to describe the distortion, where $1+\epsilon=\psi$.
%how far the answer of a query (denoted as $query(u,v)$) is larger than the true distance $\delta(u,v)$ in expectation, i.e. $\psi = E_{\forall u,v\in V} \frac{query(u,v)-\delta(u,v)}{\delta(u,v)}$. 

%The most straigh-forward approach is the \defn{landmark-based approach}\cite{}. 
\emph{Landmark labeling (LL)} is one of the widely-used approaches of ADOs and probably the simplest. 
The basic idea is to select a subset $H$ of vertices as landmarks, and precompute the distances between each landmark $h\in H$ and all the vertices $u\in V$. 
When the distance between two vertices, $u$ and $v$, is queried, $\mbox{query}(u,v)$ answers the minimum $\delta(u,h)+\delta(h,v)$ over all the landmarks $h\in H$ as an estimation.
We show the high-level idea of LL in \cref{alg:simpleLL}. 
%and the estimation is always one sidded. % (i.e., no smaller than the true distance). 
%Although the landmark-based approach is simple, its low distortion on complex graphs makes it widely used \cite{}. 

% Generally, the distortion for each query depends on whether actual shortest paths pass nearby landmarks. 
Generally, the distortion of a query depends on how far the actual shortest paths are from their nearest landmark. 
If the actual shortest paths contain any landmark vertex, the query answer is equal to the true distance. 
Therefore, adding more landmarks can decrease the distortion, but it also needs more space and time to store and compute the index.
% This property motivates us to apply our \batchBFS to landmark labeling. 
In this paper, we propose to use \ccbfs{} to optimize LL. 
In particular, we select clusters of vertices as landmarks instead of single vertices.
As discussed, a \ccbfs{} on a cluster of size $O(\log n)$ 
has costs (both time and space)  asymptotically the same as running BFS on one vertex. 
In this way, we can select $w$ times more landmarks with asymptotically the same cost as the plain LL. %according to \cref{theorem:batchBFScost}. 
We note that the quality of the $w$ landmarks in one cluster may not be as good as choosing them independently, 
since they are highly correlated. 
%in a cluster are correlated so they are not as good as the same amount of independent landmarks. 
However, we experimentally observe that with the same memory limit, using \batchBFS in landmark LL significantly improves the performance both
in distortion and preprocessing time (see \cref{sec:exp_appx}). 
% Therefore, with the same memory usage for the index, applying \batchBFS may increase the precision, as well as accelerate index construction. 
% Another factor people found affecting the precision is the selecting vertices sit more sentral in the graph is better than select random vertices. A typical way to select landmarks is to prioritize vertices with higher degrees and choose top $k$ as landmarks. 

%In the following of this section, we will first introduce how to apply our \batchBFS to landmark labeling in two steps: selecting landmarks and answering queries. 
%Then we will introduce two optimizations to improve the query in time and distortion.
% Therefore, by selecting central vertices as landmarks, the accuracy of estimates becomes much better than selecting random landmarks \cite{}. Typically, people rank the vertices by certain criterions, and choose the top ranking vertices as landmarks. Practical criterions including vertex degrees, estimated centralities and xxx \cite{}. 

%\myparagraph{Applying \batchBFS to Landmark Labeling.} 
% To applying \batchBFS to landmark labeling, we first need to figure out how to select landmarks in clusters, and then apply the \batchBFS to the selected clusters and get the shortest distances in the compact representation mentioned in \cref{sec:bitwise}. Finally, we need to figure out how to answer the queries from the compact distances. 
To apply \ccbfs{} to LL, we need two subroutines: 
1) selecting landmarks in clusters with low distortion, and 
2) answering the queries from the \compactdis{s} computed by \ccbfs{}.

\myparagraph{Selecting Landmarks in Clusters.} The landmark selection is crucial in LL as it affects the query quality. 
%Selecting vertices siting more in the central of a graph is better than selecting vertices randomly, because more shortest paths passing through central vertices than random vertices. 
A typical way is to prioritize vertices with high degrees~\cite{li2019scaling,Potamias09,akiba2013fast}. 
%We also first choose the vertex with the highest degree as a landmark. 
We also employ this approach. 
In this step, we aim to identify clusters with a specified size $k=w$ and diameter $d$. Within these constraints, the selection of landmarks should prioritize vertices with higher degrees.
To find a cluster, we first identify the vertex with the highest degree. 
Among all its $\lfloor d/2 \rfloor$-hop neighbors, we then select $k-1$ additional vertices with the highest degrees. 
For instance, if $k=64$ and $d=2$, we select the vertex with the highest degree and $63$ of its neighbors with the next highest degrees to create the first cluster.
%We start with choosing the vertex with the highest degree along with its $w-1$ neighbors as the first cluster of landmarks. 
%If the vertex has more than $w-1$ neighbors, we will choose those with the highest degrees. 
%In particular, to enable some query optimizations (which will be introduced below), 
%we always use a vertex and up to $w-1$ of its neighbors as a cluster, where $w$ is the word length. 
%We start with choosing the vertex with the highest degree along with its $w-1$ neighbors as the first cluster of landmarks. 
%If the vertex has more than $w-1$ neighbors, we will choose those with the highest degrees. 
%\letong{Mention how to select 2-hop neighbors when there is not enough neighbors.}
Once we select a cluster, we will mark all the vertices in the cluster, and will not select them again. 
We repeat this process until we select $r$ clusters (giving $rw$ landmarks in total).
%More generally, the user can prioritize the vertices with other criteria. 
One can also select landmarks using other heuristics~\cite{Potamias09,li2017experimental}. 
% No matter what criteria are used, as long as an order of vertices is given, we can always first choose vertices and their neighbors with higher priorities. 
%After selecting clusters, applying \batchBFS to the selected clusters are straightforward. 
After selecting clusters, we apply \ccbfs{} to all selected clusters using \cref{alg:c-bfs}. %the algorithm in \cref{sec:bitwise}. 
By doing this, we obtain the \compactdis{s} between each vertex and each cluster.  
%The shortest distances from the clusters are computed in the \compactdis{} mentioned in \cref{sec:bitwise}. 

\myparagraph{Answering Queries with Clustered Landmarks.} To extend the original landmark idea to work with \ccbfs{}, 
we need to show how to use the \compactdis{s} to answer the queries.  
A query answers the shortest distances between two vertices through any landmark. 
When all landmarks are independent vertices, $\query(u,v)$ returns the smallest $\delta(u,l)+\delta(l,v)$ over all landmark vertices $l\in L$. 
%When the landmark unit is a single vertex,  $query(u,v)$ returns the smallest $\delta(u,l)+\delta(l,v)$ over all landmark vertices $l\in L$. 
In our case, the landmarks are grouped into clusters. 
Therefore, we first compute, within each cluster $S$, the smallest value $\delta(u,s)+\delta(s,v)$ for all $s\in S$. 
Then we will take the minimum among all clusters. 
%$query(u,v)$ returns the smallest distances passing through a cluster $\min_{s\in S}(\delta(u,s)+\delta(s,v))$ over all clusters $S\in L$. 
%we will take the minimum value among all clusters,

The problem boils down to finding $\delta(u,s)+\delta(s,v)$ for all sources $s$ in a given cluster $S$. 
Recall that for each cluster $S$, both vertices $u$ and~$v$ have obtained their \compactdis{} from \ccbfs{}, 
denoted as $\langle S_u[0..d],\Sdist{u}  \rangle$ and $\langle S_v[0..d],\Sdist{v} \rangle$. 
The possible shortest distances passing through $S$ is in the range $[\Sdist{u}+\Sdist{v}, \Sdist{u}+\Sdist{v}+2d]$.  For two \bitvector, $\Ssubset{S}{u}{i}$ and $\Ssubset{S}{v}{j}$, if their intersection is not empty, it means there is a path connecting $u$ and $v$ with distance $\Sdist{u}+\Sdist{v}+i+j$ passing through any source vertex in the intersection.  
We check the intersections from the lowest possible distances (e.g., $\Ssubset{S}{u}{0} \cap \Ssubset{S}{v}{0}$) to higher possible distances, until we find a nonempty intersection, and return their distance sum. The complexity of the query grows in a quadratic manner as $d$ grows. For the simplest case, $d=2$, we only need to check three intersections for each cluster to get the answer.  

\hide{
We will use an optimization from~\cite{akiba2013fast}. 
In particular, we always choose a cluster as a vertex and its neighbors, forming a star-like cluster. 
Let the center of the cluster be $s^*$.
Note that we only need to find the lowest $\delta(u,s)+\delta(s,v)$ among all $s\in S$, which is upper bounded by
the value $\delta(u,s^*)+\delta(s^*,v)$ using the center $s^*$. 
Since $s^*$ is the center of the star, for any $s\in S$ and $v\in V$, the distance $\delta(s,v)$ and $\delta(s^*,v)$
differ by at most 1. 
This means that we do not need to process any source $s\in S$ with $\delta(s,v)>\delta(s^*,v)$---even though $s$ can be closer to $u$ than $s^*$ (by at most one), the total distance to $u$ and $v$ from $s$
will not be better than $s^*$, and therefore we can skip them. 
Based on this idea, we only store two \bitvector{s} in a \compactdis{}: the one containing $s^*$, and the one that is one closer than $s^*$. \letong{extend this optimization to more general clusters. The center of the star $s^*$ provides an distance, we first find the subsets of $u$ and $v$ that $s^*$ appears in, and then testing the other combinations that may generate distances smaller than $\delta(u,s^*)+\delta(s^*,v)$}.
}

% We can easily apply this optimization to our \batchBFS mainly by changing the \condf{$(v)$}: 
% Instead of checking whether $v$ has been added to the frontier for $d$ times, 
% we change \condf{$(v)$} to check whether $s^*$ has visited $v$. 
% In this way, once $s^*$ reaches $v$, 
% we no longer need to compute the distances from other sources to $v$. 

We have shown another optimization, bidirectional searching, for answering queries that can reduce distortion without much more overhead in querying time. \ifconference{Since this optimization is independent to \ccbfs itself, due to the page limit, we put the description and experiments in our full version paper.}
\iffullversion{Since this optimization is independent with \ccbfs itself, due to the page limit, we put both the description and experiments in \cref{sec:biBFS}.}
% Then, $\delta(v)$ stores distance when the last time $v$ is in the frontier, which is $\delta(s^*, v)$. We only need to store $S_0[v]$ and $S_1[v]$ for each vertex $v\in V$, where $S_1[v]$ stores the subset of $S$ that has the same distance to $v$ as $\delta(s^*, v)$ and $S_0[v]$ stores the subset of $S$ that has distance shorter than $\delta(s^*,v)$.  
%Certainly, there are more details in the implementation to maintain the correctness of the boundary conditions, but overall, it is simple to apply this optimization to our \batchBFS framework. 

\hide{
vertex $u$ and $v$ store the distances from $S$ to them by $\langle S_0[u], ... , S_c[u], \delta(u) \rangle$ and $\langle S_0[v], ..., S_c[v], \delta(u) \rangle$. To answer the shortest distance between $u,v$ through a path passing through any vertex in $S$, we take the intersection on all the combination of their subsets, e.g. $S_i[u]\cap S_j[v], \forall i,j\leq c$. The distance of paths passing through the vertex in intersection $S_i[u]\cap S_j[v]$ is $\delta(u)+\delta(v)+i+j$. The non empty intersection with the smallest distance is the shortest distance passing through $S$. 
% The number of combinations is $O(c^2)$, but the optimization mentioned below can reduce the number of combinations need to be check to a quarter of the original. 
To reduce the cost for query, we use an optimization from \cite{akiba2013fast}.
We will skip the details here and introduce the implementation of query in the following paragraphs. 
}

% We can apply our \batchBFS by choosing a \batch as a landmark instead of single vertex. When query the shortest distance between two vertices, it answers the minimum distance of paths passing over all the vertices in selected \batch{s}.

\hide{
\myparagraph{Optimizations for the $query$ function.} We introduce two optimizations for answering queries. The first one reduces the number of interactions need to check for answering a query, thus reduce the time. The second one is to deal with close pairs whose distances are not correctly answered by landmarks, thus reduce the distortion.

Note that, answering a query does not need to know the distances passing through all the vertices in a \batch, but only the smallest distances passing any vertex in the \batch. Therefore, the algorithm in \cite{akiba2013fast} does not store the distances from all the vertices in a cluster to any vertex $v$. 
The idea is to use a vertex $s^*$ and its neighbors as a \batch $S$, then use the path between $u$ and $v$ passing throuth $s^*$ as an upper estimation for the shortest distance passing through $S$. 
% They select a vertex $s^*$ in each \batch $S$ as a center. 
For two vertices $u$ and $v$, the distance $\delta(u,s^*)+\delta(s^*,v)$ is no larger than the shortest distances passing any vertices in $S$, i.e. $\min_{s\in S} \delta(u,s)+\delta(s,v)\leq \delta(u,s^*)+\delta(s^*,v)$.  
Therefore, for each vertex $v$, they only store the distances from vertices in $S$ whose distance to $v$ is no larger than $\delta(s^*,v)$.
% $\{s\in S | \delta(s,v) \leq \delta(s^*,v)\}$.  
% Note that this inequality applies for any $s^*$ in $S$. They choose a vertex and its neighbors as a \batch, so they set the vertex as the center vertex $s^*$. 
Besides, the distances from $s^*$ to other vertices in $S$ are only one, so the distance from any vertex in $S$ to $v$ are in range $[\delta(s^*,v)-1, \delta(s^*,v)+1]$. We only need to store two subsets of $S$, one containing $s^*$ ($S_1[v]$), and potentially the one containing vertices with distance $\delta(s^*,v)-1$ ($S_0[v]$).  Together with the distance $\delta(s^*, v)$, we can answer the shortest distance between any two vertices passing through the \batch $S$ by checking only three intersections: $S_0[u]\cap S_0[u]$, $S_0[u]\cap S_1[v]$ and $S_1[u]\cap S_0[v]$. More generally, if the center vertex $s^*$ has distances to other $s\in S$ no more than $c/2$, it can reduce the number of subsets to store from $c+1$ to $c/2+1$, and thus reduce the number of combinations of subsets from two vertices to one quarter. 

We can easily apply this optimization to our \batchBFS mainly by changing the $cond\_f(v)$. Instead set the $cond\_f(v)$ as whether $v$ has been added to the frontier for $c$ times, we change it to be whether $s^*$ has visited $v$. 
% Then, $\delta(v)$ stores distance when the last time $v$ is in the frontier, which is $\delta(s^*, v)$. We only need to store $S_0[v]$ and $S_1[v]$ for each vertex $v\in V$, where $S_1[v]$ stores the subset of $S$ that has the same distance to $v$ as $\delta(s^*, v)$ and $S_0[v]$ stores the subset of $S$ that has distance shorter than $\delta(s^*,v)$.  
Certainly, there are more details in the implementation to maintain the correctness of the boundary conditions, but overall, it is simple to apply this optimization to our \batchBFS framework. 
}

%% file: figs_algs/alg_LL_simple.tex
\begin{algorithm}[t]
  \small
 \caption{Framework of Landmark Labeling\label{alg:simpleLL}}
%  \KwIn{A graph $G=(V,E)$ and a list of landmarks $s\in S$}
 \SetKwFor{parForEach}{parallel\_for\_each}{do}{endfor}
 \SetKwInOut{Maintains}{Maintains}
 \DontPrintSemicolon
 The algorithm maintains $L[\cdot][\cdot]$ as the index with size $n\times |S|$. 
 $L[v][i]$ is the distance between vertex $v$ and landmark $h_i$, initialized as $\infty$\\
 \myfunc{\upshape\textsc{Construct\_Index}$(G,H)$\tcp*[f]{$G=(V,E)$}}{
  % \comment{A $n\times |S|$ size array, $\delta[v][i]$ stores the shortest distance between vertex $v$ and landmark $s_i\in S$}\\
  \tcp{Each $h_i\in H$ is a vertex in the plain LL, and is a cluster in \ccbfs-based LL}
  \For{$h_i \in H$ }{
  \tcp{In our algorithm, we replace BFS with \ccbfs{}}
    $t[1..n] \gets \textsc{BFS}(G,h_i)$  \\
    \lForEach{$v\in V$}{
      $L[v][i]\gets t[v]$
    }
  }
  \Return $L$
 }
 \myfunc{\upshape\textsc{Query}$(u,v)$}{
  $ans \gets\infty$ \tcp*[f]{the answer of the query}\\
  \For{$i\gets 0$ to $|H|-1$}{
  \tcp{Our algorithm computes the shortest distance via all vertices in all clusters}
    $\dis \gets L[u][i]+L[v][i]$ \\
    $\ans \gets \min(\ans, \dis)$
  }
  \Return $\ans$
 }
\end{algorithm} 

%% file: experiment.tex
\section{Experiments}

% \begin{table}[htbp]
%   \centering
%   \small
%   \input{figs_algs/table0.tex}
%   \label{table:graph_info}
%   \caption{
%     The graph information.
%   }
% \end{table}

\myparagraph{Setup.}
We run our experiments on a 96-core (192 hyperthreads) machine with four Intel Xeon Gold 6252 CPUs and 1.5 TB of main memory.  We implemented all algorithms in C++ using ParlayLib~\cite{blelloch2020parlaylib} for fork-join parallelism and parallel primitives (e.g., sorting).  
We use \texttt{numactl -i all} for parallel tests to interleave the memory pages across CPUs in a round-robin fashion.

We tested 18 undirected graphs, which are either social or web graphs with low diameters. 
Graph information is given in \cref{table:microbenchmark}. 
All graphs are from commonly used open-source graph datasets~\cite{leskovec2014snap,BoVWFI,Boldi-2011-layered,nr}. 
When comparing the average running times or speedups across all the graphs, we use the geometric mean. 

\myparagraph{Baseline Algorithms.} We compare our algorithm to two existing implementations: 
1) \defn{Ligra}, the parallel BFS in the {Ligra}~\cite{shun2013ligra} library (only using thread-level parallelism),
and 2) \defn{AIY}, which is the \ccbfs{} component from Akiba et al.~\cite{akiba2012shortest} (sequential, only using bit-level parallelism). 
We also compared AIY for the 2-hop distance oracle as one of the applications. 
%The idea of the parallel BFS algorithm implemented by \edgemap was first proposed by Ligra~\cite{shun2013ligra}. In our paper, we use \defn{Ligra} to refer to the Ligra BFS algorithm implemented in the graph library \parlay~\cite{blelloch2020parlaylib}. In the exact 2-hop distance oracle, we compared to the baseline implementation in \cite{akiba2013fast}.

%We tested on 18 undirected graphs, which are either social or web graphs with low diameters (so-called ``scale-free networks''). 
%Basic information on graphs is given in \cref{table:microbenchmark}. 
%All graphs are from commonly used open-source graph datasets~\cite{leskovec2014snap,BoVWFI,Boldi-2011-layered,nr}. 
%When comparing the average running times or speedups across all the graphs, we use the geometric mean. 

%\myparagraph{Baseline Algorithms.} The idea of the parallel BFS algorithm implemented by \edgemap was first proposed by Ligra~\cite{shun2013ligra}. In our paper, we use \defn{Ligra} to refer to the Ligra BFS algorithm implemented in the graph library \parlay~\cite{blelloch2020parlaylib}. In the exact 2-hop distance oracle, we compared to the baseline implementation in \cite{akiba2013fast}. \citeauthor{akiba2013fast} implemented a sequential \ccbfs \cite{akiba2013fast}. In this paper, we use \defn{AIY} to refer to the \ccbfs implemented by \citeauthor{akiba2013fast}.

\input{exp_clusterBFS.tex}
\subsection{2-Hop Distance Oracle}
As mentioned in \cref{sec:applications}, we can replace the \ccbfs{} in Akiba et al.\ to accelerate their algorithm for an EDO.
\iffullversion{Due to the page limit, we present the results in \cref{sec:exp:edo}.} 
\ifconference{Due to the page limit, we present the results in our full version paper.}
To do this, we also need to parallelize the other parts in their algorithm.
In a nutshell, our algorithm with thread-level parallelism accelerates their algorithm by 9--36$\times$, and can also process much larger graphs than they can do. 

\input{exp_approx.tex}

\input{exp_CCBFS_d.tex}

%% file: exp_clusterBFS.tex
% \subsection{Microbenchmarks for \titlecap{\batchBFS}}
\subsection{Microbenchmarks for \BatchBFS}

\input{figs_algs/tab_graph_BFS.tex}

\input{figs_algs/fig_bfs_compare.tex}

We start with testing our \batchBFS{} (\ccbfs{}) as a building block.
We first test the simplest case where $d=2$ and $k=w=64$, where each cluster is a star (a vertex and its up to 63 neighbors). 
%Our experiment runs on 10 different clusters on each graph and the average is reported.
We present a detailed experimental study for this simple case because one of our baselines, AIY, only supports sources as star-shaped clusters. 
Another reason is that in previous work (as well as new results in this paper),
we observed that using $d=2$ gives the best overall performance for the two applications discussed in this paper. 
We present some studies about varying $d$ at the end of this subsection. 

Recall that our new \ccbfs{} benefits from two aspects: 
1) using bit-level parallelism with the idea of clustering to compute the results from $O(w)$ sources simultaneously, and 
2) using thread-level parallelism and known parallel techniques for optimizing BFS (e.g., directional optimization). 
In our test, we choose ten different clusters and report the average running time of them, as well as the plain sequential BFS as the simplest baseline. 
%Each of the ten clusters is a vertex, and its 63 neighbors, given a total size of machine word size $w=64$ and cluster diameter $d=2$. 
%Each \ccbfs{} runs in parallel, but all the ten clusters are evaluated one by one. 
The plain sequential BFS processes all 64 sources independently, 
and the running time is in the column ``Seq-BFS time'' in \cref{table:microbenchmark}. 
Our final running time of \ccbfs{} using all techniques is provided in the column ``Par-Time (s)'' in \cref{table:microbenchmark}.
To evaluate the performance gain by both techniques, we compared \ccbfs{} with both Ligra and AIY. 
AIY only supports \ccbfs{} on star-shaped clusters. 
The column ``AIY'' and  ``Ligra'' in \cref{table:microbenchmark} provide the speedups over the plain sequential BFS. 
To better illustrate the results, we show the speedups relative to the plain version ``Seq-BFS'' in \cref{fig:par_compare}. 
Essentially, the column ``AIY'' means the speedup that can be achieved by applying bit-level parallelism on a \batchBFS{} in the sequential setting.  
Similarly, the column ``Ligra'' provides the speedup that can be achieved by applying thread-level parallelism for running $k=64$ regular (non-cluster) BFS. 

%To further study the impact on the performance of each technique, we compared two related work with the plain version. 
%``AIY''\cite{akiba2013fast} and  ``Ligra''\cite{shun2013ligra} are baseline of sequential \ccbfs and parallel single BFS. 
%The column ``AIY'' and  ``Ligra'' in \cref{table:microbenchmark} provide the speedups over the plain sequential BFS. 
%The column ``AIY'' shows the speedup of applying bit-level parallelism in the sequential setting.  The column ``Ligra'' provides the speedup of applying thread-level parallelism and related parallel techniques in the single BFS setting. 

As shown in the column  ``AIY'', using clusters and bit-parallelism gets up to 18.8--32.0$\times$ improvement, which is uniform on different graphs. 
Note that here, we have $k=64$ sources, so  the maximum speedup can be 64$\times$.
Since \ccbfs{} is more complicated than the plain BFS, there are some constant overheads, resulting in an average 22.4$\times$ speedup in a sequential setting.

For applying thread-level parallelism, the improvement on different graphs varies greatly, from 3.91$\times$ (on smallest graphs) to up to 187$\times$ on a 96-core machine with hyperthreads. The benefit on certain graphs (e.g., OK, TW, and FT) is significant, which is over 100$\times$. One reason for the difference in improvement is the directional optimization. On some dense graphs, the backward step
can be applied across many of the rounds and thus save significant work
\iffullversion{(see \cref{sec:parallel_BFS} for more details).}
\ifconference{(see our full version paper for more details).}
Another reason is that there is more parallelism available on larger graphs.
This is consistent with the observations in prior work~\cite{shun2013ligra,Beamer12}.  
On average, effectively utilizing parallelism gives 27.0$\times$ speedup over plain sequential BFS.

\hide{\edgemap slightly improves the performance (1.03$\times$ on average), but the benefit on certain graphs (e.g., HW, OK and TW) is
significant. We note that the \edgemap{} framework is mainly designed to achieve better parallelism,
and naturally incur mild overhead on top of the standard sequential BFS, as is the case on many graphs. 
The benefit of using \edgemap{} in the sequential setting is from the use of the direction optimization that is more cache-friendly. 
On certain graphs, this makes the sequential version of the parallel algorithm more efficient than a standard sequential algorithm.
This is also observed in previous work~\cite{shun2013ligra, wang2023parallel}. Overall, across all graphs, \edgemap{} still enables about 3\% speedup. 
}

\hide{We then change the BFS to \ccbfs{} and show improvement over the previous version (only using \edgemap) in the column ``+Cluster'' in \cref{table:microbenchmark}. 
This is obtained by running our \ccbfs{} on one core. 
Compared to the previous version, running BFS in clusters always improves the performance, 
and the improvement (in the sequential setting) is 5.2--23.5$\times$.
This version runs 64 BFSs in a cluster with a small constant fraction of overhead
due to computing a vector of $d+2$ values in the \compactdis{} (only one distance per vertex needs to be computed in regular BFS).
%Indeed the improvement is about 
%The result roughly indicates that such overhead is about 5.2$\times$ on average in this setting.
In summary, using clusters gives 12.3$\times$ additional speedup on average. 
}

\hide{Finally, we parallelize the BFS for each cluster and obtain our final version. 
We present the speedup compared to the previous version in the column ``+parallel'' in \cref{table:microbenchmark}.
This speedup shows the benefit of using parallelism on \ccbfs{}, 
and is also the self-relative speedup of our \ccbfs{}. 
On a 96-core machine, this speedup ranges from 9.4$\times$ (on the smallest graph) to more than 60$\times$ (on most large graphs). 
We further show the scalability curve on five representative graphs in \cref{fig:bfs_scale}. 
On all the graphs, \ccbfs{} scales almost linearly to 192 hyperthreads. 
In summary, effectively utilizing parallelism gives 36.9$\times$ additional speedup for our \ccbfs{}. 
}

The columns ``Par-Time(s)'' and ``Final'' show our parallel \ccbfs{} running time and overall improvement over ``Seq-BFS''.  
Our algorithm combines the strengths of both bit-level and thread-level parallelism. 
%We can see \ccbfs that applys both bit-level and thread-level parallelism 
Our solution is always better than any of the baselines. 
Compared to the plain sequential baseline, the improvement is 94.1--1119$\times$, and 500$\times$ on average. 

For comparing the techniques and improvements of all baselines, we show a summary figure in \cref{fig:performance}. 
The time and speedup numbers are average on all tested graphs. 
Compared to Ligra, our algorithm improves the performance by 18.5$\times$ by utilizing clusters and bit-level parallelism.
Compared to AIY, our algorithm improves the performance by 22.3$\times$ by utilizing thread-level parallelism. 
%We note that for both baselines, there are certain optimizations that cannot be directly applied in our setting,
%but our algorithm still achieves competitive improvements. 
This indicates that our combination of bit- and thread-level parallelism works very well in synergy. 
Each of them still (almost) fully contributes to the performance, achieving the same level of improvement
as when used independently. 
%Similar results are achieved for different values of $d$, where the total improvement is xx times. \yihan{check if this is tue}
Therefore, we believe our work on an efficient implementation combining thread- and bit-level parallelism fills the gap in the existing 
study of both \ccbfs{} and parallel BFS.

\myparagraph{Self-relative Speedup and Scalability.} In addition to the aforementioned set of baselines, 
we further tested \ccbfs{} in the sequential setting to study the self-relative speedup (in column ``Self-Spd.''). 
The speedup numbers are from 9.44$\times$ (on smallest graphs) to more than 40$\times$ (on most large graphs). 
%Although our \ccbfs uses the directional optimization, we can not exist early in the backward traversing because of running from a cluster of sources. Thus, no graphs has self speedup over 80. 
In summary, the self-speedup of applying thread-level parallelism is 35.4$\times$ on average.

\input{figs_algs/fig_BFS_scale.tex}

  In addition to the overall self-relative speedup on 96 cores with hyperthreads, the scalability curve, as shown in \cref{fig:bfs_scale}, presents the self-relative speedup results across different number of cores. 
  The curve demonstrates that our algorithm achieves nearly linear speedup for most of the graphs, indicating efficient parallel scalability. One exception is the DBLP graph, which deviates from this pattern due to its smaller size.

\hide{
Interestingly, we find  that the improvement of parallelism
is more significant in the clustered setting (36.9$\times$) than in the non-clustered setting (23.6$\times$).
Similarly, the improvement of using clustering is more significant in parallel than 
in the sequential setting. 
These results indicate that both clustering and parallelism 
more positively impact the performance when the other technique presents,
and work very well in synergy. 
We believe that the reason is that parallel BFS (and other traverse-based graph algorithms) are mostly memory-access (I/O) bottlenecked.
Our clustering technique on BFS can save space usage and thus memory I/Os, which will amplify the benefit of parallelism on multicore platforms.
Therefore, we believe our work on a high-performance implementation of 
a combination of them fills the gap in the existing study of both \ccbfs{} 
and parallel BFS, and presents interesting observations and performance evaluations. 
}
%In later experiments, we will further show that their combination
%effectively improves the performance of multiple applications based on distance oracles. 

\hide{
\input{figs_algs/fig_BFS_scale.tex}
}

\myparagraph{Influence of Cluster Diameter $d$ in \ccbfs.} 
There are no existing implementations supporting \ccbfs with general $d$. 
Recall that our \ccbfs supports general clusters with diameter $d$ instead of star-shaped clusters (a vertex and its neighbors, where $d=2$). 
As shown in Thm.~\ref{theorem:batchBFScost}, the work is proportional to $d$. 
We tested different $d$ from $2$ to $6$ on all the graphs.
We choose 10 representative graphs and show the \ccbfs running time on clusters with different~$d$ in \cref{fig:ccbfs_d}. 
The full running time is shown in \cref{table:BFS_d} in \cref{sec:exp_d}. 
The running time increases as $d$ grows. %, especially when $d$ is small. 
%Although $d$ is usually chosen as a small constant, it still significantly affects the performance. 
A large $d$ allows for better flexibility for the shape of the cluster (the vertices can be further from each other), but significantly affects the performance. 
In \cref{sec:exp_d}, we will further show that using a large $d$ in LL incurs overhead in space and time, 
and therefore, using $d=2$ is almost always more effective in applications. 
%For most graphs, the running time grows slower when $d$ is greater than four. 

\input{figs_algs/fig_BFS_d.tex}

\hide{ 
Recall that our new \ccbfs{} benefits from three important techniques: 
1) using the \edgemap framework to achieve optimizations including direction optimizations, 
2) using the idea of clustering to compute the results from $O(w)$ sources simultaneously, and
3) using parallelism to improve performance. 
}

%% file: figs_algs/tab_graph_BFS.tex
\begin{table*}[!t]
  \centering
  \footnotesize
  \input{figs_algs/table1_2.tex}
  \caption{\small\textbf{
    Tested graphs and microbenchmarks on different BFS algorithms from a cluster of vertices with size 64. 
  }
  The numbers followed by `$\times$' are speedups, higher is better. Others are running time, lower is better. The columns ``AIY'', ``Ligra'' in related work and ``Final'' show the speedup over the ``Seq-BFS''. ``AIY'' is referred to sequential \ccbfs from~\cite{akiba2013fast}, ``Ligra'' is referred to parallel single BFS \cite{shun2013ligra}, and ``Final'' is referred to our parallel \ccbfs.  The ``self-speedup'' is the speedup running the algorithm in parallel over running it in sequential. 
  \hide{
  The `seq BFS' is conducting a regular sequential BFS algorithm. `+EdgeMap' is the speedup of performing Ligra BFSs on one core over the `seq BFS'. `+Cluster' is the speedup of conducting \ccbfs from one cluster on one core over the previous algorithm. `+parallel' is the speedup of running \ccbfs on 96 cores over the previous baseline, which is also the self-speedup of \ccbfs. The second last column shows the overall speedup of \ccbfs in parallel over the `seq BFS'. The last column is the running time of \ccbfs in parallel. The last line, `Mean' shows the geometric mean across different graphs. 
  } 
  \label{table:microbenchmark}
  }
\end{table*}

% \begin{tabular}{>{\bf}l|ccl|c|ccc|cc}
%   \toprule
%           & \multicolumn{3}{c|}{Graph Information} & seq BFS & \multicolumn{3}{c|}{Improvement, relative to previous} &\multicolumn{2}{c}{Final Result: CCBFS} \\
%   Dataset &                 $n$ &     $m$ &                                                          Notes & Time(s) &       +EdgeMap &     +Cluster &    +parallel &  Improvement       & Time(s) \\

%% file: figs_algs/table1_2.tex
\begin{tabular}{l|ccl|c|cc|cc|cc}
\toprule
        & \multicolumn{3}{c|}{Graph Information} & Seq-BFS & \multicolumn{2}{c|}{Related Work} & \multicolumn{2}{c|}{Parallel C-BFS} & \multicolumn{2}{c}{Self-Speedup} \\
Dataset &               $n$ &   $m$ &                                                       Notes & Time(s) &        AIY &        Ligra &          Final & Time(s) &        C-BFS &        Ligra \\
\midrule
     EP &             75.9K &  811K &                         Epinions1~\cite{backstrom2006group} &    0.18 & 20.6$\times$ & 4.02$\times$ &    102$\times$ &   0.002 & 4.39$\times$ & 9.44$\times$ \\
   SLDT &             77.4K &  938K &                       Slashdot~\cite{leskovec2009community} &    0.21 & 18.8$\times$ & 3.91$\times$ &   94.1$\times$ &   0.002 & 3.47$\times$ & 9.55$\times$ \\
   DBLP &              317K & 2.10M &                                DBLP~\cite{yang2015defining} &    0.77 & 20.3$\times$ & 6.22$\times$ &    183$\times$ &   0.004 & 10.2$\times$ & 17.1$\times$ \\
     YT &             1.13M & 5.98M &                         com-youtube~\cite{yang2015defining} &    3.30 & 22.9$\times$ & 17.8$\times$ &    445$\times$ &   0.007 & 20.6$\times$ & 31.6$\times$ \\
     SK &             1.69M & 22.2M &                                  skitter~\cite{nr, skitter} &    6.56 & 21.0$\times$ & 30.4$\times$ &    496$\times$ &   0.013 & 26.7$\times$ & 33.1$\times$ \\
   IN04 &             1.38M & 27.6M &                   in\_2004~\cite{BoVWFI,Boldi-2011-layered} &    4.08 & 20.9$\times$ & 4.00$\times$ &    171$\times$ &   0.024 & 10.1$\times$ & 17.8$\times$ \\
     LJ &             4.85M & 85.7M &                  soc-LiveJournal1~\cite{backstrom2006group} &    39.8 & 23.7$\times$ & 61.8$\times$ &   1017$\times$ &   0.039 & 47.7$\times$ & 53.9$\times$ \\
     HW &             1.07M &  112M &                            hollywood\_2009~\cite{nr,BoVWFI} &    18.7 & 20.9$\times$ & 89.7$\times$ &    928$\times$ &   0.020 & 32.4$\times$ & 48.7$\times$ \\
   FBUU &             58.8M &  184M &    socfb-uci-uni~\cite{nr,traud2012social,red2011comparing} &     268 & 32.0$\times$ & 49.6$\times$ &    973$\times$ &   0.276 & 54.4$\times$ & 52.8$\times$ \\
   FBKN &             59.2M &  185M &     socfb-konect~\cite{nr,traud2012social,red2011comparing} &     176 & 27.9$\times$ & 38.8$\times$ &    712$\times$ &   0.247 & 53.1$\times$ & 51.8$\times$ \\
     OK &             3.07M &  234M &                           com-orkut~\cite{yang2015defining} &    61.6 & 19.8$\times$ &  102$\times$ &   1119$\times$ &   0.055 & 49.0$\times$ & 65.4$\times$ \\
   INDO &             7.41M &  301M & indochina~\cite{nr,Boldi-2011-layered,boldi2004-ubicrawler} &    38.8 & 21.9$\times$ & 12.4$\times$ &    452$\times$ &   0.086 & 25.7$\times$ & 35.9$\times$ \\
     EU &             11.3M &  521M &          eu-2015-host~\cite{BoVWFI,Boldi-2011-layered,BMSB} &     119 & 23.9$\times$ & 26.6$\times$ &    821$\times$ &   0.145 & 18.5$\times$ & 41.3$\times$ \\
     UK &             18.5M &  523M &                    uk-2002~\cite{BoVWFI,Boldi-2011-layered} &    91.8 & 22.7$\times$ & 30.7$\times$ &    687$\times$ &   0.134 & 42.1$\times$ & 46.7$\times$ \\
     AR &             22.7M & 1.11B &                     arabic~\cite{BoVWFI,Boldi-2011-layered} &     147 & 22.5$\times$ & 10.7$\times$ &    461$\times$ &   0.319 & 18.0$\times$ & 33.8$\times$ \\
     TW &             41.7M & 2.41B &                             Twitter~\cite{kwak2010twitter}  &     861 & 20.6$\times$ &  157$\times$ &    856$\times$ &   1.006 & 56.3$\times$ & 60.2$\times$ \\
     FT &             65.6M & 3.61B &                          Friendster~\cite{yang2015defining} &    2084 & 20.4$\times$ &  187$\times$ &    813$\times$ &   2.563 & 59.4$\times$ & 64.6$\times$ \\
     SD &             89.2M & 3.88B &                                     sd\_arc~\cite{webgraph} &    1898 & 25.0$\times$ & 80.3$\times$ &    945$\times$ &   2.008 & 55.7$\times$ & 62.5$\times$ \\
     \midrule
GeoMean &                   &       &                                                             &    32.0 & 22.4$\times$ & 27.0$\times$ &    500$\times$ &   0.064 & 24.8$\times$ & 35.4$\times$ \\
\bottomrule
\end{tabular}

%% file: figs_algs/fig_bfs_compare.tex
\begin{figure*}[htbp]
  \centering
  \includegraphics[width=2\columnwidth]{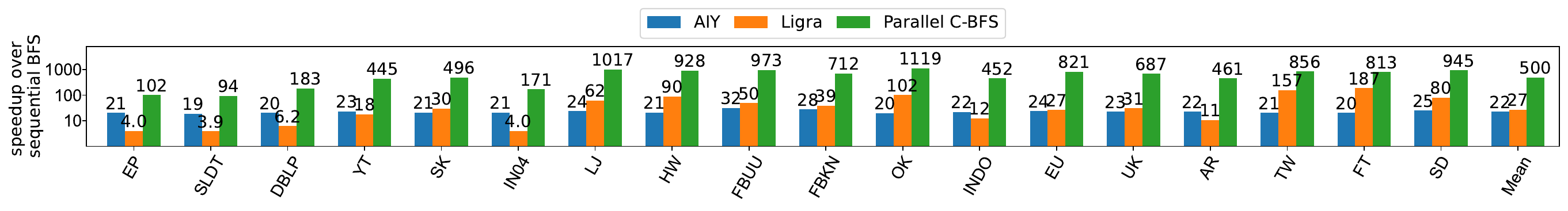}
  \caption{\small\textbf{Speedup of parallel Ligra BFSs and parallel \ccbfs over the standard sequential BFS on cluster with size 64.} $y$-axis is the speedup over sequential regular BFS in log-scale, higher is better. Each group of bars represents a graph, except the last group, which represents the average across all graphs. The numbers on the bar are the speedup of parallel algorithms over the standard sequential algorithm.  
  \label{fig:par_compare}%\vspace{-0.5em}
  }
\end{figure*} 

%% file: figs_algs/fig_BFS_scale.tex
\begin{figure}[!t]
    \centering
    \includegraphics[width=.8\columnwidth]{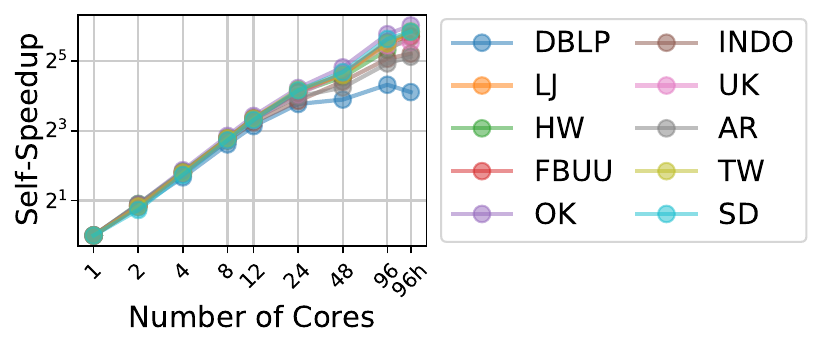}
    \caption{\small \textbf{The scalability curve on different number of processors for \ccbfs.}
    The y-axis is the self speedup.
    The \ccbfs running on one core is always 1. 
    The x-axis is the number of cores.
    96h represents 96 cores with hyperthreads. 
    }
    \label{fig:bfs_scale}
\end{figure} 

%% file: figs_algs/fig_BFS_d.tex
\begin{figure}[!t]
    \centering
    \includegraphics[width=.8\columnwidth]{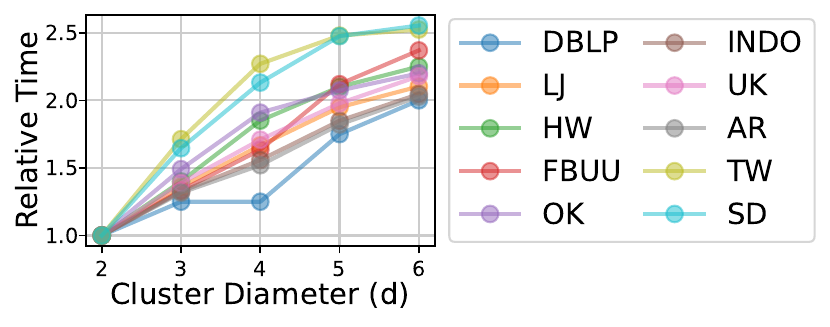}
    \caption{\small \textbf{The running time of \ccbfs on various cluster diameter $d$.}
    The $y$-axis shows the relative running time over $d=2$. The $x$-axis shows the cluster diameter $d$. 
    }
    \label{fig:ccbfs_d}
\end{figure}

%% file: exp_approx.tex
\subsection{Approximate Landmark Labeling}\label{sec:exp_appx}

\begin{table}[!t]
  \centering
  \footnotesize
  \input{figs_algs/table2.tex}
  \caption{
    \small\textbf{The index construction time, (1$+\epsilon$) distortion, and query time for ADO
    based on landmark labeling.} The ``Plain'' is the plain LL algorithm that each landmark is a single vertex. The ``$w=64$'' and ``$w=8$'' \ccbfs-based LL that landmarks are in clusters with size $w$.
    The memory budget is 1024 bytes per vertex. For both index time and $\epsilon$, lower is better. 
    \label{table:ApproxLL}
  }
\end{table}

% \multicolumn{4}{c|}{GeoMean Query Time for 1e5 Queries(ms)} & 1.8 & 1.8 & 2.3\\ % Your custom row

\input{figs_algs/fig_approx.tex}

%\subsection{Landmark Labeling}
We now show how \ccbfs{} can significantly improve the landmark
labeling (LL) approach with respect to both running time and accuracy.
%and obtain a parallel ADO. 
In general, more landmarks will lead to better accuracy for the distance queries, as the landmarks are more likely
to be on or close to the shortest path between the queried vertices. 
Hence, by using the clusters as landmarks, we can drastically increase the number of landmarks by a factor of $w$ with $O(d)$ overhead in time and space. 
For simplicity, we start by considering clusters with diameter $d=2$,
and later discuss clusters with $d>2$.
Following the optimization mentioned in \cref{sec:approxLL}, we use each cluster as a vertex and its $w-1$ neighbors,
and only store a distance and $d$ \bitsubset{s} in the \compactdis{}. 

We study the effectiveness of our approach by comparing our \ccbfs-based LL with a standard solution where each landmark is one vertex.
We limit the total memory usage for both algorithms to with a
parameter of $t$ bytes per vertex for each graph, and construct an LL-based index
within this memory budget. 
For \ccbfs{}, memory usage per cluster includes the distance (1 byte) and
$d$ \bitsubset{s} ($w/8$ bytes each) per vertex, while the memory usage for a regular LL is one byte (the distance) per vertex. 
For example, \ccbfs with $w=64$ and $d=2$ needs $17$ bytes to store a cluster distance vector, and \ccbfs with $w=8$ and $d=2$ only needs $3$ bytes. 
Thus, the memory usage depends on the word size and number of clusters we choose. 
For fair comparisons, we fix the memory usage per vertex in the index for different baselines.  
% Note that using a larger $w$ usually allows for faster preprocessing time, but also requires larger space. \letong{not correct? with what condition} 
With the same memory budget, a larger $w$ results in fewer (independent) clusters, typically allowing faster preprocessing.
As discussed in \cref{sec:approxLL},
the landmarks in the same cluster are highly correlated to each other, and 
may not bring the same benefit as independent ones. 
Thus, a larger $w$ may also result in less accuracy. 

In this experiment,  we tested on different $w$ from $\{8,16,32,64\}$ and $d=2$.
In \cref{table:ApproxLL}, we show the two extremes of using $w=64$ and
$w=8$. \ifconference{Due to page limits, we put the full information of all tested $w$ in our full version paper.}\iffullversion{The full information of all tested $w$ can be found in \cref{sec:full_approx_info}.}

% In this experiments, we choose the two extremes of using $w=64$ and $w=8$, respectively. 
% Therefore, with the same memory budget, using $w=64$ means to select $8\times$ fewer clusters than $w=8$. \letong{not correct, is about 5.7 times fewer}

%Although they roughly result in the same number of landmarks, as discussed in \cref{sec:approxLL},
% Although CC64 still results in more landmarks, it selects fewer (independent) clusters than CC8. 

%increasing the number of clusters will improve accuracy. 
% As we will show in the experimental results, using $w=64$ can be faster in preprocessing, but less accurate than $w=8$, giving the same memory budget. 

We show the full result of using memory limit as $t=1024$ bytes per vertex in \cref{table:ApproxLL}, where ``plain'' refers to using simple parallel BFS, 
``$w=64$'' is to use \ccbfs{} with $w=64$ (17 bytes per cluster), 
and ``$w=8$'' is \ccbfs{} with $w=8$ (3 bytes per cluster).
With the memory limit of 1024 bytes per vertex, 
% we can choose 1024 landmarks for regular LL, 341 clusters (2728 landmarks) for CC8, or 60 clusters (3840 landmarks) for CC64. \letong{may not be simply 341*64, because some clusters may not contain 64 vertices}
we can choose 1024 landmarks for regular LL, 341 clusters for CC8, or 60 clusters for CC64. 
For regular LL, each landmarks are chosen based on prioritizing the high-degree vertices. 
We compare the index time and the \emph{error} $\epsilon$ shown in percentage. 
%which is defined as $\epsilon=query(u,v)-\delta(u,v)$ for a query pair $(u,v)$.
%It also means that the ADO has $(1+\epsilon)$ distortion. 
It means that the ADO has $(1+\epsilon)$ distortion (defined in \cref{sec:approxLL}). 
We take the average of 100,000 pairs to estimate the error on all the graphs, except for the largest two graphs FT and SD. We only compute the distance of 10,000 pairs on these two graphs, since generating the ground truth is expensive on large graphs.

On all 18 tested graphs, \ccbfs{} with $w=64$ always gives a lower running time.
When $w=8$, the running time is 2.4--5.5$\times$ higher than $w=64$, but still mostly faster 
than the plain version. 
%The few exceptions are all on large graphs (TW, FT and SD).
%\yihan{any explanations? I'm not sure why that's the case, but all cases that CC8 is slow are on large graphs.}
With the same memory budget, $w=64$ roughly processes 5.68$\times$ fewer clusters than $w=8$.
Similarly, comparing $w=64$ and the plain version,
the number of (clustered) BFSs performed by $w=64$ is $17\times$ fewer
than the number of (single) BFSs by regular LL.
The running time can be up to $53\times$ faster, but on average, it is around 10$\times$---each \ccbfs{} is still more expensive than a single BFS,
but the numbers indicate that the overhead is small. 
%CC64 maintains two more \bitsubset{s},

Regarding distortions, $w=8$ generally gives better accuracy than $w=64$, 
but it can also be worse in several instances.
That is because $w=64$ selects more landmarks than $w=8$ (3840 vs.\ 2728) 
but fewer independent clusters (60 vs.\ 341).
Therefore, the loss of using fewer clusters may or may not be compensated by
more (correlated) landmarks.
Empirically, the results still suggest that $w=8$ gives overall
better accuracy, as it is more accurate on 11 out of 18 tested graphs. 
Both $w=8$ and $w=64$ are more accurate than the plain version on at least 16 out of 18 graphs.
This indicates that the increased number of landmarks, although less independent, still positively affects
the accuracy on most of the 18 scale-free networks we tested here. 

It is worth noting that the query times for all graphs are similar, 
and we show the average (geometric mean) query time at the last line in \cref{table:ApproxLL}. 
This is because the query time is completely determined by the number of landmarks,
instead of the graph size. %\letong{I choose to geomean, explain more?}
%Later in this section, we show how our optimization based on bidirectional search
%improves query accuracy with light overhead in time. 

The experimental results suggest that, if the primary objective is to reduce preprocessing time,
using \ccbfs{} with $w=64$ will always give a more efficient version than the plain version. 
The precision is also superior in most cases. 
When achieving better precision is the top priority, \ccbfs{} with a smaller $w$
can significantly reduce the distortion.
When using $w=8$, the distortion is almost always better than both $w=64$ and the plain version,
%and the preprocessing time is also better than regular LL in most cases. 
and the running time also improves over the plain version on most of the graphs. 

We also show the trade-off between preprocessing time and distortion in \cref{fig:approx} on three representative graphs,
particularly including those that $w=8$ and $64$ may perform poorly on. 
We include results using $w=16$ and $w=32$, with different memory budgets ($t$ bytes per vertex, shown in the $x$-axis).
In general, it is still true that large $w$ results in faster preprocessing but larger distortion.
The red line shows the baseline of regular LL. 
For both preprocessing time and distortion, lower is better.
%As the two extremes, $w=8$ or 64, may perform worse than the baseline in distortion.
Overall, using $w=16$ or $w=32$ provides a compromise for both time and distortion, 
and can be a more stable choice across all graphs. 

For both preprocessing time and distortion, using \ccbfs{} provides a significant improvement
on almost all graphs.  
This verifies the effectiveness of our \ccbfs{} to achieve a practical distance oracle.

\myparagraph{\ccbfs with Diameter $d>2$.}
Unlike the (sequential) AIY algorithm that can only apply to $d=2$ case (a vertex and its neighbors),
our \ccbfs is general and works for any given $d$.
Hence, it is interesting to understand how the performance and quality are affected by varying $d$.
Note that larger $d$ gives us flexibility in selecting the sources.
It can generally reduce the correlations between the sources, and even sometimes enable a larger $k$ if no vertices in the graph have large degrees (although this is rare in the scale-free networks tested in this paper).
Therefore, we tested the time and distortion of LL on clusters with larger $d$ and present the results in \cref{sec:exp_d}. 
The takeaway is that, although larger $d$ may provide higher source quality, the space and time consumption are also linearly proportional to $d$.
Hence, given the same memory limit, choosing clusters with $d>2$ does not help with decreasing distortion (other than the graph OK, where $d=3$ improves the distortion over $d=2$ by $3\%$). 
However, generally, the faster running time is due to fewer clusters that can be selected for the same memory budget.

%% file: figs_algs/table2.tex
\begin{tabular}{l|c@{  }@{  }c@{  }@{  }c|r@{  }@{  }r@{  }@{  }r}
\toprule
     & \multicolumn{3}{c|}{Index Time (s)} & \multicolumn{3}{c}{$\epsilon$ (\%)} \\
Data &         Plain & $w=64$ & $w=8$ &          Plain & $w=64$ & $w=8$ \\
\midrule
  EP &          1.26 &   0.02 &  0.08 &            0.4 &    0.1 &   0.1 \\
SLDT &          1.15 &   0.02 &  0.07 &            0.7 &    0.1 &   0.1 \\
DBLP &          3.57 &   0.08 &  0.25 &            2.5 &    2.2 &   1.0 \\
  YT &          9.22 &   0.23 &  0.59 &            0.3 &    0.3 &   0.1 \\
  SK &          13.4 &   0.55 &  1.77 &            1.4 &    0.7 &   0.4 \\
IN04 &          20.0 &   0.96 &  3.88 &            2.1 &    1.9 &   0.9 \\
  LJ &          36.2 &   1.72 &  5.63 &            5.0 &    4.3 &   3.5 \\
  HW &          12.4 &   0.93 &  4.10 &           10.6 &    5.6 &   7.1 \\
FBUU &           138 &   11.3 &  27.0 &            6.2 &   11.9 &   6.9 \\
FBKN &           127 &   10.5 &  24.9 &            6.2 &   11.9 &   6.9 \\
  OK &          26.3 &   2.87 &  10.1 &            8.7 &    7.7 &   7.3 \\
INDO &          83.2 &   5.44 &  29.8 &            3.1 &    1.5 &   1.3 \\
  EU &          87.3 &   7.01 &  34.9 &            2.6 &    1.3 &   1.7 \\
  UK &          80.4 &   8.28 &  38.8 &            3.9 &    4.9 &   3.1 \\
  AR &           148 &   17.6 &  86.8 &            2.6 &    4.0 &   2.2 \\
  TW &           112 &   31.0 &  99.3 &            1.5 &    1.4 &   1.1 \\
  FT &           251 &   61.1 &   193 &           16.8 &   12.4 &  12.8 \\
  SD &           318 &   75.6 &   255 &            0.6 &    0.3 &   0.3 \\
\bottomrule
\end{tabular}

%% file: figs_algs/fig_approx.tex
\begin{figure*}
    \centering
    \vspace{-1em}
    \includegraphics[width=1.9\columnwidth]{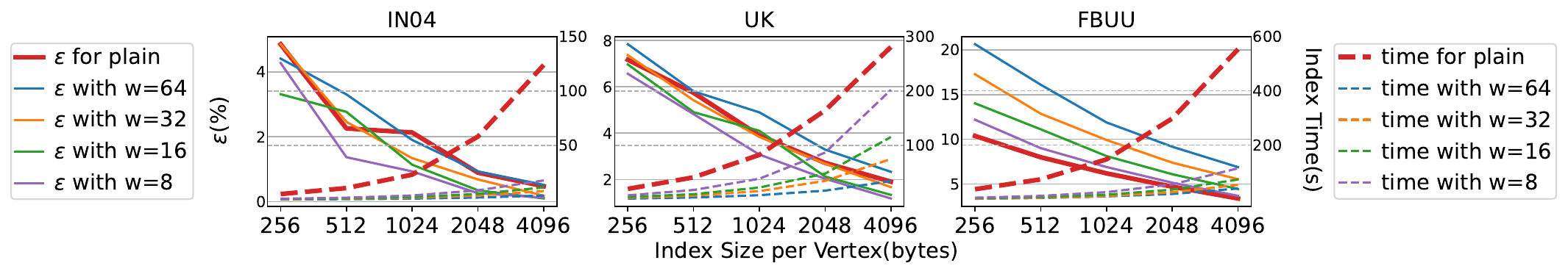}    \caption{\small \textbf{Tradeoffs between index size and distortion/construction time}
    The $x$-axis is the memory limits per vertex in bytes, and is in log-scale. The $y$-axis on the left shows the ($1+\epsilon$) distortion.  The $y$-axis on the right shows the preprocessing time. For both preprocessing time and
    distortion, lower is better. For the algorithms compared here, `plain' is the regular LL, others are the \ccbfs-based LL that choose clusters with size $w$ as landmarks.
    \label{fig:approx} %\vspace{-.75em}
    }
\end{figure*}

%% file: exp_CCBFS_d.tex
\subsection{Further Study on the \texorpdfstring{$d > 2$}{d > 2} Case in Approximate Distance Oracle} \label{sec:exp_d}

Here, we provide additional information for the $d>2$ case for the landmark labeling, and continue the discussion from \cref{sec:exp_appx}.
Recall $d$ is the diameter of clusters.
We first study its influence in \ccbfs. Then, we study the influence of $d$ in the application Landmark Labeling introduced in in \cref{sec:approxLL}. 

% \subsection{The Performance of \ccbfs on Different $d$}
\myparagraph{The Performance of \ccbfs on Different $d$.}
According to Thm.~\ref{theorem:batchBFScost}, the work of \ccbfs is proportional to $d$. 
In order to know how to choose proper $d$ in different applications, we need to first know how different $d$ affect the running time of \ccbfs in practice.
We tested the running time of \ccbfs with clusters of different diameter $d$.
The results are shown in \cref{table:BFS_d}. 
The running time increases as $d$ grows larger, as shown in \cref{table:BFS_d}, and  space usage ($O(d)$ space to store the distances of a cluster to a vertex) also increases. 
When applying \ccbfs to other applications, we need to weigh the benefits brought by more general clusters and the overhead on time and space costs. 
\cref{table:BFS_d} provides a reference for overhead on running time.

\begin{table}[!t]
  \centering
  \footnotesize
  \input{figs_algs/table_BFS_d.tex}
  \caption{
    \small\textbf{The parallel \ccbfs time (seconds) for one cluster with size 64 on different cluster diameter $d$.} 
    \label{table:BFS_d}
  }
\end{table}

% \subsection{The Performance of Different $d$ on Landmark Labeling}
\myparagraph{The Performance of Different $d$ on Landmark Labeling.}
Our \ccbfs is the first implementation of \ccbfs that supports general clusters. It gives us a chance to study the performance and quality of LL for clusters with larger $d$. Clusters with larger $d$ have fewer correlations for vertices in the same cluster, but the computational cost for \ccbfs also increases. We are interested in whether it is worth using clusters with larger $d$ in the LL application.

We show the full result of using memory limit as $t=1024$ bytes per vertex and word size $w=64$ in \cref{table:ApproxLL_d}, where  ``plain'' refers to using simple parallel BFS (1 bytes per landmark), ``$d=2$'' is to use \ccbfs with $d=2$ (17 bytes per cluster), ``$d=3$'' is \ccbfs with $d=3$ (25 bytes per cluster)  and  ``$d=4$'' is \ccbfs with $d=4$ (33 bytes per cluster). With memory limits 1024 bytes per vertex, we can choose 1024 landmarks for regular LL, 60 clusters for $d=2$, 40 clusters for $d=3$ , and 31 clusters for $d=4$. Landmarks are chosen based on prioritizing the high-degree vertices, similar to the setting previously.

The construction time is reversely proportional to $d$, since larger $d$ leads to fewer sources and clusters.
However, for distortion, other than the graph OK, larger $d$ always leads to lower accuracy.
For the graph OK, the improvement is only about $3\%$.
In conclusion, since for distance oracles, generally the space usage and accuracy are the most crucial---given the same memory budget, choosing clusters from diameter $d>2$ does not help with accuracy on most graphs. 
However, since our \ccbfs algorithm is highly parallel and efficient, it provides opportunities for researchers in the future to study other applications on whether more general sources with $d>2$ can be more effective than star-size ones with $d=2$.

\begin{table}[!t]
  \centering
  \footnotesize
  \input{figs_algs/table_ApproxLL_d.tex}
  \caption{
    \small\textbf{The Approximate Landmark Labeling Time and Distortion for cluster with size 64 and memory limits 1024 bytes per vertex on different cluster diameter $d$.} 
    \label{table:ApproxLL_d}
  }
\end{table}

%% file: figs_algs/table_BFS_d.tex
\begin{tabular}{l|rrrrrr}
\toprule
   Data & $d=2$ & $d=3$ & $d=4$ & $d=5$ & $d=6$ \\
\midrule
     EP & 0.002 & 0.002 & 0.003 & 0.003 & 0.003 \\
   SLDT & 0.002 & 0.002 & 0.002 & 0.002 & 0.003 \\
   DBLP & 0.004 & 0.005 & 0.005 & 0.007 & 0.008 \\
     YT & 0.007 & 0.011 & 0.012 & 0.014 & 0.015 \\
     SK & 0.013 & 0.017 & 0.020 & 0.023 & 0.026 \\
   IN04 & 0.024 & 0.027 & 0.029 & 0.032 & 0.033 \\
     LJ & 0.039 & 0.053 & 0.065 & 0.076 & 0.082 \\
     HW & 0.020 & 0.028 & 0.037 & 0.042 & 0.045 \\
   FBUU & 0.276 & 0.368 & 0.451 & 0.585 & 0.654 \\
   FBKN & 0.247 & 0.318 & 0.388 & 0.501 & 0.564 \\
     OK & 0.055 & 0.082 & 0.105 & 0.114 & 0.121 \\
   INDO & 0.086 & 0.114 & 0.134 & 0.159 & 0.176 \\
     EU & 0.145 & 0.195 & 0.239 & 0.284 & 0.290 \\
     UK & 0.134 & 0.186 & 0.229 & 0.265 & 0.292 \\
     AR & 0.319 & 0.420 & 0.485 & 0.581 & 0.647 \\
     TW & 1.006 & 1.723 & 2.284 & 2.494 & 2.539 \\
     FT & 2.563 & 4.412 & 5.880 & 6.594 & 6.815 \\
     SD & 2.008 & 3.305 & 4.280 & 4.970 & 5.128 \\
\midrule
GeoMean & 0.064 & 0.086 & 0.104 & 0.120 & 0.131 \\
\bottomrule
\end{tabular}

%% file: figs_algs/table_ApproxLL_d.tex
\begin{tabular}{l|r@{ }r@{ }r@{ }r|r@{ }r@{ }r@{ }r}
\toprule
     & \multicolumn{4}{c|}{Construction Time (s)} & \multicolumn{4}{c}{Distortion $\epsilon$ (\%)} \\
Data &         Plain & $d=2$ & $d=3$ & $d=4$ &          Plain & $d=2$ & $d=3$ & $d=4$ \\
\midrule
  EP &          1.26 &  0.03 &  0.02 &  0.02 &            0.4 &   0.1 &   0.2 &   0.2 \\
SLDT &          1.15 &  0.02 &  0.02 &  0.02 &            0.7 &   0.1 &   0.4 &   0.6 \\
DBLP &          3.57 &  0.09 &  0.07 &  0.06 &            2.5 &   2.2 &   3.4 &   4.0 \\
  YT &          9.22 &  0.23 &  0.19 &  0.17 &            0.3 &   0.3 &   0.6 &   0.8 \\
  SK &          13.4 &  0.58 &  0.41 &  0.37 &            1.4 &   0.7 &   1.5 &   2.1 \\
IN04 &          20.0 &  1.06 &  0.70 &  0.58 &            2.1 &   1.9 &   2.4 &   2.6 \\
  LJ &          36.2 &  1.79 &  1.41 &  1.26 &            5.0 &   4.3 &   5.5 &   6.0 \\
  HW &          12.4 &  0.99 &  0.75 &  0.63 &           10.6 &   5.6 &   6.5 &   7.0 \\
FBUU &           138 &  11.7 &  9.31 &  8.26 &            6.2 &  11.9 &  13.4 &  15.9 \\
FBKN &           127 &  10.8 &  8.82 &  7.52 &            6.2 &  11.9 &  13.5 &  16.0 \\
  OK &          26.3 &  3.05 &  2.65 &  2.32 &            8.7 &   7.7 &   7.5 &   7.6 \\
INDO &          83.2 &  6.09 &  3.92 &  3.40 &            3.1 &   1.5 &   2.2 &   2.6 \\
  EU &          87.3 &  8.00 &  6.33 &  5.55 &            2.6 &   1.3 &   1.9 &   2.2 \\
  UK &          80.4 &  9.06 &  6.63 &  5.90 &            3.9 &   4.9 &   5.3 &   5.6 \\
  AR &           148 &  19.4 &  13.4 &  11.7 &            2.6 &   4.0 &   6.8 &   7.1 \\
  TW &           112 &  31.0 &  27.8 &  23.9 &            1.5 &   1.4 &   1.5 &   1.8 \\
  FT &           251 &  61.1 &  52.0 &  45.8 &           16.8 &  12.4 &  13.7 &  14.6 \\
  SD &           318 &  75.6 &  61.2 &  53.1 &            0.6 &   0.3 &   0.8 &   0.9 \\
\bottomrule
\end{tabular}

%% file: related.tex
\section{Related Work}\label{sec:related}

This paper mainly focuses on  scale-free (small-diameter) networks, and we refer the audience to an excellent survey~\cite{BastDGMPSWW16} of algorithms for large-diameter graphs (e.g., road networks).
In this paper, we discuss the approaches based on landmark labeling, and here, we review other approaches for distance queries.
We first review the approximate solutions.
The concept of approximate distance oracles (ADOs) was proposed by \citet{thorup2005approximate}, and has later been theoretically studied in dozens of papers.
Practically, papers~\cite{Potamias09,gubichev2010fast,das2010sketch} discuss how to select the best ``landmarks'' for these type of sketch-based solutions, and showed that degrees or betweenesses are reasonable metrics.
Other solutions include embedding-based solutions~\cite{zhao2010orion,rao1999small,shavitt2008hyperbolic,meng2015grecs} (embedding graph metric into other simpler ones such as Euclidean space), tree-based approaches~\cite{blelloch2017efficient,zeng2023litehst}, and some recent attempts using deep learning~\cite{rizi2018shortest}.
Most of these approaches are more complicated than the landmark-based labeling mentioned in this paper.

Regarding the exact distance queries, Pruned Landmark Labeling (PLL)~\cite{akiba2013fast} is among the latest solutions for scale-free networks.  
Li et al.~\cite{li2019scaling} showed another implementation, but they did not release their code, so we cannot compare their running time with ours.
Other techniques, such as contraction hierarchies (CH)~\cite{geisberger2012exact} and transit nodes routing~\cite{bast2007fast}, focus more on large-diameter graphs like road networks.

%% file: conclusion.tex
\section{Conclusion} 

In this paper, we present parallel implementations for \batchBFS{},
which runs BFS from a cluster of vertices with diameter $d$. 
Our algorithm is work-efficient in theory, and also leads to high parallelism on low-diameter graphs as we tested in the experiments.
We employ both bit-level and thread-level parallelism to optimize the performance. 
Both of them lead to significant speedup.
Especially, we observed that bit-level and 
thread-level parallelism work well in synergy. 
We also show that the combination of the techniques also
leads to performance improvement in two applications in distance oracles
in multiple measurements of preprocessing time and accuracy, 
and allows our implementation to scale to much larger graphs than a sequential algorithm. Besides, our \ccbfs is the first implementation that supports general clusters with diameter $d$ instead of  star-shaped clusters, which give us a chance to study the benifits and overhead of choosing clusters with larger $d$ in \ccbfs and in its applications. 

\hide{ 
\yihan{The paper is 1+ pages over for now. A simple way to remove some content is to remove everything related to bidirectional search since it doesn't seem to 
affect the main story much. WDYT?
} 
}

%% file: acknowledge.tex
\section*{Acknowledgments}
This work is supported by NSF grants CCF-1919223, CCF-2103483, CCF-2119352, IIS-2227669, TI-2346223, and NSF CAREER Awards CCF-2238358 and CCF-2339310, the UCR Regents Faculty Development Award, and the Google Research Scholar Program.

%% file: appendix.tex
\section{Parallel BFS}
\input{parallel_BFS.tex}

% \input{exp_CCBFS_appendix.tex}

\section{Applying Bi-Directional Search in the Queries} \label{sec:biBFS}
In the ADO based on landmark labeling, 
the distortion between two faraway vertices can be reasonably good,
but the returned distance of two nearby vertices may not be as accurate~\cite{akiba2012shortest}.
This is because when the shortest path between $u$ and $v$ is short,
it is less likely to pass a landmark, or to be close to a landmark. 
To further improve the accuracy, several techniques were proposed~\cite{gubichev2010fast,tretyakov2011fast,qiao2012approximate}. They typically store shortest-path trees rooted at the landmarks instead of just storing distances, and finding loops or shortcuts on the trees during the query to reduce the distortion. While they improve the accuracy, the query time becomes much slower.
%The overall distortion for two random vertices are good (see \cref{sec:exp_appx}). 
%However, the distortion for pairs with small distance is worse than that for pairs with large distance\cite{}, because the lengths of shortest paths for close pairs are small so that they are more unlikely to pass nearby landmarks. There is lots of related work on reducing the distortion for close pairs \cite{}. We use a simple way that works well in practice. 
% We found that, applying length-limited bi-directional local search from the two query vertices can , which is a trade-off for precision and query time. 
In our implementation, we add a fixed-size bi-directional search between queried vertices $u$ and $v$ before we run the actual query using landmarks.
In particular, we will search $\tau$ vertices from each side by a BFS order, and take an intersection to find the minimum distance among them. 
Conceptually, the vertices encountered in the search can be viewed as landmarks generated on the fly. 
The intuition is that, for two nearby vertices, a bidirectional search should take a short time but give an exact distance with no distortion. 
In our experiments, when choosing a proper search size $\tau$, this optimization greatly improved query quality with a small overhead in query time. 
%In addition, we can always skip the landmarks 
\hide{For the pairs that are close, we can direcly apply a bi-directional search from the two queried vertices. Note that, the two vertices are close to each other, such search will not be expensive.  Further more, we can skip the landmarks during the search, because if the shortest paths passing through any landmark, it must be correctly answered by the index. Besides, we choose landmarks prioritize by their degrees, when a reasonable number of landmarks have been chosen, we can expect the remaining vertices does not has large degrees. According to the above intuitions, we apply length limited bi-directional search from the two query vertices and skip landmarks during the process. We take the minimum distance obtained by the index and bi-directional search as the final answer of a query. This optimization provides a trade-off between query time and distortion, and we explored this trade-off in \cref{sec:exp_appx}. 
}

\input{figs_algs/fig_approx_query.tex} 

\myparagraph{Performance Study of Query Optimizations.} 
As we mentioned, using a bidirectional search may help improve the accuracy for close-by vertices. 
We test different local search sizes and show the improvement in distortion and their impact on query time in \cref{fig:approx_query}.
For almost all cases, using a local search size of 512 or more, we can see a clear improvement in distortion.
However, with the local search size reaches 2048, the query time increases dramatically.

\section{Approximate Landmark Labeling Full Information}\label{sec:full_approx_info}
In the paper, due to page limits, we only show the two extremes $w=64$ and $w=8$. 
We provide the entire experimental results for all the tested $w$, and their query time in \cref{table:ApproxLL_full}. 
\begin{table*}[htbp]
  \centering
  \small
  \input{figs_algs/table2_full.tex}
  \caption{
    \small\textbf{The index construction time, (1$+\epsilon$) distortion, and query time for ADO
    based on landmark labeling.} The ``Plain'' is the normal LL algorithm in which each landmark is a single vertex. Others are \ccbfs-based LL that landmarks are in clusters with size $w$.
    The memory budget is 1024 bytes per vertex. For both index time and $\epsilon$, lower is better. 
    \label{table:ApproxLL_full}
  }
\end{table*}

\section{Exact 2-Hop Distance Oracle}
\label{sec:2hop_distance}
\input{application_2hop.tex}
\input{exp_2hop.tex}

\clearpage 

%% file: parallel_BFS.tex
\label{sec:parallel_BFS}
We briefly review direction-optimizing parallel BFS since it is the state-of-the-art parallel BFS~\cite{gbbs2021}, and many of the ideas are also used in our parallel \batchBFS. 
%We start with BFS from a single source $s\in V$ (high-level idea in \cref{alg:simpleBFS}). 
We present the high-level idea of BFS from a single source $s\in V$ in \cref{alg:simpleBFS}. 
The algorithm maintains a \defn{frontier} of vertices to explore in each round, starting from the source, and the algorithm finishes in $D$ rounds. 
It also maintains the distance from the source to each vertex in array $\delta$, initialized as infinity except for the source. 
In round $i$, the algorithm processes the current frontier $\ff_i$,  adding their out neighbors to the next frontier $\ff_{i+1}$ if the
neighbor has not been visited.
If multiple vertices in $\ff_i$ attempt to add the same vertex to $\ff_{i+1}$ to the next frontier, a \textsc{compare\_and\_swap} is used to ensure that only one can succeed in updating the distance from $\infty$ to the current distance.  
% The process that maps a subset of vertices (current frontier) to another subset of vertices (next frontier) by applying certain function to the edges ($cond\_f$) from current frontier, is abstracted as \defn{\edgemap} by Ligra\cite{}. 
% We use the framework defined by Ligra, which is shown in \cref{alg:edgemap} \textsc{EdgeMap-Sparse}.

Following the existing graph processing library Ligra \cite{shun2013ligra,shun2015smaller}, we use the \edgemap{} framework with the direction-optimizing parallel BFS (see \cref{alg:simpleBFS,alg:edgemap}). 
It maps a subset of vertices (current frontier) to another subset of vertices (next frontier) by applying a given function to the out-edges from the current frontier. 
\edgemap requires two user-defined functions, \condf and \edgef. 
\condf{$(v)$} is a function to indicate whether the vertex $v$ needs further processing.  In BFS, it checks whether the vertex has
not been visited---i.e., whether its distance is still $\infty$ (\cref{line:bfs_cond} in \cref{alg:simpleBFS}). 
\edgef{$(u,v)$} is a function for edge $(u,v)$, which processes the edge, and returns a boolean value indicating whether $v$ should be added to the next frontier by $u$. 
In BFS, it sets the distance of $v$ to one plus the distance of $u$. 
When multiple vertices want to add $v$ at the same time, $u$ should add $v$ iff. it wins \textsc{compare\_and\_swap} (\cref{line:bfs_edge} in \cref{alg:simpleBFS}). 
% For example, to implement the BFS mentioned above using \edgemap framework, $cond\_f(v)$ and $edge\_f(u,v)$ are set as \cref{line:bfs_cond,line:bfs_edge} in \cref{alg:simpleBFS}. \edgemap is the building block for parallel BFS and many other vertex-based graph algorithms. Our \batchBFS also uses \edgemap. 

%\guy{The following paragraphs should be shortened.  Too much detail.}
The idea of \emph{direction optimization} means to implement 
\edgemap{} in two different ``modes'': \mapsparse (forward) and \mapdense (backward), as shown in \cref{alg:edgemap}. 
In the sparse (forward) mode, we start with the current frontier and consider its out-neighbors as mentioned above. 
However, if the number of out-neighbors of the frontier is sufficiently large (i.e., close to $n$), 
it can be more efficient to use the dense (backward) mode instead, where all unvisited vertices look at all their in-neighbors to see if there is any on the frontier. 
With a large frontier, the dense mode can avoid costly atomic operations and make better use of cache locality, giving better performance. 
\edgemap is a building block for parallel BFS and many other vertex-based graph algorithms. Our \batchBFS also uses \edgemap with the directional optimization.

%There are two typical ways to implement \edgemap{}, \mapsparse and \mapdense (see \cref{alg:edgemap}). 

\input{figs_algs/alg_simpleBFS}
\hide{
\mapsparse lets vertices in the frontier process their neighbors that satisfy the condition \condf. If \edgef{$(u,v)$} succeeds, $u$ will put its neighbor $v$ to the next frontier (as we mentioned previously). 
In \mapdense{}, each the vertex $v\in V$ satisfying the condition \condf will let each of its neighbor $u$ to visit itself by using \edgef{$(u,v)$}. 
All these neighbors $u$ will be processed sequentially. 
Once a neighbor $u$ processes \edgef successfully, $v$ stops processing the remaining neighbors and put itself to the next frontier.  
These two versions of \edgemap have their own advantageous scenarios.  
When the frontier size is small, \mapsparse is more efficient, and the work is proportional to the number of edges processed. 
When the frontier size is large, \mapdense is more efficient because it can avoid atomic operations when processes edges---
%In BFS, \mapsparse may have multiple vertices want to process the same neighbor and update its distance at the same time, so atomic operations are needed in $edge\_f$ to avoid data race. 
in \mapdense, one vertex looks up its neighbors and updates its own distance, so no atomic update is needed. 
%Because no other vertices will update its distance, no atomic operations are needed in $edge\_f$. Besides, \mapdense can exit early that avoids visiting all the edges. Because our algorithm can not exit early, so we do not explain its details here. 
% \edgemap has an optimization when the frontier size is comparable to $n$, which is called \textsc{EdgeMap-Dense}. Instead of letting vertices in the frontier forward propagate information to their neighbors, \textsc{EdgeMap-Dense} lets all the vertices in $V$ backward pull information from their neighbors in the frontier (see \textsc{EdgeMap-Dense} in \cref{alg:edgemap}). 
% There are two benifits of using \textsc{EdgeMap-Dense}, one is that it may skip visiting edges: once it finds a neighbor in the frontier, it stops visiting remaining neighbors (see \cref{line:exit_early}), which makes parallel BFS can be even faster than sequential BFS algorithm on one core \cite{}. The other benifit is that it avoids atomic operation when adding vertices to the next frontier. 
% \textsc{EdgeMap-Sparse} needs atomic operations to guanrantee one vertex is only added once to the next frontier.
% \textsc{EdgeMap-Dense} let each vertex decide whether to put itself to the next frontier, which already guanrantees that no duplicates in the next frontier, so there is no need for atomic operation. 
% To distinguish the two kinds of \textsc{EdgeMap}, Ligra call the normal version as \textsc{EdgeMap-Sparse}, and the optimization version as \textsc{EdgeMap-Dense}.
\edgemap will automatically decide which version to use according to heuristics related to the frontier size (see \cref{alg:edgemap}). 
%Note that not all the algorithms are able to use \mapdense. 
%Our \batchBFS can switch between \mapsparse and \mapdense to take advantages of both of them.
Our \batchBFS also uses both \mapsparse and \mapdense to optimize the performance.
% Although we can not exit early in the dense version for correctness, we can still take advantage of avoiding atomic writes. 
}

\input{figs_algs/alg_edgemap.tex}

% \subsection{Problem Definitions}

%% file: figs_algs/alg_simpleBFS.tex
\begin{algorithm}[t]
  \small
 \caption{Framework of Parallel BFS\label{alg:simpleBFS}}
 \KwIn{A graph $G=(V,E)$ and a source $s\in V$}
 \SetKwFor{parForEach}{parallel\_for\_each}{do}{endfor}
 \SetKwInOut{Maintains}{Maintains}
 \DontPrintSemicolon
 $\delta \gets \{\infty, ... ,\infty\}$\\
 $\delta[s]\gets 0$\\
 $\ff_0\gets\{s\}$\\
 $i\gets 0$\\
 $\condf{(v)} \gets$ \Return $\delta[v] =
 \infty$ \label{line:bfs_cond} \\%\guy{I changed from not equal to equal.}.
 $\edgef{(u,v)} \gets$ \Return \textsc{compare\_and\_swap}$(\&\delta[v], \infty, \delta[u]+1)$ \label{line:bfs_edge}\\
 \While {$\ff_{i}\ne \emptyset$}{
  % \parForEach { $\{(u,v)| u\in \ff_i, v\in N(u)\}$ \label{line:EdgeMap_start}}{
  %   \If{$\textsc{compare\_and\_swap}(\&\delta[v], \infty, \delta[u]+1)$}{
  %     add $v$ to $\ff_{i+1}$ \label{line:EdgeMap_end}
  %   }
  % }
  $\ff_{i+1} \gets \textsc{EdgeMap}(\ff_i, \condf, \edgef)$ \label{line:EdgeMap}\\
 $i\gets i+1$
 }
 \Return $\delta$
 \end{algorithm}

%% file: figs_algs/alg_edgemap.tex
\begin{algorithm}[t]
  \footnotesize
 \caption{Framework of \textsc{EdgeMap}\label{alg:edgemap}}
 \KwIn{A subset of vertices $\ff_{in}$, a condition function for vertex \condf, and a mapping function for edge \edgef}
 \KwOut{A subset of vertices $\ff_{out}$}
 \DontPrintSemicolon
%  \myfunc{\upshape\textsc{EdgeMap}$(F_i, f, cond)$}{
%     $l$ is the frontier size.\\
%     \If{$F_i$ is sparse}{
%         $d$ is the degree sum of the frontier\\
%         \lIf {$(l+d)> (m/10)$}{\Return{$\textsc{EdgeMap-Dense}(F_i)$}}
%         \lElse{\Return {$\textsc{EdgeMap-Sparse}(F_i)$}}
%     }\Else{
%     \lIf{$l>n/10$} {\Return {$\textsc{EdgeMap-Dense}(F_i)$}}
%     \lElse{\Return $\textsc{EdgeMap-Sparse}(F_i)$}
%     }
% }
\myfunc{\upshape\textsc{EdgeMap}$(\ff_{in}, \condf, \edgef)$}{
  \If{$N^+\left(\ff_{in}\right)$ \text{is large}}{\Return \upshape\textsc{EdgeMap-Dense}($\ff_{in}$, \condf, \edgef)}
  \Else{\Return \upshape\textsc{EdgeMap-Sparse}($\ff_{in}$, \condf, \edgef)}
}

\myfunc{\upshape\textsc{EdgeMap-Sparse}($\ff_{in}$, \condf, \edgef)}{
  $\ff_{out} = \emptyset$\\
  \parForEach{ $\{(u,v)~|~u\in \ff_{in}, v\in N^+(u)\}$}{
      \If{\upshape\condf{$(v)$}}{
         \lIf {\upshape\edgef{$(u,v)$}}{
            $\ff_{out} \gets \ff_{out}\cup \{v\}$ 
        }
      }
  }
  \Return {$\ff_{out}$}
}
\myfunc{\upshape\textsc{EdgeMap-Dense($\ff_{in}$, \condf, \edgef)}}{
  $\ff_{out} = \emptyset$\\
  \parForEach{$v \in V$}{
    \If {\upshape\condf{$(v)$}}{
      \For{$u\in N^-(v)$}{
        \If {\upshape$u \in \ff_{in}$ and \edgef{$(u,v)$}}{
             $\ff_{out} \gets \ff_{out}\cup \{v\}$\\
          break \label{line:exit_early}
        }
      }
    }
  }
  \Return {$\ff_{out}$}
}
\end{algorithm}

%% file: figs_algs/fig_approx_query.tex
\begin{figure*}[htbp]
    \centering
    \includegraphics[width=1.5\columnwidth]{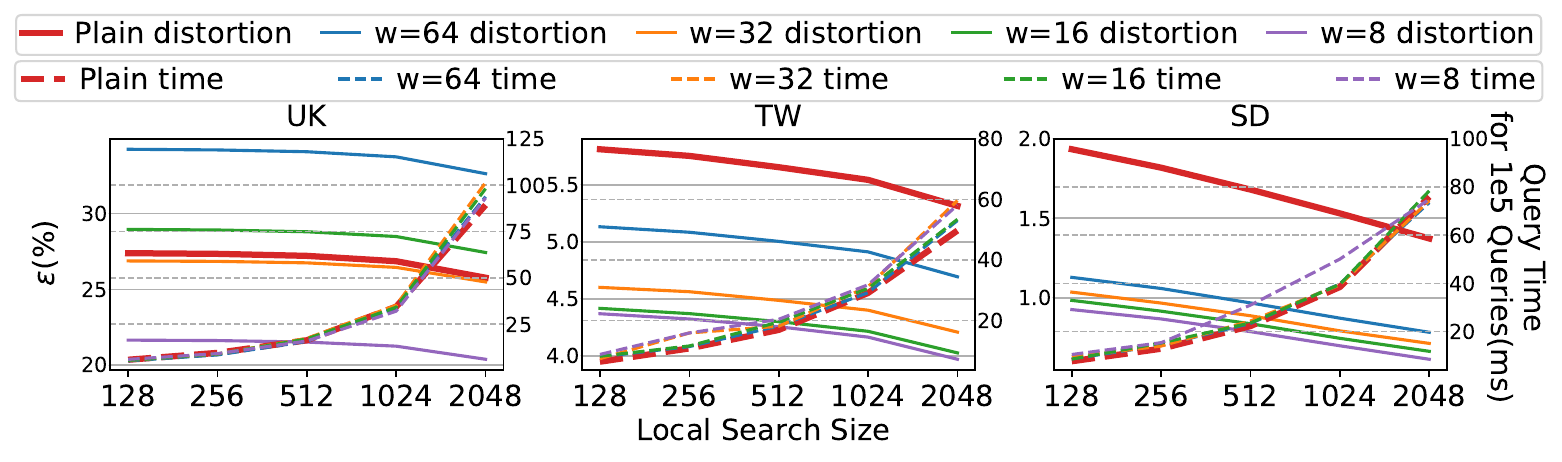}
    \caption{\small \textbf{Tradeoffs between local search size and distortion/query time.}
    With a memory limit of 1024 bytes per vertex for the index, we vary the local search size from 128 to 2048 to see the influence on distortion and query time. The $x$-axis is the local search size, which is the total number of vertices explored during the bi-directional BFS from two queried vertices. The $y$-axis on the left is the ($1+\epsilon$) distortion, corresponding to the solid lines in the figures. The $y$-axis on the right is the parallel query time for $10^4$ queries in milliseconds, corresponding to the dashed lines in the figures. For both distortion and preprocessing time, lower is better.  For algorithms compared here, `Plain' is the regular LL, and others are \ccbfs-based LL with cluster size $w$.
    }
    \label{fig:approx_query}
\end{figure*} 

%% file: figs_algs/table2_full.tex
\begin{tabular}{l|c@{ }c@{ }c@{ }c@{ }c|r@{ }r@{ }r@{ }r@{ }r|c@{ }c@{ }c@{ }c@{ }c@{ }}
\toprule
     & \multicolumn{5}{c|}{Index Time(s)} & \multicolumn{5}{c|}{$\epsilon$(\%)} & \multicolumn{5}{c}{Query Time(ms)} \\
Data &         Plain & $w=64$ & $w=32$ & $w=16$ & $w=8$ &          Plain & $w=64$ & $w=32$ & $w=16$ & $w=8$ &          Plain & $w=64$ & $w=32$ & $w=16$ & $w=8$ \\
\midrule
  EP &          1.26 &   0.02 &   0.04 &   0.05 &  0.08 &            0.4 &    0.1 &    0.1 &    0.1 &   0.1 &            2.4 &    2.1 &    2.0 &    1.7 &   2.4 \\
SLDT &          1.15 &   0.02 &   0.03 &   0.06 &  0.07 &            0.7 &    0.1 &    0.1 &    0.1 &   0.1 &            2.0 &    2.0 &    1.9 &    1.9 &   2.4 \\
DBLP &          3.57 &   0.08 &   0.11 &   0.16 &  0.25 &            2.5 &    2.2 &    1.4 &    1.1 &   1.0 &            2.0 &    1.8 &    1.9 &    1.9 &   2.5 \\
  YT &          9.22 &   0.23 &   0.27 &   0.37 &  0.59 &            0.3 &    0.3 &    0.2 &    0.2 &   0.1 &            1.7 &    1.8 &    2.1 &    2.1 &   2.5 \\
  SK &          13.4 &   0.55 &   0.81 &   1.14 &  1.77 &            1.4 &    0.7 &    0.5 &    0.4 &   0.4 &            1.7 &    1.8 &    2.0 &    2.0 &   2.7 \\
IN04 &          20.0 &   0.96 &   1.84 &   2.38 &  3.88 &            2.1 &    1.9 &    1.3 &    1.1 &   0.9 &            1.7 &    1.8 &    2.2 &    2.0 &   2.1 \\
  LJ &          36.2 &   1.72 &   2.48 &   3.49 &  5.63 &            5.0 &    4.3 &    3.7 &    3.6 &   3.5 &            1.8 &    1.8 &    2.0 &    2.0 &   2.3 \\
  HW &          12.4 &   0.93 &   1.73 &   2.49 &  4.10 &           10.6 &    5.6 &    5.9 &    6.5 &   7.1 &            1.7 &    1.8 &    2.1 &    2.1 &   2.4 \\
FBUU &           138 &   11.3 &   13.1 &   17.8 &  27.0 &            6.2 &   11.9 &    9.9 &    8.1 &   6.9 &            1.8 &    1.8 &    2.0 &    1.9 &   2.2 \\
FBKN &           127 &   10.5 &   12.1 &   16.3 &  24.9 &            6.2 &   11.9 &   10.0 &    8.2 &   6.9 &            1.8 &    1.8 &    2.0 &    1.9 &   2.2 \\
  OK &          26.3 &   2.87 &   4.53 &   6.43 &  10.1 &            8.7 &    7.7 &    7.6 &    7.3 &   7.3 &            1.8 &    1.8 &    2.0 &    2.0 &   2.3 \\
INDO &          83.2 &   5.44 &   11.1 &   18.0 &  29.8 &            3.1 &    1.5 &    1.3 &    1.2 &   1.3 &            1.8 &    1.8 &    2.0 &    1.9 &   2.0 \\
  EU &          87.3 &   7.01 &   14.5 &   21.9 &  34.9 &            2.6 &    1.3 &    1.0 &    1.2 &   1.7 &            1.8 &    1.8 &    2.0 &    1.9 &   2.0 \\
  UK &          80.4 &   8.28 &   15.8 &   22.1 &  38.8 &            3.9 &    4.9 &    3.8 &    4.1 &   3.1 &            1.8 &    1.8 &    2.0 &    1.9 &   2.1 \\
  AR &           148 &   17.6 &   36.6 &   54.9 &  86.8 &            2.6 &    4.0 &    3.2 &    2.2 &   2.2 &            1.8 &    1.8 &    2.0 &    1.9 &   2.0 \\
  TW &           112 &   31.0 &   47.0 &   62.7 &  99.3 &            1.5 &    1.4 &    1.2 &    1.1 &   1.1 &            1.8 &    1.8 &    1.9 &    1.9 &   2.3 \\
  FT &           251 &   61.1 &   88.2 &    120 &   193 &           16.8 &   12.4 &   12.0 &   12.1 &  12.8 &            1.8 &    1.8 &    2.0 &    1.9 &   2.4 \\
  SD &           318 &   75.6 &    125 &    163 &   255 &            0.6 &    0.3 &    0.3 &    0.3 &   0.3 &            1.8 &    1.8 &    1.9 &    1.9 &   2.4 \\
\bottomrule
\end{tabular}

%% file: application_2hop.tex
\subsection{Algorithm Description of Exact 2-Hop Distance Oracle}
\label{sec:2hopLL}

Our second application is an exact distance oracle, which always answers the correct shortest distance in queries. 
Many EDOs are based on the idea of 2-hop cover~\cite{CohenHKZ02,cheng2009line,abraham2012hierarchical}, described  as follows.

For each vertex $v$, 2-hop labeling methods select a subset of vertices $u\in V$ as \emph{hubs} for $v$, and precompute their distance $\delta(u,v)$. 
We call the precomputed distances for a vertex $v$ as the \emph{label} of $v$, and denote it as $L(v)$: a set of pairs $(u,\delta(u,v))$ for each of $v$'s hub $u$. 
Note that the hubs can be different for different vertices. 
Then, $query(s,t)$ finds the shortest distance passing through the intersection of their hubs, 
% i.e. $\min\{\delta(s,v)+\delta(t,v) | (v,\cdot)\in L(s), (v,\cdot) \in L(t)\}$. 
i.e., $\min\{\delta(s,v)+\delta(t,v) ~|~ (v,\cdot)\in L(s) \cap L(t)\}$.
We call $L$ a \emph{2-hop cover} of $G$ if $query(s,t)$ can correctly answer the distance between any pair of vertices. 
Finding a small 2-hop cover efficiently is a long-standing challenge~\cite{CohenHKZ02,cheng2009line,abraham2012hierarchical}. 

One of the state-of-the-art approaches is \defn{Pruned Landmark Labeling (PLL)}~\cite{akiba2013fast}. 
Given a graph $G$ and a vertex order $v_1,v_2,\dots, v_n$, PLL runs BFSs (with pruning, introduced below) 
from vertices in order and construct the labels. 
PLL starts with an empty index $L_0$, where $L_0(v)=\emptyset$ for every $v\in V$. 
In round $i$, PLL conducts a BFS from vertex $v_i$, and adds distances from $v_i$ to labels of reached vertices, 
that is, $L_i(u) = L_{i-1}(u)\cup \{(v_i,\delta(v_i, u))\}$ for each $u$ that $v_i$ can reach. 
%Further more, PLL prunes the BFS to reduce the number of labels added and also guarantee the labels is a 2-hop cover. 
To minimize the size of the index, PLL \emph{prunes} unnecessary labels and searches during each BFS from $v_i$: 
%before adding $v_i$ as a label of any vertex $u$, PLL checks if the 
when $v_i$ visits $u$, PLL checks if the 
existing index can already report the distance between $v_i$ and $u$. 
If so, the BFS will skip $u$, and $v_i$ will not be added to the labels of $u$. 
%otherwise, the distance pair $(v_i,\delta(v_i,u))$ will be added to the label of $u$ and $u$ will be expanede righ away. 
%The PLL is proved to be able to correctly answer the shortest distance between any vertex pairs in Theorem 4.1 in \cite{}, and shown the index constructed is minimal in Theorem 4.2 in \cite{}. 
PLL is proved to be a correct 2-hop cover, and the index constructed is minimal~\cite{akiba2013fast}.

To further speed up both the preprocessing and querying, PLL~\cite{akiba2013fast} applies \ccbfs{}. 
Instead of running pruned BFS, in the first several rounds, \ccbfs{} is used (without any pruning),
such that all sources in the clusters will be added to the labels of all vertices in $V$. 
The intuition is that, 
in the first several rounds, PLL can hardly prune any vertices since the index size is small.
Therefore, we can ignore the pruning but use clusters to improve the performance. 
%so that we can use \batchBFS for the first several rounds to take the advantage of compaction and parallelism brought by \batch{s}. 
As a result, the entire PLL algorithm in \cite{akiba2013fast} has two phases: 1) several rounds of \ccbfs{s} %, which is similar to the ADO based on landmark labeling mentioned in \cref{sec:approxLL}, 
and 2) pruned BFSs on the rest of the vertices. 
%Although PLL has been widely studied and used in practice, we are unaware of any parallel version of it. \guy{is this true?} \letong{\cite{li2019scaling} is a parallel PLL work}
%The long construction time in the sequential setting severely limits the graph size that PLL is applicable to. 
%In our experiments, the sequential algorithm fails to construct the index within 3 hours on a graph with 11M vertices. 
In our paper, we develop a parallel version for both \ccbfs{} and pruned BFS to the sequential algorithm in \cite{akiba2013fast}. 

\myparagraph{Our Implementation.} We apply our parallel \ccbfs{} to PLL, 
and also provide a parallel implementation for the pruned BFS. 
%We parallel both \batchBFS and pruned BFSs. 
We directly apply our parallel \ccbfs{} mentioned in \cref{sec:parallel_batchBFS} with the optimizations mentioned in \cref{sec:applications} to the first phase. 
Note that in the original PLL algorithm, 
the second phase incurs running pruned BFS from almost all vertices one by one. 
It is essential to parallelize this part to achieve high performance of the entire process of PLL. 
Ideally, we want to 1) run BFSs from multiple sources in parallel, 
but also 2) make them see as much of the index constructed by other vertices to enable effective pruning. 
To do this, we use a prefix-doubling-like scheme~\cite{blelloch2016parallelism,blelloch2020randomized,SGBFG15}, 
which splits the vertices into batches of exponentially growing sizes, 
and we parallelize the BFSs within each batch. 
The batches will be executed one by one from the smallest one. 
We empirically set the size of the $i$-th batch as $200\times (1.5)^i$, and stop increasing the batch size when the batch size is large enough (1000 in our implementation) for sufficient parallelism.
With both phases well-parallelized, our parallel version improves the preprocessing time of the original sequential code from~\cite{akiba2013fast} by up to 36.5$\times$,
and it can process much larger graphs than the sequential algorithm.

%% file: exp_2hop.tex
\subsection{Experiments on Exact 2-Hop Distance Oracle} 
\label{sec:exp:edo}

%\begin{table*}[htbp]
%  \centering
%  \small
%  \input{figs_algs/table_pll.tex}
%  \label{table:pll}
%  \caption{
%    .
%  }
%\end{table*}
% &  & \multicolumn{2}{c|}{avg. \# normal labels} &    \multicolumn{1}{c|}{Index} & \multicolumn{2}{c|}{CCBFS time} & \multicolumn{2}{c|}{Pruned BFS time} & Total Index & Total Query \\ 

% Table generated by Excel2LaTeX from sheet 'PLL final'
% Table generated by Excel2LaTeX from sheet 'PLL final'
% Table generated by Excel2LaTeX from sheet 'PLL final'
\begin{table}[htbp]
  \centering
  \footnotesize

  \begin{tabular}{@{}l@{}r@{ }r@{}r@{}r|@{}r@{  }@{  }r@{  }@{  }r}
    %\toprule
          &       &       & \multicolumn{1}{r}{\textbf{avg.}} & \multicolumn{1}{c|}{\textbf{Index}} & \multicolumn{3}{c}{\textbf{Running Time}} \\
          & \multicolumn{1}{c}{\textbf{$r$}} &   \textbf{Alg.}    & \multicolumn{1}{r}{\textbf{labs}} & \multicolumn{1}{c|}{\textbf{Size}} & \textbf{C-BFS} & \textbf{P-BFS} & \textbf{Total} \\
          \midrule
    % \multicolumn{1}{@{}l}{\multirow{3}[2]{*}{\textbf{Wiki}}} & \multicolumn{1}{c}{\multirow{3}[2]{*}{16}} & AYY & 33.9 & 0.98  & 4.88  & 48.4 & 54.2 \\
    % \multicolumn{1}{@{}c}{} & \multicolumn{1}{c}{} & Ours  & 34.9 & 1.00  & 0.27  & 8.20  & 8.57 \\
    % \multicolumn{1}{@{}c}{} & \multicolumn{1}{c}{} & Spd &  &  & 18.2$\times$ & 5.89$\times$ & 6.32$\times$ \\
    % \midrule
    \multicolumn{1}{@{}l}{\multirow{3}[2]{*}{\textbf{SK}}} & \multicolumn{1}{c}{\multirow{3}[2]{*}{64}} & AIY & 123 & 2.68  & 32.9 & 265 & 299 \\
    \multicolumn{1}{@{}c}{} & \multicolumn{1}{c}{} & Ours  & 126 & 2.71  & 1.05  & 15.0 & 16.2 \\
    \multicolumn{1}{@{}c}{} & \multicolumn{1}{c}{} & Spd &  &  & 31.3$\times$ & 17.7$\times$ & 18.5$\times$ \\
    % \midrule
    % \multicolumn{1}{@{}l}{\multirow{3}[2]{*}{\textbf{FL}}} & \multicolumn{1}{c}{\multirow{3}[2]{*}{64}} & AYY & 268 & 3.88  & 40.5 & 602 & 644 \\
    % \multicolumn{1}{@{}c}{} & \multicolumn{1}{c}{} & Ours  & 270 & 3.90  & 1.02  & 24.0 & 25.2 \\
    % \multicolumn{1}{@{}c}{} & \multicolumn{1}{c}{} & Spd &  &  & 39.8$\times$ & 25.1$\times$ & 25.6$\times$ \\
    \midrule
    \multicolumn{1}{@{}l}{\multirow{3}[2]{*}{\textbf{HW}}} & \multicolumn{1}{c}{\multirow{3}[2]{*}{64}} & AIY & 2237 & 12.2 & 103 & 11010 & 11115 \\
    \multicolumn{1}{@{}c}{} & \multicolumn{1}{c}{} & Ours  & 2280 & 12.4 & 1.26  & 303 & 304 \\
    \multicolumn{1}{@{}c}{} & \multicolumn{1}{c}{} & Spd &  &  & 82.3$\times$ & 36.4$\times$ & 36.5$\times$ \\
    \midrule
    \multicolumn{1}{@{}l}{\multirow{3}[2]{*}{\textbf{INDO}}} & \multicolumn{1}{@{}c@{}}{\multirow{3}[2]{*}{64}} & AIY & 323 & 18.7 & 246 & 3740 & 4051 \\
    \multicolumn{1}{@{}c}{} & \multicolumn{1}{c}{} & Ours  & 349 & 19.6 & 5.63  & 421 & 428 \\
    \multicolumn{1}{@{}c}{} & \multicolumn{1}{c}{} & Spd &  &  & 43.7$\times$ & 9.02$\times$ & 9.46$\times$ \\
    \midrule
    \midrule
    \multicolumn{1}{@{}l}{\textbf{EU}} & \multicolumn{1}{@{}c@{}}{64} & Ours  & 944 & 61.0 & 9.92  & 1385 & 1396 \\
    \multicolumn{1}{@{}l}{\textbf{LJ}} & \multicolumn{1}{@{}c@{}}{512} & Ours  & 2585 & 97.7 & 23.0 & 2718 & 2742 \\
    \multicolumn{1}{@{}l}{\textbf{AR}} & \multicolumn{1}{@{}c@{}}{256} & Ours  & 989 & 197 & 72.7 & 7690 & 7767 \\
    \multicolumn{1}{@{}l}{\textbf{OK}} &
    \multicolumn{1}{@{}c@{}}{2048} & Ours & 6881 & 198 & 119 & 13407 & 13527 \\
    \bottomrule
    \end{tabular}%
      \caption{\textbf{Performance on an exact distance oracle based on pruned landmark labeling.} 
      ``AIY'': the sequential implementation from ~\cite{akiba2013fast}. 
      $r$: the number of clusters used in \ccbfs{}. 
      Index sizes are in GB. 
      ``\ccbfs'': time for \batchBFS{}. ``P-BFS'': time for pruned BFS. ``Spd'': speed-up of ours over AIY. For \#labels/vertex, index size, and running time, lower is better. See more details in \cref{sec:2hopLL}.}
  \label{table:pll}%
\end{table}%

We also apply our \ccbfs{} to an exact distance oracle using the pruned landmark labeling (PLL) described in \cref{sec:2hopLL}. 
We follow the high-level idea in \cite{akiba2013fast}, which consists of a \batchBFS{} (\emp{\ccbfs{}}) phase on $r$ clusters and
a pruned BFS (\emp{P-BFS}) phase on the rest of vertices in $V$. 
We compare our parallel implementation with the original sequential code provided in \cite{akiba2013fast} (referred to as the \emph{AIY algorithm}).
Note that due to the need to report the \emph{exact} distance, the index size is much larger than the ADOs reported in \cref{sec:exp_appx}.
For the baseline algorithm (AIY), the largest graph it can process is INDO with 7.41M vertices and 301M edges.
Because of better parallelism, our \ccbfs{} can scale to much larger graphs.
We show four graphs (EU, LJ, AR, OK) that can be processed by our parallel algorithm but not the sequential version.
The largest graph includes 22.7M vertices and 1.11B edges. 

To choose the parameter $r$ (number of cluster searches), for all graphs that have been tested in \cite{akiba2013fast}, we use the same
value $r$ as they reported giving the almost best memory usage. 
For other graphs, we test a wide range of $r$ and present the overall best performance considering both space and preprocessing time.
We present the parameter $r$ and running time for each graph in \cref{table:pll}.

As mentioned, our algorithm may result in more labels (thus larger sizes) over the original AIY algorithm,
because of running BFS in batches in the pruned BFS phase: the vertices in the same batch
may not be able to see and use each others' labels for pruning, and thus more labels may be added.
In our results, such a loss is reasonably small. On all tested graphs, it is at most 8\% of the number of labels and at most 5\% more of the total index size.

For the running time, the time on P-BFS always dominates the cost, both in sequential and in parallel.
Our algorithm parallelizes both steps well, with better speedup on the \ccbfs{} phase.
Note that although \ccbfs{} is not the major cost of the sequential PLL, 
without parallelizing, it will make its cost comparable to or even larger than parallel P-BFS.
Therefore, it is important to combine \ccbfs{} with parallelism and make this part negligible in the parallel running time. 
In total, on the five small graphs, our algorithm achieves 6.3--36.5$\times$ speedup over the sequential AIY algorithm.
The advantage of our algorithm is more significant with more expensive sequential running time. 

On four larger graphs, our algorithm generates the index in 4 hours, and scales to graph AR with up to 22.7M vertices and 1.11B edges. 
Using a similar amount of time, the sequential AIY code can only process a much
smaller graph (HW) with 1.07M vertices and 112M edges, about one order of magnitude smaller.

%% file: main.bbl

\begin{thebibliography}{56}

%%% the Web address.

\ifx \showCODEN    \undefined \def \showCODEN     #1{\unskip}     \fi
\ifx \showDOI      \undefined \def \showDOI       #1{#1}\fi
\ifx \showISBNx    \undefined \def \showISBNx     #1{\unskip}     \fi
\ifx \showISBNxiii \undefined \def \showISBNxiii  #1{\unskip}     \fi
\ifx \showISSN     \undefined \def \showISSN      #1{\unskip}     \fi
\ifx \showLCCN     \undefined \def \showLCCN      #1{\unskip}     \fi
\ifx \shownote     \undefined \def \shownote      #1{#1}          \fi
\ifx \showarticletitle \undefined \def \showarticletitle #1{#1}   \fi
\ifx \showURL      \undefined \def \showURL       {\relax}        \fi
% invisible to TeX
\providecommand\bibfield[2]{#2}
\providecommand\bibinfo[2]{#2}
\providecommand\natexlab[1]{#1}
\providecommand\showeprint[2][]{arXiv:#2}

\bibitem[\protect\citeauthoryear{Abraham, Delling, Goldberg, and
  Werneck}{Abraham et~al\mbox{.}}{2012}]%
        {abraham2012hierarchical}
\bibfield{author}{\bibinfo{person}{Ittai Abraham}, \bibinfo{person}{Daniel
  Delling}, \bibinfo{person}{Andrew~V Goldberg}, {and}
  \bibinfo{person}{Renato~F Werneck}.} \bibinfo{year}{2012}\natexlab{}.
\newblock \showarticletitle{Hierarchical hub labelings for shortest paths}. In
  \bibinfo{booktitle}{\emph{European Symposium on Algorithms (ESA)}}. Springer,
  \bibinfo{pages}{24--35}.
\newblock


\bibitem[\protect\citeauthoryear{Akiba, Iwata, and Yoshida}{Akiba
  et~al\mbox{.}}{2013}]%
        {akiba2013fast}
\bibfield{author}{\bibinfo{person}{Takuya Akiba}, \bibinfo{person}{Yoichi
  Iwata}, {and} \bibinfo{person}{Yuichi Yoshida}.}
  \bibinfo{year}{2013}\natexlab{}.
\newblock \showarticletitle{Fast exact shortest-path distance queries on large
  networks by pruned landmark labeling}. In \bibinfo{booktitle}{\emph{ACM
  SIGMOD International Conference on Management of Data (SIGMOD)}}.
  \bibinfo{pages}{349--360}.
\newblock


\bibitem[\protect\citeauthoryear{Akiba, Sommer, and Kawarabayashi}{Akiba
  et~al\mbox{.}}{2012}]%
        {akiba2012shortest}
\bibfield{author}{\bibinfo{person}{Takuya Akiba}, \bibinfo{person}{Christian
  Sommer}, {and} \bibinfo{person}{Ken-ichi Kawarabayashi}.}
  \bibinfo{year}{2012}\natexlab{}.
\newblock \showarticletitle{Shortest-path queries for complex networks:
  exploiting low tree-width outside the core}. In \bibinfo{booktitle}{\emph{ACM
  International Conference on Extending Database Technology}}.
  \bibinfo{pages}{144--155}.
\newblock


\bibitem[\protect\citeauthoryear{Backstrom, Huttenlocher, Kleinberg, and
  Lan}{Backstrom et~al\mbox{.}}{2006}]%
        {backstrom2006group}
\bibfield{author}{\bibinfo{person}{Lars Backstrom}, \bibinfo{person}{Dan
  Huttenlocher}, \bibinfo{person}{Jon Kleinberg}, {and}
  \bibinfo{person}{Xiangyang Lan}.} \bibinfo{year}{2006}\natexlab{}.
\newblock \showarticletitle{Group formation in large social networks:
  membership, growth, and evolution}. In \bibinfo{booktitle}{\emph{ACM
  International Conference on Knowledge Discovery and Data Mining (SIGKDD)}}.
  \bibinfo{pages}{44--54}.
\newblock


\bibitem[\protect\citeauthoryear{Bast, Delling, Goldberg,
  {\"{u}}ller{-}Hannemann, Pajor, Sanders, Wagner, and Werneck}{Bast
  et~al\mbox{.}}{2016}]%
        {BastDGMPSWW16}
\bibfield{author}{\bibinfo{person}{Hannah Bast}, \bibinfo{person}{Daniel
  Delling}, \bibinfo{person}{Andrew~V. Goldberg}, \bibinfo{person}{Matthias~M
  {\"{u}}ller{-}Hannemann}, \bibinfo{person}{Thomas Pajor},
  \bibinfo{person}{Peter Sanders}, \bibinfo{person}{Dorothea Wagner}, {and}
  \bibinfo{person}{Renato~F. Werneck}.} \bibinfo{year}{2016}\natexlab{}.
\newblock \showarticletitle{Route Planning in Transportation Networks}.
\newblock In \bibinfo{booktitle}{\emph{Algorithm Engineering - Selected Results
  and Surveys}}, \bibfield{editor}{\bibinfo{person}{Lasse Kliemann} {and}
  \bibinfo{person}{Peter Sanders}} (Eds.). \bibinfo{series}{Lecture Notes in
  Computer Science}, Vol.~\bibinfo{volume}{9220}. \bibinfo{pages}{19--80}.
\newblock
\urldef\tempurl%
\url{https://doi.org/10.1007/978-3-319-49487-6\_2}
\showDOI{\tempurl}


\bibitem[\protect\citeauthoryear{Bast, Funke, Sanders, and Schultes}{Bast
  et~al\mbox{.}}{2007}]%
        {bast2007fast}
\bibfield{author}{\bibinfo{person}{Holger Bast}, \bibinfo{person}{Stefan
  Funke}, \bibinfo{person}{Peter Sanders}, {and} \bibinfo{person}{Dominik
  Schultes}.} \bibinfo{year}{2007}\natexlab{}.
\newblock \showarticletitle{Fast routing in road networks with transit nodes}.
\newblock \bibinfo{journal}{\emph{Science}} \bibinfo{volume}{316},
  \bibinfo{number}{5824} (\bibinfo{year}{2007}), \bibinfo{pages}{566--566}.
\newblock


\bibitem[\protect\citeauthoryear{Beamer, Asanovi\'{c}, and Patterson}{Beamer
  et~al\mbox{.}}{2012}]%
        {Beamer12}
\bibfield{author}{\bibinfo{person}{Scott Beamer}, \bibinfo{person}{Krste
  Asanovi\'{c}}, {and} \bibinfo{person}{David Patterson}.}
  \bibinfo{year}{2012}\natexlab{}.
\newblock \showarticletitle{Direction-optimizing breadth-first search}. In
  \bibinfo{booktitle}{\emph{International Conference for High Performance
  Computing, Networking, Storage, and Analysis (SC)}}. \bibinfo{pages}{1--10}.
\newblock


\bibitem[\protect\citeauthoryear{Blelloch, Anderson, and Dhulipala}{Blelloch
  et~al\mbox{.}}{2020a}]%
        {blelloch2020parlaylib}
\bibfield{author}{\bibinfo{person}{Guy~E. Blelloch}, \bibinfo{person}{Daniel
  Anderson}, {and} \bibinfo{person}{Laxman Dhulipala}.}
  \bibinfo{year}{2020}\natexlab{a}.
\newblock \showarticletitle{ParlayLib --- a toolkit for parallel algorithms on
  shared-memory multicore machines}. In \bibinfo{booktitle}{\emph{{ACM}
  Symposium on Parallelism in Algorithms and Architectures (SPAA)}}.
  \bibinfo{pages}{507--509}.
\newblock


\bibitem[\protect\citeauthoryear{Blelloch, Fineman, Gu, and Sun}{Blelloch
  et~al\mbox{.}}{2020b}]%
        {blelloch2020optimal}
\bibfield{author}{\bibinfo{person}{Guy~E. Blelloch}, \bibinfo{person}{Jeremy~T.
  Fineman}, \bibinfo{person}{Yan Gu}, {and} \bibinfo{person}{Yihan Sun}.}
  \bibinfo{year}{2020}\natexlab{b}.
\newblock \showarticletitle{Optimal parallel algorithms in the binary-forking
  model}. In \bibinfo{booktitle}{\emph{{ACM} Symposium on Parallelism in
  Algorithms and Architectures (SPAA)}}. \bibinfo{pages}{89--102}.
\newblock


\bibitem[\protect\citeauthoryear{Blelloch, Gu, Shun, and Sun}{Blelloch
  et~al\mbox{.}}{2020c}]%
        {blelloch2016parallelism}
\bibfield{author}{\bibinfo{person}{Guy~E. Blelloch}, \bibinfo{person}{Yan Gu},
  \bibinfo{person}{Julian Shun}, {and} \bibinfo{person}{Yihan Sun}.}
  \bibinfo{year}{2020}\natexlab{c}.
\newblock \showarticletitle{Parallelism in Randomized Incremental Algorithms}.
\newblock \bibinfo{journal}{\emph{J. {ACM}}} \bibinfo{volume}{67},
  \bibinfo{number}{5} (\bibinfo{year}{2020}), \bibinfo{pages}{1--27}.
\newblock


\bibitem[\protect\citeauthoryear{Blelloch, Gu, Shun, and Sun}{Blelloch
  et~al\mbox{.}}{2020d}]%
        {blelloch2020randomized}
\bibfield{author}{\bibinfo{person}{Guy~E. Blelloch}, \bibinfo{person}{Yan Gu},
  \bibinfo{person}{Julian Shun}, {and} \bibinfo{person}{Yihan Sun}.}
  \bibinfo{year}{2020}\natexlab{d}.
\newblock \showarticletitle{Randomized Incremental Convex Hull is Highly
  Parallel}. In \bibinfo{booktitle}{\emph{{ACM} Symposium on Parallelism in
  Algorithms and Architectures (SPAA)}}.
\newblock


\bibitem[\protect\citeauthoryear{Blelloch, Gu, and Sun}{Blelloch
  et~al\mbox{.}}{2017}]%
        {blelloch2017efficient}
\bibfield{author}{\bibinfo{person}{Guy~E. Blelloch}, \bibinfo{person}{Yan Gu},
  {and} \bibinfo{person}{Yihan Sun}.} \bibinfo{year}{2017}\natexlab{}.
\newblock \showarticletitle{Efficient Construction of Probabilistic Tree
  Embeddings}. In \bibinfo{booktitle}{\emph{Intl. Colloq. on Automata,
  Languages and Programming {(ICALP)}}}.
\newblock


\bibitem[\protect\citeauthoryear{Blumofe and Leiserson}{Blumofe and
  Leiserson}{1999}]%
        {blumofe1999scheduling}
\bibfield{author}{\bibinfo{person}{Robert~D. Blumofe} {and}
  \bibinfo{person}{Charles~E. Leiserson}.} \bibinfo{year}{1999}\natexlab{}.
\newblock \showarticletitle{Scheduling multithreaded computations by work
  stealing}.
\newblock \bibinfo{journal}{\emph{J. {ACM}}} \bibinfo{volume}{46},
  \bibinfo{number}{5} (\bibinfo{year}{1999}), \bibinfo{pages}{720--748}.
\newblock


\bibitem[\protect\citeauthoryear{Boldi, Codenotti, Santini, and Vigna}{Boldi
  et~al\mbox{.}}{2004}]%
        {boldi2004-ubicrawler}
\bibfield{author}{\bibinfo{person}{Paolo Boldi}, \bibinfo{person}{Bruno
  Codenotti}, \bibinfo{person}{Massimo Santini}, {and}
  \bibinfo{person}{Sebastiano Vigna}.} \bibinfo{year}{2004}\natexlab{}.
\newblock \showarticletitle{{UbiCrawler}: A Scalable Fully Distributed Web
  Crawler}.
\newblock \bibinfo{journal}{\emph{spe}} \bibinfo{volume}{34},
  \bibinfo{number}{8} (\bibinfo{year}{2004}), \bibinfo{pages}{711--726}.
\newblock


\bibitem[\protect\citeauthoryear{Boldi, Marino, Santini, and Vigna}{Boldi
  et~al\mbox{.}}{2014}]%
        {BMSB}
\bibfield{author}{\bibinfo{person}{Paolo Boldi}, \bibinfo{person}{Andrea
  Marino}, \bibinfo{person}{Massimo Santini}, {and} \bibinfo{person}{Sebastiano
  Vigna}.} \bibinfo{year}{2014}\natexlab{}.
\newblock \showarticletitle{{BUbiNG}: Massive Crawling for the Masses}. In
  \bibinfo{booktitle}{\emph{www}}. \bibinfo{publisher}{International World Wide
  Web Conferences Steering Committee}, \bibinfo{pages}{227--228}.
\newblock


\bibitem[\protect\citeauthoryear{Boldi, Rosa, Santini, and Vigna}{Boldi
  et~al\mbox{.}}{2011}]%
        {Boldi-2011-layered}
\bibfield{author}{\bibinfo{person}{Paolo Boldi}, \bibinfo{person}{Marco Rosa},
  \bibinfo{person}{Massimo Santini}, {and} \bibinfo{person}{Sebastiano Vigna}.}
  \bibinfo{year}{2011}\natexlab{}.
\newblock \showarticletitle{Layered Label Propagation: A MultiResolution
  Coordinate-Free Ordering for Compressing Social Networks}. In
  \bibinfo{booktitle}{\emph{www}}. \bibinfo{pages}{587--596}.
\newblock


\bibitem[\protect\citeauthoryear{Boldi and Vigna}{Boldi and Vigna}{2004}]%
        {BoVWFI}
\bibfield{author}{\bibinfo{person}{Paolo Boldi} {and}
  \bibinfo{person}{Sebastiano Vigna}.} \bibinfo{year}{2004}\natexlab{}.
\newblock \showarticletitle{The {W}eb{G}raph Framework {I}: {C}ompression
  Techniques}. In \bibinfo{booktitle}{\emph{www}}. \bibinfo{publisher}{ACM
  Press}, \bibinfo{address}{Manhattan, USA}, \bibinfo{pages}{595--601}.
\newblock


\bibitem[\protect\citeauthoryear{{CAIDA}}{{CAIDA}}{[n.d.]}]%
        {skitter}
\bibfield{author}{\bibinfo{person}{{CAIDA}}.}
  \bibinfo{year}{[n.d.]}\natexlab{}.
\newblock \bibinfo{title}{Skitter}.
\newblock
\newblock
\newblock
\shownote{{\scriptsize \url{http://caida.org/tools/measurement/skitter/}}.}


\bibitem[\protect\citeauthoryear{Chan}{Chan}{2012}]%
        {chan2012all}
\bibfield{author}{\bibinfo{person}{Timothy~M Chan}.}
  \bibinfo{year}{2012}\natexlab{}.
\newblock \showarticletitle{All-pairs shortest paths for unweighted undirected
  graphs in o (mn) time}.
\newblock \bibinfo{journal}{\emph{talg}} \bibinfo{volume}{8},
  \bibinfo{number}{4} (\bibinfo{year}{2012}), \bibinfo{pages}{1--17}.
\newblock


\bibitem[\protect\citeauthoryear{Cheng and Yu}{Cheng and Yu}{2009}]%
        {cheng2009line}
\bibfield{author}{\bibinfo{person}{Jiefeng Cheng} {and}
  \bibinfo{person}{Jeffrey~Xu Yu}.} \bibinfo{year}{2009}\natexlab{}.
\newblock \showarticletitle{On-line exact shortest distance query processing}.
  In \bibinfo{booktitle}{\emph{ACM International Conference on Extending
  Database Technology}}. \bibinfo{pages}{481--492}.
\newblock


\bibitem[\protect\citeauthoryear{Cohen, Halperin, Kaplan, and Zwick}{Cohen
  et~al\mbox{.}}{2002}]%
        {CohenHKZ02}
\bibfield{author}{\bibinfo{person}{Edith Cohen}, \bibinfo{person}{Eran
  Halperin}, \bibinfo{person}{Haim Kaplan}, {and} \bibinfo{person}{Uri Zwick}.}
  \bibinfo{year}{2002}\natexlab{}.
\newblock \showarticletitle{Reachability and Distance Queries Via 2-Hop
  Labels}.
\newblock \bibinfo{journal}{\emph{SIAM J. Comput.}}  \bibinfo{volume}{32}
  (\bibinfo{date}{01} \bibinfo{year}{2002}).
\newblock
\urldef\tempurl%
\url{https://doi.org/10.1137/S0097539702403098}
\showDOI{\tempurl}


\bibitem[\protect\citeauthoryear{Cormen, Leiserson, Rivest, and Stein}{Cormen
  et~al\mbox{.}}{2009}]%
        {CLRS}
\bibfield{author}{\bibinfo{person}{Thomas~H. Cormen},
  \bibinfo{person}{Charles~E. Leiserson}, \bibinfo{person}{Ronald~L. Rivest},
  {and} \bibinfo{person}{Clifford Stein}.} \bibinfo{year}{2009}\natexlab{}.
\newblock \bibinfo{booktitle}{\emph{Introduction to Algorithms (3rd edition)}}.
\newblock \bibinfo{publisher}{MIT Press}.
\newblock


\bibitem[\protect\citeauthoryear{Das~Sarma, Gollapudi, Najork, and
  Panigrahy}{Das~Sarma et~al\mbox{.}}{2010}]%
        {das2010sketch}
\bibfield{author}{\bibinfo{person}{Atish Das~Sarma}, \bibinfo{person}{Sreenivas
  Gollapudi}, \bibinfo{person}{Marc Najork}, {and} \bibinfo{person}{Rina
  Panigrahy}.} \bibinfo{year}{2010}\natexlab{}.
\newblock \showarticletitle{A sketch-based distance oracle for web-scale
  graphs}. In \bibinfo{booktitle}{\emph{wsdm}}. \bibinfo{pages}{401--410}.
\newblock


\bibitem[\protect\citeauthoryear{Dhulipala, Blelloch, and Shun}{Dhulipala
  et~al\mbox{.}}{2021}]%
        {gbbs2021}
\bibfield{author}{\bibinfo{person}{Laxman Dhulipala}, \bibinfo{person}{Guy~E.
  Blelloch}, {and} \bibinfo{person}{Julian Shun}.}
  \bibinfo{year}{2021}\natexlab{}.
\newblock \showarticletitle{Theoretically efficient parallel graph algorithms
  can be fast and scalable}.
\newblock \bibinfo{journal}{\emph{{ACM} Transactions on Parallel Computing
  (TOPC)}} \bibinfo{volume}{8}, \bibinfo{number}{1} (\bibinfo{year}{2021}),
  \bibinfo{pages}{1--70}.
\newblock


\bibitem[\protect\citeauthoryear{Geisberger, Sanders, Schultes, and
  Vetter}{Geisberger et~al\mbox{.}}{2012}]%
        {geisberger2012exact}
\bibfield{author}{\bibinfo{person}{Robert Geisberger}, \bibinfo{person}{Peter
  Sanders}, \bibinfo{person}{Dominik Schultes}, {and}
  \bibinfo{person}{Christian Vetter}.} \bibinfo{year}{2012}\natexlab{}.
\newblock \showarticletitle{Exact routing in large road networks using
  contraction hierarchies}.
\newblock \bibinfo{journal}{\emph{Transportation Science}}
  \bibinfo{volume}{46}, \bibinfo{number}{3} (\bibinfo{year}{2012}),
  \bibinfo{pages}{388--404}.
\newblock


\bibitem[\protect\citeauthoryear{Gu, Napier, and Sun}{Gu et~al\mbox{.}}{2022}]%
        {gu2022analysis}
\bibfield{author}{\bibinfo{person}{Yan Gu}, \bibinfo{person}{Zachary Napier},
  {and} \bibinfo{person}{Yihan Sun}.} \bibinfo{year}{2022}\natexlab{}.
\newblock \showarticletitle{Analysis of Work-Stealing and Parallel Cache
  Complexity}. In \bibinfo{booktitle}{\emph{{SIAM} Symposium on Algorithmic
  Principles of Computer Systems (APOCS)}}. SIAM, \bibinfo{pages}{46--60}.
\newblock


\bibitem[\protect\citeauthoryear{Gu, Obeya, and Shun}{Gu et~al\mbox{.}}{2021}]%
        {gu2021parallel}
\bibfield{author}{\bibinfo{person}{Yan Gu}, \bibinfo{person}{Omar Obeya}, {and}
  \bibinfo{person}{Julian Shun}.} \bibinfo{year}{2021}\natexlab{}.
\newblock \showarticletitle{Parallel In-Place Algorithms: Theory and Practice}.
  In \bibinfo{booktitle}{\emph{{SIAM} Symposium on Algorithmic Principles of
  Computer Systems (APOCS)}}. \bibinfo{pages}{114--128}.
\newblock


\bibitem[\protect\citeauthoryear{Gubichev, Bedathur, Seufert, and
  Weikum}{Gubichev et~al\mbox{.}}{2010}]%
        {gubichev2010fast}
\bibfield{author}{\bibinfo{person}{Andrey Gubichev}, \bibinfo{person}{Srikanta
  Bedathur}, \bibinfo{person}{Stephan Seufert}, {and} \bibinfo{person}{Gerhard
  Weikum}.} \bibinfo{year}{2010}\natexlab{}.
\newblock \showarticletitle{Fast and accurate estimation of shortest paths in
  large graphs}. In \bibinfo{booktitle}{\emph{ACM International Conference on
  Information and Knowledge Management}}. \bibinfo{pages}{499--508}.
\newblock


\bibitem[\protect\citeauthoryear{Kwak, Lee, Park, and Moon}{Kwak
  et~al\mbox{.}}{2010}]%
        {kwak2010twitter}
\bibfield{author}{\bibinfo{person}{Haewoon Kwak}, \bibinfo{person}{Changhyun
  Lee}, \bibinfo{person}{Hosung Park}, {and} \bibinfo{person}{Sue Moon}.}
  \bibinfo{year}{2010}\natexlab{}.
\newblock \showarticletitle{What is Twitter, a social network or a news
  media?}. In \bibinfo{booktitle}{\emph{International World Wide Web Conference
  (WWW)}}. \bibinfo{pages}{591--600}.
\newblock


\bibitem[\protect\citeauthoryear{Leskovec and Krevl}{Leskovec and
  Krevl}{2014}]%
        {leskovec2014snap}
\bibfield{author}{\bibinfo{person}{Jure Leskovec} {and} \bibinfo{person}{Andrej
  Krevl}.} \bibinfo{year}{2014}\natexlab{}.
\newblock \showarticletitle{{SNAP Datasets}: {Stanford} Large Network Dataset
  Collection}.
\newblock  (\bibinfo{year}{2014}).
\newblock


\bibitem[\protect\citeauthoryear{Leskovec, Lang, Dasgupta, and
  Mahoney}{Leskovec et~al\mbox{.}}{2009}]%
        {leskovec2009community}
\bibfield{author}{\bibinfo{person}{Jure Leskovec}, \bibinfo{person}{Kevin~J
  Lang}, \bibinfo{person}{Anirban Dasgupta}, {and} \bibinfo{person}{Michael~W
  Mahoney}.} \bibinfo{year}{2009}\natexlab{}.
\newblock \showarticletitle{Community structure in large networks: Natural
  cluster sizes and the absence of large well-defined clusters}.
\newblock \bibinfo{journal}{\emph{im}} \bibinfo{volume}{6}, \bibinfo{number}{1}
  (\bibinfo{year}{2009}), \bibinfo{pages}{29--123}.
\newblock


\bibitem[\protect\citeauthoryear{Li, Qiao, Qin, Zhang, Chang, and Lin}{Li
  et~al\mbox{.}}{2019}]%
        {li2019scaling}
\bibfield{author}{\bibinfo{person}{Wentao Li}, \bibinfo{person}{Miao Qiao},
  \bibinfo{person}{Lu Qin}, \bibinfo{person}{Ying Zhang},
  \bibinfo{person}{Lijun Chang}, {and} \bibinfo{person}{Xuemin Lin}.}
  \bibinfo{year}{2019}\natexlab{}.
\newblock \showarticletitle{Scaling distance labeling on small-world networks}.
  In \bibinfo{booktitle}{\emph{ACM SIGMOD International Conference on
  Management of Data (SIGMOD)}}. \bibinfo{pages}{1060--1077}.
\newblock


\bibitem[\protect\citeauthoryear{Li, U, Yiu, and Kou}{Li et~al\mbox{.}}{2017}]%
        {li2017experimental}
\bibfield{author}{\bibinfo{person}{Ye Li}, \bibinfo{person}{Leong~Hou U},
  \bibinfo{person}{Man~Lung Yiu}, {and} \bibinfo{person}{Ngai~Meng Kou}.}
  \bibinfo{year}{2017}\natexlab{}.
\newblock \showarticletitle{An experimental study on hub labeling based
  shortest path algorithms}.
\newblock \bibinfo{journal}{\emph{Proceedings of the VLDB Endowment}}
  \bibinfo{volume}{11}, \bibinfo{number}{4} (\bibinfo{year}{2017}),
  \bibinfo{pages}{445--457}.
\newblock


\bibitem[\protect\citeauthoryear{Meng, Kamara, Nissim, and Kollios}{Meng
  et~al\mbox{.}}{2015}]%
        {meng2015grecs}
\bibfield{author}{\bibinfo{person}{Xianrui Meng}, \bibinfo{person}{Seny
  Kamara}, \bibinfo{person}{Kobbi Nissim}, {and} \bibinfo{person}{George
  Kollios}.} \bibinfo{year}{2015}\natexlab{}.
\newblock \showarticletitle{Grecs: Graph encryption for approximate shortest
  distance queries}. In \bibinfo{booktitle}{\emph{ACM SIGSAC Conference on
  Computer and Communications Security}}. \bibinfo{pages}{504--517}.
\newblock


\bibitem[\protect\citeauthoryear{Meusel, Lehmberg, Bizer, and Vigna}{Meusel
  et~al\mbox{.}}{2014}]%
        {webgraph}
\bibfield{author}{\bibinfo{person}{Robert Meusel}, \bibinfo{person}{Oliver
  Lehmberg}, \bibinfo{person}{Christian Bizer}, {and}
  \bibinfo{person}{Sebastiano Vigna}.} \bibinfo{year}{2014}\natexlab{}.
\newblock \bibinfo{title}{Web Data Commons --- Hyperlink Graphs}.
\newblock
  \bibinfo{howpublished}{\url{http://webdatacommons.org/hyperlinkgraph}}.
\newblock


\bibitem[\protect\citeauthoryear{Miller, Peng, and Xu}{Miller
  et~al\mbox{.}}{2013}]%
        {miller2013parallel}
\bibfield{author}{\bibinfo{person}{Gary~L Miller}, \bibinfo{person}{Richard
  Peng}, {and} \bibinfo{person}{Shen~Chen Xu}.}
  \bibinfo{year}{2013}\natexlab{}.
\newblock \showarticletitle{Parallel graph decompositions using random shifts}.
  In \bibinfo{booktitle}{\emph{{ACM} Symposium on Parallelism in Algorithms and
  Architectures (SPAA)}}.
\newblock


\bibitem[\protect\citeauthoryear{Potamias, Bonchi, Castillo, and
  Gionis}{Potamias et~al\mbox{.}}{2009}]%
        {Potamias09}
\bibfield{author}{\bibinfo{person}{Michalis Potamias},
  \bibinfo{person}{Francesco Bonchi}, \bibinfo{person}{Carlos Castillo}, {and}
  \bibinfo{person}{Aristides Gionis}.} \bibinfo{year}{2009}\natexlab{}.
\newblock \showarticletitle{Fast shortest path distance estimation in large
  networks}. In \bibinfo{booktitle}{\emph{ACM International Conference on
  Information and Knowledge Management}}. \bibinfo{publisher}{Association for
  Computing Machinery}.
\newblock


\bibitem[\protect\citeauthoryear{Qiao, Cheng, Chang, and Yu}{Qiao
  et~al\mbox{.}}{2012}]%
        {qiao2012approximate}
\bibfield{author}{\bibinfo{person}{Miao Qiao}, \bibinfo{person}{Hong Cheng},
  \bibinfo{person}{Lijun Chang}, {and} \bibinfo{person}{Jeffrey~Xu Yu}.}
  \bibinfo{year}{2012}\natexlab{}.
\newblock \showarticletitle{Approximate shortest distance computing: A
  query-dependent local landmark scheme}.
\newblock \bibinfo{journal}{\emph{IEEE Transactions on Knowledge and Data
  Engineering}} \bibinfo{volume}{26}, \bibinfo{number}{1}
  (\bibinfo{year}{2012}), \bibinfo{pages}{55--68}.
\newblock


\bibitem[\protect\citeauthoryear{Rao}{Rao}{1999}]%
        {rao1999small}
\bibfield{author}{\bibinfo{person}{Satish Rao}.}
  \bibinfo{year}{1999}\natexlab{}.
\newblock \showarticletitle{Small distortion and volume preserving embeddings
  for planar and Euclidean metrics}. In \bibinfo{booktitle}{\emph{ACM Symposium
  on Computational Geometry (SoCG)}}. \bibinfo{pages}{300--306}.
\newblock


\bibitem[\protect\citeauthoryear{Red, Kelsic, Mucha, and Porter}{Red
  et~al\mbox{.}}{2011}]%
        {red2011comparing}
\bibfield{author}{\bibinfo{person}{Veronica Red}, \bibinfo{person}{Eric~D
  Kelsic}, \bibinfo{person}{Peter~J Mucha}, {and} \bibinfo{person}{Mason~A
  Porter}.} \bibinfo{year}{2011}\natexlab{}.
\newblock \showarticletitle{Comparing community structure to characteristics in
  online collegiate social networks}.
\newblock \bibinfo{journal}{\emph{SIAM review}} \bibinfo{volume}{53},
  \bibinfo{number}{3} (\bibinfo{year}{2011}), \bibinfo{pages}{526--543}.
\newblock


\bibitem[\protect\citeauthoryear{Rizi, Schloetterer, and Granitzer}{Rizi
  et~al\mbox{.}}{2018}]%
        {rizi2018shortest}
\bibfield{author}{\bibinfo{person}{Fatemeh~Salehi Rizi}, \bibinfo{person}{Joerg
  Schloetterer}, {and} \bibinfo{person}{Michael Granitzer}.}
  \bibinfo{year}{2018}\natexlab{}.
\newblock \showarticletitle{Shortest path distance approximation using deep
  learning techniques}. In \bibinfo{booktitle}{\emph{IEEE/ACM International
  Conference on Advances in Social Networks Analysis and Mining}}. IEEE,
  \bibinfo{pages}{1007--1014}.
\newblock


\bibitem[\protect\citeauthoryear{Rossi and Ahmed}{Rossi and Ahmed}{2015}]%
        {nr}
\bibfield{author}{\bibinfo{person}{Ryan~A. Rossi} {and}
  \bibinfo{person}{Nesreen~K. Ahmed}.} \bibinfo{year}{2015}\natexlab{}.
\newblock \showarticletitle{The Network Data Repository with Interactive Graph
  Analytics and Visualization}. In \bibinfo{booktitle}{\emph{aaai}}.
\newblock
\urldef\tempurl%
\url{https://networkrepository.com}
\showURL{%
\tempurl}


\bibitem[\protect\citeauthoryear{Shavitt and Tankel}{Shavitt and
  Tankel}{2008}]%
        {shavitt2008hyperbolic}
\bibfield{author}{\bibinfo{person}{Yuval Shavitt} {and} \bibinfo{person}{Tomer
  Tankel}.} \bibinfo{year}{2008}\natexlab{}.
\newblock \showarticletitle{Hyperbolic embedding of internet graph for distance
  estimation and overlay construction}.
\newblock \bibinfo{journal}{\emph{IEEE/ACM Transactions on Networking}}
  \bibinfo{volume}{16}, \bibinfo{number}{1} (\bibinfo{year}{2008}),
  \bibinfo{pages}{25--36}.
\newblock


\bibitem[\protect\citeauthoryear{Shun and Blelloch}{Shun and Blelloch}{2013}]%
        {shun2013ligra}
\bibfield{author}{\bibinfo{person}{Julian Shun} {and} \bibinfo{person}{Guy~E.
  Blelloch}.} \bibinfo{year}{2013}\natexlab{}.
\newblock \showarticletitle{Ligra: A Lightweight Graph Processing Framework for
  Shared Memory}. In \bibinfo{booktitle}{\emph{{ACM} Symposium on Principles
  and Practice of Parallel Programming (PPOPP)}}. \bibinfo{pages}{135--146}.
\newblock


\bibitem[\protect\citeauthoryear{Shun, Dhulipala, and Blelloch}{Shun
  et~al\mbox{.}}{2015a}]%
        {shun2015smaller}
\bibfield{author}{\bibinfo{person}{Julian Shun}, \bibinfo{person}{Laxman
  Dhulipala}, {and} \bibinfo{person}{Guy~E Blelloch}.}
  \bibinfo{year}{2015}\natexlab{a}.
\newblock \showarticletitle{Smaller and faster: Parallel processing of
  compressed graphs with Ligra+}. In \bibinfo{booktitle}{\emph{{IEEE Data
  Compression Conference (DCC)}}}. \bibinfo{pages}{403--412}.
\newblock


\bibitem[\protect\citeauthoryear{Shun, Gu, Blelloch, Fineman, and Gibbons}{Shun
  et~al\mbox{.}}{2015b}]%
        {SGBFG15}
\bibfield{author}{\bibinfo{person}{Julian Shun}, \bibinfo{person}{Yan Gu},
  \bibinfo{person}{Guy~E. Blelloch}, \bibinfo{person}{Jeremy~T. Fineman}, {and}
  \bibinfo{person}{Phillip~B. Gibbons}.} \bibinfo{year}{2015}\natexlab{b}.
\newblock \showarticletitle{Sequential Random Permutation, List Contraction and
  Tree Contraction are Highly Parallel}. In
  \bibinfo{booktitle}{\emph{{ACM-SIAM} Symposium on Discrete Algorithms
  (SODA)}}. \bibinfo{pages}{431--448}.
\newblock


\bibitem[\protect\citeauthoryear{Tang and Crovella}{Tang and Crovella}{2003}]%
        {tang2003virtual}
\bibfield{author}{\bibinfo{person}{Liying Tang} {and} \bibinfo{person}{Mark
  Crovella}.} \bibinfo{year}{2003}\natexlab{}.
\newblock \showarticletitle{Virtual landmarks for the internet}. In
  \bibinfo{booktitle}{\emph{ACM SIGCOMM conference on Internet measurement}}.
  \bibinfo{pages}{143--152}.
\newblock


\bibitem[\protect\citeauthoryear{Thorup and Zwick}{Thorup and Zwick}{2005}]%
        {thorup2005approximate}
\bibfield{author}{\bibinfo{person}{Mikkel Thorup} {and} \bibinfo{person}{Uri
  Zwick}.} \bibinfo{year}{2005}\natexlab{}.
\newblock \showarticletitle{Approximate distance oracles}.
\newblock \bibinfo{journal}{\emph{J. {ACM}}} \bibinfo{volume}{52},
  \bibinfo{number}{1} (\bibinfo{year}{2005}), \bibinfo{pages}{1--24}.
\newblock


\bibitem[\protect\citeauthoryear{Traud, Mucha, and Porter}{Traud
  et~al\mbox{.}}{2012}]%
        {traud2012social}
\bibfield{author}{\bibinfo{person}{Amanda~L Traud}, \bibinfo{person}{Peter~J
  Mucha}, {and} \bibinfo{person}{Mason~A Porter}.}
  \bibinfo{year}{2012}\natexlab{}.
\newblock \showarticletitle{Social structure of facebook networks}.
\newblock \bibinfo{journal}{\emph{sma}} \bibinfo{volume}{391},
  \bibinfo{number}{16} (\bibinfo{year}{2012}), \bibinfo{pages}{4165--4180}.
\newblock


\bibitem[\protect\citeauthoryear{Tretyakov, Armas-Cervantes, Garc{\'\i}a-Ba{
  \~n}uelos, Vilo, and Dumas}{Tretyakov et~al\mbox{.}}{2011}]%
        {tretyakov2011fast}
\bibfield{author}{\bibinfo{person}{Konstantin Tretyakov}, \bibinfo{person}{Abel
  Armas-Cervantes}, \bibinfo{person}{Luciano Garc{\'\i}a-Ba{ \~n}uelos},
  \bibinfo{person}{Jaak Vilo}, {and} \bibinfo{person}{Marlon Dumas}.}
  \bibinfo{year}{2011}\natexlab{}.
\newblock \showarticletitle{Fast fully dynamic landmark-based estimation of
  shortest path distances in very large graphs}. In
  \bibinfo{booktitle}{\emph{ACM International Conference on Information and
  Knowledge Management}}. \bibinfo{pages}{1785--1794}.
\newblock


\bibitem[\protect\citeauthoryear{Vieira, Fonseca, Damazio, Golgher, Reis, and
  Ribeiro-Neto}{Vieira et~al\mbox{.}}{2007}]%
        {vieira2007efficient}
\bibfield{author}{\bibinfo{person}{Monique~V Vieira}, \bibinfo{person}{Bruno~M
  Fonseca}, \bibinfo{person}{Rodrigo Damazio}, \bibinfo{person}{Paulo~B
  Golgher}, \bibinfo{person}{Davi de~Castro Reis}, {and}
  \bibinfo{person}{Berthier Ribeiro-Neto}.} \bibinfo{year}{2007}\natexlab{}.
\newblock \showarticletitle{Efficient search ranking in social networks}. In
  \bibinfo{booktitle}{\emph{ACM International Conference on Information and
  Knowledge Management}}. \bibinfo{pages}{563--572}.
\newblock


\bibitem[\protect\citeauthoryear{Wang}{Wang}{2024a}]%
        {wang_2024_13905461}
\bibfield{author}{\bibinfo{person}{Letong Wang}.}
  \bibinfo{year}{2024}\natexlab{a}.
\newblock \bibinfo{title}{{Parallel Cluster-BFS and Applications to Shortest
  Paths}}.
\newblock
\newblock
\urldef\tempurl%
\url{https://doi.org/10.5281/zenodo.13905461}
\showDOI{\tempurl}


\bibitem[\protect\citeauthoryear{Wang}{Wang}{2024b}]%
        {wang_2024_13909778}
\bibfield{author}{\bibinfo{person}{Letong Wang}.}
  \bibinfo{year}{2024}\natexlab{b}.
\newblock \bibinfo{booktitle}{\emph{Undirected scale-free graphs for
  cluster-BFS}}.
\newblock
\urldef\tempurl%
\url{https://doi.org/10.5281/zenodo.13909778}
\showDOI{\tempurl}


\bibitem[\protect\citeauthoryear{Yang and Leskovec}{Yang and Leskovec}{2015}]%
        {yang2015defining}
\bibfield{author}{\bibinfo{person}{Jaewon Yang} {and} \bibinfo{person}{Jure
  Leskovec}.} \bibinfo{year}{2015}\natexlab{}.
\newblock \showarticletitle{Defining and evaluating network communities based
  on ground-truth}.
\newblock \bibinfo{journal}{\emph{Knowledge and Information Systems}}
  \bibinfo{volume}{42}, \bibinfo{number}{1} (\bibinfo{year}{2015}),
  \bibinfo{pages}{181--213}.
\newblock


\bibitem[\protect\citeauthoryear{Zeng, Tong, and Chen}{Zeng
  et~al\mbox{.}}{2023}]%
        {zeng2023litehst}
\bibfield{author}{\bibinfo{person}{Yuxiang Zeng}, \bibinfo{person}{Yongxin
  Tong}, {and} \bibinfo{person}{Lei Chen}.} \bibinfo{year}{2023}\natexlab{}.
\newblock \showarticletitle{LiteHST: A Tree Embedding based Method for
  Similarity Search}.
\newblock \bibinfo{journal}{\emph{ACM SIGMOD International Conference on
  Management of Data (SIGMOD)}} \bibinfo{volume}{1}, \bibinfo{number}{1}
  (\bibinfo{year}{2023}), \bibinfo{pages}{1--26}.
\newblock


\bibitem[\protect\citeauthoryear{Zhao and Zheng}{Zhao and Zheng}{2010}]%
        {zhao2010orion}
\bibfield{author}{\bibinfo{person}{Xiaohan Zhao} {and} \bibinfo{person}{Haitao
  Zheng}.} \bibinfo{year}{2010}\natexlab{}.
\newblock \showarticletitle{Orion: shortest path estimation for large social
  graphs}. In \bibinfo{booktitle}{\emph{Workshop on Online social networks}}.
\newblock


\end{thebibliography}
